\documentclass{IEEEtran}

\usepackage{amsfonts}
\usepackage{mathrsfs}
\usepackage{amssymb,amsmath}
\usepackage[noblocks]{authblk}
\usepackage[compress]{cite}
\usepackage{bm}
\usepackage{algorithm}
\usepackage{algorithmic}
\usepackage{amsmath,amssymb,epsfig,graphics,subfigure}
\usepackage{theorem}
\usepackage{array,color}
\usepackage[compress]{cite}
\usepackage{algorithm}
\usepackage{algorithmic}
\usepackage{stfloats}

\newtheorem{property}{Property}
\newtheorem{highlight}{Highlight}
\newtheorem{remark}{Remark}
\newtheorem{theorem}{Theorem}
\newtheorem{lemma}{Lemma}
\newtheorem{conclusion}{Conclusion}

\newtheorem{proof}{Proof}

\begin{document}

\title{Unified Framework of KKT Conditions Based Matrix Optimizations for MIMO Communications}

\author{Chengwen Xing, \textsl{Member}, \textsl{IEEE}, Yindi Jing, \textsl{Member},
 \textsl{IEEE}, Shuai Wang, \textsl{Member}, \textsl{IEEE}, Jiaheng Wang, \textsl{Member}, \textsl{IEEE},
 Soon Xin Ng, \textsl{Senior Member}, \textsl{IEEE},  Sheng Chen, \textsl{Fellow}, \textsl{IEEE},
 Lajos Hanzo, \textsl{Fellow}, \textsl{IEEE} \\
 \textsl{Overview Article}
\thanks{C.~Xing and S.~Wang, are with School of Information and Electronics, Beijing
 Institute of Technology, Beijing 100081, China (E-mails: chengwenxing@ieee.org,
 swang@bit.edu.cn).} %
\thanks{Y.~Jing is with the Department of Electrical and Computer Engineering,
 University of Alberta, Edmonton, AB T6G 2V4, Canada (E-mail: yindi@ualberta.ca).} %
\thanks{J.~Wang is with the National Mobile Communications Research Laboratory,
 Southeast University, Nanjing, China (E-mail: jhwang@seu.edu.cn).} %
\thanks{S.~X.~Ng, S.~Chen and L.~Hanzo are with School of Electronics and Computer Science,
 University of Southampton, U.K. (E-mails: sxn@ecs.soton.ac.uk, sqc@ecs.soton.ac.uk, lh@ecs.soton.ac.uk).
 S. Chen is also with King Abdulaziz University, Jeddah, Saudi Arabia.} %
\vspace*{-5mm}
}
\maketitle

\begin{abstract}
 For multi-input multi-output (MIMO) communication systems, many transceiver design
 problems involve the optimization of the covariance matrices of the transmitted
 signals. The  derivation of the optimal solutions based on Karush-Kuhn-Tucker (KKT)
 conditions is a most popular method, and many results have been reported for
 different scenarios of MIMO systems. In this overview paper, we propose a unified
 framework in formulating the KKT conditions for general MIMO systems. Based on this
 framework, the optimal water-filling structures of the transmission covariance
 matrices are derived rigorously, which are applicable to a wide range of MIMO
 systems. Our results show that for seemingly different MIMO systems with various
 power constraints and objective functions, the derivations and water-filling
 structures for the optimal covariance matrix solutions are fundamentally the same.
 Thus, our unified framework and solution reveal the underlying relationships among
 the different water-filling structures of the covariance matrices. Furthermore, our
 results provide new solutions to the covariance matrix optimization of many
 complicated MIMO systems with multiple users and imperfect channel state information
 which were unknown before.
\end{abstract}


\section{Introduction}\label{S1}

 Our world is entering an amazing era of wireless technologies with many novel and
 revolutionary concepts coming forth, such as cloud computing, smart city, and green
 communications
 \cite{Larsson03,Baligh2014,Barbarossa2014,Wubben2014,Banelli2014,Razavizadeh2014,Rusek2013,zhongshan01,zhongshan02}.
 These unprecedented concepts demonstrate the demands and the desires for innovation
 in the wireless technology developments. In order to realize these goals, powerful
 physical layer technologies are expected, which are characterized by high power
 efficiency and high spectrum efficiency. Multi-input multi-output (MIMO) technology
 is widely seen as one of the most important ingredients for a variety of new
 wireless systems \cite{Weingarten2006,Telatar1999,JYangh1994,Sampth01,Scaglione2002,
Schizas07,Ding2007,Guan08,Medina07,Palomar03,YUWEI2004,PoliteWF_1,MatrixWaterFilling}.
 From the optimization viewpoint, many design problems for MIMO communication systems
 aim to optimize the covariance matrix of the transmitted signals
 \cite{Larsson03,Telatar1999,YUWEI2004}.

 Interestingly, for different performance metrics and different power constraints, the
 optimal transmission covariance matrices always have water-filling structures
 \cite{Waterfilling2004,Sampth01,Scaglione2002,JYangh1994,Schizas07,Ding2007,Guan08,
Medina07,Palomar03,YUWEI2004,PoliteWF_1,MatrixWaterFilling}.
 Such structures can greatly simplify the considered optimization problems and reveal
 the underlying physical meanings. Many variants of water-filling structures have
 been discovered, e.g., general water-filling \cite{GeneralizedWF}, polite water-filling
 \cite{PoliteWF_1}, cluster water-filling \cite{XingCL2013}, matrix-field water-filling
 \cite{MatrixWaterFilling}, and cave water-filling \cite{FFGaoTSP}. Along with the
 evolvements and developments of wireless communication systems, the number of the
 papers on MIMO optimization with water-filling structures in IEEE database is surging.
 Up to date, a large volume of elegant results have been published for various MIMO
 scenarios, including single-user (SU) MIMO systems \cite{Palomar03}, multi-user (MU)
 MIMO systems \cite{YUWEI2007}, distributed MIMO networks \cite{PoliteWF_2,Schizas07},
 and multi-hop amplify-and-forward MIMO relaying networks \cite{XingTSP2013,XingTSP201501}.
 Moreover, it has been shown that even for some transceiver optimizations with channel
 state information (CSI) errors, water-filling structures also hold
 \cite{XingTSP2013,Ding09,Ding10}. It is highly desired to reconsider the large volume
 of closely related existing works, to reveal the underlying fundamental connections of
 them, and to propose a unified framework.

 A natural question that arises is how these water-filling structures are derived.
 Generally speaking, there are three kinds of methods. The first one is the
 Karush-Kuhn-Tucker (KKT)-condition-based method \cite{Sampth01,Ding09,Ding10}. Under
 some mild conditions, KKT conditions are necessary conditions for the optimal solutions
 \cite{Boyd04}. Based on this fact, common properties derived from KKT conditions are
 the properties of the optimal solutions. The second one is based on matrix inequalities
 \cite{Guan08,Ding2007} or majorization theory \cite{Palomar03,XingTSP2013}. Compared
 with the first kind of method, majorization theory is less popular from the traditional
 communication theoretical perspective. Roughly speaking, it can reveal the inequality
 relationships between the diagonal elements and eigenvalues of a matrix \cite{Marshall79}.
 The third one is called matrix-monotonic optimization which exploits the monotonicity
 in the field of positive semi-definite matrices \cite{XingTSP201501,XingTSP201601}.
 The matrix-monotonic optimization framework has very strict restrictions on the objective
 functions and constraints of the optimization problems \cite{XingTSP201501,XingTSP201601}.
 Matrix inequality based methods are usually powerful and rigorous. When the objective
 functions are Schur-convex or Schur-concave functions of the diagonal elements of the
 mean-squared-error (MSE) matrix, the optimal structure of the matrix variables can be
 derived \cite{Palomar03,XingTSP2013}. Unfortunately, this kind of methods also suffer
 from many strict limitations in their applications \cite{Xing2016}.

 Generally speaking, the second and third kinds of methods are only applicable in
 single-source and single-destination systems. By contrast, the KKT-condition-based
 method is the most widely used and has a straightforward logic. True, KKT conditions
 are only necessary conditions for optimality. If the underlying problem is convex,
 KKT conditions are also sufficient for optimality. Even for nonconvex problems, they
 can still be very useful. Based on complex matrix derivatives, the KKT conditions of
 the optimization problems are derived first, from which the structures of the optimal
 solutions can be obtained. It is not limited by the format of the covariance matrix.
 e.g., SU or MU, nor by the format of the transmit power constraints. Thus this kind
 of methods has the widest range of applications.

 In the existing literature, the detailed derivation procedures for water-filling
 solutions from KKT conditions can seem to be largely different with each other
 \cite{Mai2011,Sampth01,Ding09}. In many works, it was argued that the derivation
 procedures are distinctive due to their specific system models and optimization
 problems, making the respective works significantly different from others. In our
 opinion, the theories and technologies for physical layer designs have some common
 underlying fundamentals. Due to the common root of all the KKT-condition-based
 methods, we believe that many seemingly different mathematical derivations can be
 unified into a common framework when the underlying fundamental nature is
 understood. This is the motivation of this work. In this paper, we investigate the
 most widely used KKT-condition-based method for the transmission covariance matrix
 optimization of MIMO systems. A unified framework is proposed for the derivations
 of the water-filling structure. The proposed unified framework provides general
 modeling, formulation, and methodology for MIMO transmission covariance optimization
 that can lead to new discoveries and solve new problems.

 Regarding the unified framework for the water-filling structured transmission
 covariance matrix based on KKT conditions, the main contribution of our work is
 two-fold. First, a fundamental and general solution is given based on which
 water-filling structures can be derived for a wide range of MIMO systems. We discover
 that the water-filling structures of these seemingly different MIMO systems are in
 nature closely related with each other. Second, based on the proposed solutions, we
 derive the water-filling structured transmission covariance matrices for various SU
 MIMO and MU MIMO systems under different objective functions, power constraints, and
 CSI assumptions, which include:
\begin{itemize}
\item the capacity maximization for SU MIMO systems under weighted power constraint
 and perfect CSI;
\item the MSE minimization for SU MIMO systems under weighted power constraint and
 perfect CSI;
\item the capacity maximization for SU MIMO systems under multiple weighted power
 constraints and perfect CSI;
\item the MSE minimization for SU MIMO systems under multiple weighted power constraints
 and perfect CSI;
\item the capacity maximization for SU MIMO systems with total power constraint and
 imperfect CSI;
\item the MSE minimization for SU MIMO systems under weighted power constraint and
 imperfect CSI;
\item the capacity maximization for SU MIMO systems under multiple weighted power
 constraints and CSI error;
\item the MSE minimization for SU MIMO systems under multiple weighted power constraints
 and imperfect CSI;
\item the capacity maximization for MU MIMO uplink under multiple weighted power
 constraints and perfect CSI;
\item the capacity maximization for MU MIMO uplink under multiple weighted power
 constraints and imperfect CSI;
\item the capacity maximization for MU MIMO downlink under multiple weighted power
 constraints and perfect CSI;
\item the capacity maximization for MU MIMO downlink under multiple weighted power
 constraints and imperfect CSI error;
\item the capacity maximization for network MIMO systems under multiple weighted
 power constraints and imperfect CSI.
\end{itemize}
 Many of these proposed solutions are unknown in the existing works. For example,
 to our best knowledge, the water-filling solutions for MIMO systems under
 multiple weighted power constraints and imperfect CSI are firstly revealed in
 this work.

 Throughout this paper, the following notation and symbol conventions are adopted.
 $(\cdot )^*$, $(\cdot )^{\rm{T}}$ and $(\cdot )^{\rm{H}}$ stand for the conjugate,
 transpose and Hermitian transpose operators, respectively, while ${\rm Tr}(\bm{Z})$
 and $|\bm{Z}|$ are the trace and determinant of complex matrix $\bm{Z}$,
 respectively. A complex matrix $\bm{Z}$ is represented by $\bm{Z}=\bm{Z}_{\rm R}
 +\textsf{j}\bm{Z}_{\rm I}$, where $\textsf{j}=\sqrt{-1}$, while $\bm{Z}_{\rm R}$
 and $\bm{Z}_{\rm I}$ are the real and imaginary parts of $\bm{Z}$. For matrix
 $\bm{Z}$ with rank $\text{Rank}(\bm{Z})=N$, $[\bm{Z}]_{:,1:N}$ denotes the first
 $N$ columns of $\bm{Z}$ and the sub-matrix $[\bm{Z}]_{1:N,1:N}$ consists of the
 first $N$ rows and the first $N$ columns of $\bm{Z}$, while $[\bm{Z}]_{i,i}$
 denotes the $i$th diagonal element of $\bm{Z}$. $\mathbb{E}\{\cdot\}$ denotes the
 expectation operation, $\bm{A}^{\frac{1}{2}}$ is the Hermitian square root of the
 positive semi-definite matrix $\bm{A}$, and $(a)^{+}=\max \{ 0,a\}$, while
 $\lambda_i(\bm{B})$ is the $i$th largest eigenvalue of matrix $\bm{B}$, and
 $\sigma_i(\bm{Z})$ is the $i$th largest singular value of matrix $\bm{Z}$. For
 two Hermitian matrices $\bm{A}$ and $\bm{B}$, $\bm{A}\succeq \bm{B}$ means that
 $\bm{A}-\bm{B}$ is positive semi-definite. The identity matrix of appropriate
 dimension is denoted by $\bm{I}$, and $\bm{0}$ is the zero matrix of appropriate
 dimension. In matrix decomposition, to clarify the order of the eigenvalues or
 singular values, we use $\bm{\Lambda} \searrow $ to represent a rectangular
 diagonal matrix with the diagonal elements arranged in decreasing order.

\section{Fundamental Results on Matrix Optimization}\label{S2}

\subsection{Premises}\label{S2.1}

 For a convex optimization problem, the KKT conditions are both necessary and
 sufficient for the optimal solution. When the considered problem is nonconvex,
 the KKT conditions are only necessary but not sufficient for the optimal
 solutions. However, even in this case, it is worth highlighting that KKT
 conditions can still be very useful. This is because if a property or structure
 is derived from KKT conditions, then all the solutions satisfying KKT
 conditions have this structure. This is the major motivation of our framework
 of providing a fundamental result for the structure of the optimal transmission
 covariance matrix from KKT conditions. We summarize this useful result in the
 following lemma.

\begin{lemma}\label{L1}
 If all the solutions satisfying the KKT conditions own the same structure,
 this structure is definitely satisfied by the optimal solutions of the
 associated optimization problem.
\end{lemma}

 For the KKT conditions of the optimization problem with complex matrix
 variables, complex matrix derivatives are the most fundamental tool. In this
 paper, we focus on positive semi-definite matrix. In other words, there exit
 some constraints on the complex matrix variables. We first present some
 important points for the complex matrix derivatives of positive semi-definite
 matrix. These highlights form the basis for the following results.

\begin{highlight}\label{H1}
 The objective function of the optimization problem with complex matrix variables
 must be real valued.
\end{highlight}

 It is worth noting that the role of complex matrix derivatives is to find
 extreme values. From the mathematical viewpoint, it is meaningless to argue a
 complex number is larger or smaller. Therefore, performing complex derivatives
 over a complex valued function is totally meaningless. It is true that in some
 textbooks some matrix derivatives are given such as \cite{Kay1993}
\begin{align}\label{eq1}
 \frac{\partial {\rm{Tr}}(\bm{W}\bm{X})}{\partial \bm{X}} =& \bm{W}^{\rm{T}} ,
\end{align}
 where $\bm{X}$ is the complex matrix variable and $\bm{W}$ is a complex matrix
 with proper dimensions. It is worth noting that ${\rm{Tr}}(\bm{W}\bm{X})$ can
 be a complex value. This operation is an intermediate operation. In other words,
 in real-world applications, ${\rm{Tr}}(\bm{W}\bm{X})$ and ${\rm{Tr}}(\bm{W}^{\rm{H}}
 \bm{X}^{\rm{H}})$ usually appear in company. Thus, it is more meaningful to
 define the following complex matrix derivative instead of the previous one
\begin{align}\label{eq2}
 \frac{\partial \big({\rm{Tr}}(\bm{W}\bm{X})+{\rm{Tr}}\big(\bm{W}^{\rm{H}}
  \bm{X}^{\rm{H}}\big)\big)}{\partial \bm{X}} =& \bm{W}^{\rm{T}} .
\end{align}

\begin{highlight}\label{H2}
 Complex matrix derivative is a special real matrix derivative.
\end{highlight}

 It is well-known that real number is a special case of complex number whose
 imaginary part is zero. However, the story is different for complex matrix
 derivatives. Complex matrix derivatives are only defined for some special
 functions. The objective function is just a function of the real and imaginary
 parts of the complex matrix variables. After performing matrix derivatives,
 the real and imaginary parts of the resulting matrix corresponds to the
 derivatives of the real and imaginary parts, separately. In a nutshell, a
 complex matrix derivative is a compact version of the real matrix derivatives
 with respect to both real and imaginary parts. This may be the reason why the
 existing textbooks on complex matrix derivatives are usually written by
 engineering people instead of mathematicians.

 Generally speaking, complex matrix derivatives are defined in the following
 two forms \cite{Kay1993,HJORUNGNES}
\begin{align}
 \frac{\partial f(\bm{X})}{\partial \bm{X}} =& \frac{1}{2}\left(
  \frac{\partial f(\bm{X})}{\partial \bm{X}_{\rm R}} - \textsf{j}
  \frac{\partial f(\bm{X})}{\partial \bm{X}_{\rm I}} \right) , \label{eq3} \\
 \frac{\partial f(\bm{X})}{\partial \bm{X}^*} =& \frac{1}{2}\left(
  \frac{\partial f(\bm{X})}{\partial \bm{X}_{\rm R}} + \textsf{j}
  \frac{\partial f(\bm{X})}{\partial \bm{X}_{\rm I}} \right) . \label{eq4}
\end{align}

\begin{highlight}\label{H3}
 For complex matrix variable $\bm{X}$, complex matrix derivative can be
 performed with respect to $\bm{X}$ itself or with respect to its conjugate
 $\bm{X}^*$. For these two matrix derivatives, the resulting KKT conditions
 are exactly the same.
\end{highlight}

 It should be highlighted that for some intermediate steps, the complex matrix
 derivations with respect to $\bm{X}$ itself and with respect to its conjugate
 $\bm{X}^*$ are different. For example, the following equalities are given in
 \cite{HJORUNGNES}
\begin{align}\label{eq5}
 \frac{\partial {\rm{Tr}}\big(\bm{X}^{\rm{H}}\big)}{\partial \bm{X}} = \bm{0} , ~
 \frac{\partial {\rm{Tr}}\big(\bm{X}^{\rm{H}}\big)}{\partial \bm{X}^*} = \bm{I} .
\end{align}
 However, no matter which complex matrix derivation is used, the resulting KKT
 conditions must be the same.

\begin{highlight}\label{H4}
 In the traditional definition of complex matrix derivatives, there is no
 restriction imposed on the structure of the matrix variable. The elements
 of the matrix variable are independent variables.
\end{highlight}

 In our MIMO optimization problems, the covariance matrix of the transmitted signals
 $\bm{Q}$ is the complex matrix variable, which is positive semi-definite. Different
 from a general unconstrained complex matrix variable, there are two constraints
 imposed on $\bm{Q}$: 1)~It is conjugate symmetric, and 2)~Its eigenvalues are all
 nonnegative. Based on the definition of complex matrix derivation \cite{Kay1993},
 when the objective function is a scalar real valued function, after performing a
 matrix derivative, the result is still a Hermitian matrix with the same dimension.
 Based on this fact we have
\begin{align}\label{eq6}
 \left(\frac{\partial f(\bm{Q})}{\partial \bm{Q}}\right)^{\rm T} =&
       \frac{\partial f(\bm{Q})}{\partial \bm{Q}^*} .
\end{align}
 For example, based on the complex matrix derivative definition and the fact that
 $\bm{Q}$ is a positive semi-definite matrix, we have
\begin{align}\label{eq7}
 \frac{\partial {\rm Tr}(\bm{W}\bm{Q})}{\partial \bm{Q}^*} =&
  \frac{\partial {\rm Tr}\big(\bm{W}\bm{Q}^{\rm H}\big)}{\partial \bm{Q}^*} = \bm{W} .
\end{align}
 Here the matrix $\bm{W}$ must be a Hermitian matrix. Otherwise, this derivation
 definition is meaningless.

 Moreover, as $\bm{Q}$ is a positive semi-definite matrix, the complex matrix derivative
 must be defined on the set of positive semi-definite matrices. Unfortunately, the
 existing definitions for complex matrix derivatives never consider this. In other
 words, there is a relaxation when performing complex matrix derivatives, and if the
 final solution is positive semi-definite, then there is no loss. This is usually
 guaranteed by the mathematical formula of the objective function.

\begin{highlight}\label{H5}
 After performing complex matrix derivative on a real valued function with respect to
 positive semi-definite matrix, the result must be a Hermitian matrix. The semi-positivity
 is guaranteed by the fact that for the objective function, the optimal values occur at
 the set of variables that are positive semi-definite matrices.
\end{highlight}

 We use an example to show that when we do not consider these highlights we may
 arrive at a wrong conclusion. The differential of $\log\big|\bm{I}+\bm{Q}\bm{H}^{\rm H}
 \bm{H}\big|$ could be defined as
\begin{align}\label{eq8}
 \hspace*{-2mm}{\rm d}\! \left(\log\big|\bm{I}\! +\! \bm{Q}\bm{H}^{\rm H}\bm{H}\big|\right)\!\! =&
  {\rm Tr}\big(\! \big(\bm{I}\! +\! \bm{Q}\bm{H}^{\rm H}\bm{H}\big)^{-1}
  \! {\rm d}\big(\bm{Q}\big)\bm{H}^{\rm H}\bm{H}\big) \! ,\!
\end{align}
 based on which we would have
\begin{align}\label{eq9}
 \frac{\partial \log\big|\bm{I}+\bm{Q}\bm{H}^{\rm H}\bm{H}\big|}{\partial \bm{Q}^*} =&
  \bm{H}^{\rm H}\bm{H}\big(\bm{I}+\bm{Q}\bm{H}^{\rm H}\bm{H}\big)^{-1} .
\end{align}
 This however is wrong because in the above equation in addition to $\bm{Q}$ being a
 Hermitian matrix, there exit some other constraints to ensure that the resultant matrix
 is a Hermitian matrix, e.g., the nonzero eigenvalues of $\bm{Q}$ and $\bm{H}^{\rm H}\bm{H}$
 having the same eigenvectors (for the weighted power constraints, the latter should be
 removed as it cannot be satisfied). The correct logic is because $\log\big|\bm{I}+
 \bm{Q}\bm{H}^{\rm H}\bm{H}\big|=\log\big|\bm{I}+\bm{H}\bm{Q}\bm{H}^{\rm H}\big|$
 and the differential of $\log\big|\bm{I}+\bm{H}\bm{Q}\bm{H}^{\rm H}\big|$ is
\begin{align}\label{eq10}
 {\rm d}\left(\log\big|\bm{I}+\bm{H}\bm{Q}\bm{H}^{\rm H}\big|\right) =& {\rm Tr}\big(
  \big(\bm{I}+\bm{H}\bm{Q}\bm{H}^{\rm H}\big)^{-1}\bm{H}{\rm d}\big(\bm{Q}\big)\bm{H}^{\rm H}\big) ,
\end{align}
 based on which we have
\begin{align}\label{eq11}
 \frac{\partial \log\big|\bm{I}+\bm{Q}\bm{H}^{\rm H}\bm{H}\big|}{\partial \bm{Q}^*} =&
  \bm{H}^{\rm H}\big(\bm{I}+\bm{H}\bm{Q}\bm{H}^{\rm H}\big)^{-1}\bm{H} .
\end{align}
 In this formulation, it can be concluded that if $\bm{Q}$ is Hermitian, the resulting
 matrix is definitely Hermitian.

\subsection{The General Conclusion}\label{S2.2}

 The fundamental mathematical result of the transmission covariance matrix optimization
 for MIMO systems is given in the following theorem.

\begin{theorem}\label{T1}
 Define a set of matrix equations as follows
\begin{equation}\label{two_equation} 
\left\{ \begin{array}{l}
 \hspace{-2mm}\bm{H}^{\rm H}\bm{\Pi}^{-\frac{1}{2}} \big(\bm{I}\! +\! \bm{\Pi}^{-\frac{1}{2}}\bm{H}
  \bm{Q}\bm{H}^{\rm H}\bm{\Pi}^{-\frac{1}{2}}\big)^{-K}\bm{\Pi}^{-\frac{1}{2}}\bm{H}
  = \mu\bm{\Phi}\! -\! \bm{\Psi} \!,  \\
 \hspace{-2mm}\bm{Q}^{\frac{1}{2}}\bm{\Psi}\bm{Q}^{\frac{1}{2}} = \bm{0} \text{ or }
  {\rm Tr}(\bm{\Psi}\bm{Q}) = 0 ,
\end{array} \right.
\end{equation}
 where $\bm{H}$ is a complex matrix with proper dimension, $K$ can be any positive integer,
 $\mu$ is a positive real scalar, $\bm{\Pi}$ and $\bm{\Phi}$ are positive definite matrices,
 while $\bm{Q}$ and $\bm{\Psi}$ are positive semidefinite matrices. Then based on the
 following singular value decomposition (SVD)
\begin{align}\label{theorem_1_svd} 
 \bm{\Pi}^{-\frac{1}{2}}\bm{H}\bm{\Phi}^{-\frac{1}{2}} =& \bm{U}_{\bm{\mathcal{H}}}
  \bm{\Lambda}_{\bm{\mathcal{H}}}\bm{V}_{\bm{\mathcal{H}}}^{\rm H} \text{ with }
  \bm{\Lambda}_{\bm{\mathcal{H}}} \searrow ,
\end{align}
 the matrix $\bm{Q}$ satisfying the two equations in (\ref{two_equation}) has the following
 water-filling structure
\begin{align}\label{eq14}
 \bm{Q} =& \bm{\Phi}^{-\frac{1}{2}}\big[\bm{V}_{\bm{\mathcal{H}}}\big]_{:,1:N}
  \Big( \mu^{-\frac{1}{K}}\big[\bm{\Lambda}_{\bm{\mathcal{H}}}\big]_{1:N,1:N}^{\frac{2}{K}-2}
  - \big[\bm{\Lambda}_{\bm{\mathcal{H}}}\big]_{1:N,1:N}^{-2}\Big)^{+} \nonumber \\
 & \times \big[\bm{V}_{\bm{\mathcal{H}}}\big]_{:,1:N}^{\rm H} \bm{\Phi}^{-\frac{1}{2}} ,
\end{align}
 where $N={\rm Rank}(\bm{H})$, and $(\bm{B})^+$ means applying the operation $(\cdot )^+$
 to every element of $\bm{B}$.
\end{theorem}

\begin{proof}
 By left and right multiplying the both sides of the first equation in (\ref{two_equation})
 with $\bm{Q}^{\frac{1}{2}}$ and noting the second equation in (\ref{two_equation}), we have
\begin{align}\label{KKT_1_R} 
 & \bm{Q}^{\frac{1}{2}}\bm{H}^{\rm H}\bm{\Pi}^{-\frac{1}{2}} \big(\bm{I} + \bm{\Pi}^{-\frac{1}{2}}\bm{H}
  \bm{Q}\bm{H}^{\rm H}\bm{\Pi}^{-\frac{1}{2}}\big)^{-K}\bm{\Pi}^{-\frac{1}{2}}\bm{H} \bm{Q}^{\frac{1}{2}}
 \nonumber \\
 & \hspace*{10mm} = \mu \bm{Q}^{\frac{1}{2}}\bm{\Phi}\bm{Q}^{\frac{1}{2}} .
\end{align}
 By defining the new matrix
\begin{align}\label{eq16}
 \bm{A} =& \bm{\Phi}^{\frac{1}{2}}\bm{Q}^{\frac{1}{2}} ,
\end{align}
 (\ref{KKT_1_R}) can be rewritten as
\begin{align}\label{KKT_1_F} 
 & \bm{A}^{\rm H}\bm{\Phi}^{-\frac{1}{2}}\bm{H}^{\rm H}\bm{\Pi}^{-\frac{1}{2}}\big( \bm{I}
  \! +\! \bm{\Pi}^{-\frac{1}{2}}\bm{H}\bm{\Phi}^{-\frac{1}{2}}\bm{A}\bm{A}^{\rm H}
  \bm{\Phi}^{-\frac{1}{2}}\bm{H}^{\rm H}\bm{\Pi}^{-\frac{1}{2}}\big)^{-\! K} \nonumber \\
 &\hspace*{10mm}\times \bm{\Pi}^{-\frac{1}{2}}\bm{H}\bm{\Phi}^{-\frac{1}{2}}\bm{A}
  = \mu \bm{A}^{\rm H}\bm{A} .
\end{align}
 From the SVDs of $\bm{A}$ and $\bm{\Pi}^{-\frac{1}{2}}\bm{H}\bm{\Phi}^{-\frac{1}{2}}\bm{A}$,
 (\ref{KKT_1_F}) implies that $\bm{A}$ and $\bm{\Pi}^{-\frac{1}{2}}\bm{H}\bm{\Phi}^{-\frac{1}{2}}
 \bm{A}$ have the same right SVD unitary matrix. Based on this fact, it can be concluded that
 $\bm{A}^{\rm H}\bm{A}$ and $\bm{A}^{\rm H}\bm{\Phi}^{-\frac{1}{2}}\bm{H}^{\rm H}
 \bm{\Pi}^{-1}\bm{H}\bm{\Phi}^{-\frac{1}{2}}\bm{A}$ have the same eigen-matrix. Therefore,
 the left SVD eigenvectors of $\bm{A}$ corresponding to its nonzero singular values are
 also the right SVD eigenvectors of $\bm{\Pi}^{-\frac{1}{2}}\bm{H}\bm{\Phi}^{-\frac{1}{2}}$.
 It is worth noting that for the zero singular values of $\bm{A}$, the corresponding right
 eigenvectors can be arbitrary as long as they are orthogonal to each other and orthogonal
 to the ones corresponding to the nonzero singular values. Therefore, without loss of
 optimality, we can assume the following property.

\begin{property}\label{P1}
 The left SVD unitary matrix of $\bm{A}$ is the right SVD unitary matrix of $\bm{\Pi}^{-\frac{1}{2}}
 \bm{H}\bm{\Phi}^{-\frac{1}{2}}$.
\end{property}

 Based on the definition of $\bm{A}$, the first equation in (\ref{two_equation}) is
 equivalent to the following one
\begin{align}\label{KKT_1_F_1} 
 & \bm{\Phi}^{-\frac{1}{2}}\bm{H}^{\rm H}\bm{\Pi}^{-\frac{1}{2}}\big(\bm{I} +
  \bm{\Pi}^{-\frac{1}{2}}\bm{H}\bm{\Phi}^{-\frac{1}{2}}\bm{A}\bm{A}^{\rm H}\bm{\Phi}^{-\frac{1}{2}}
  \bm{H}^{\rm H}\bm{\Pi}^{-\frac{1}{2}}\big)^{-K} \nonumber \\
 &\hspace*{10mm} \times \bm{\Pi}^{-\frac{1}{2}}\bm{H}\bm{\Phi}^{-\frac{1}{2}} =
  \mu \bm{I} - \bm{\Phi}^{-\frac{1}{2}}\bm{\Psi}\bm{\Phi}^{-\frac{1}{2}} .
\end{align}
 Based on Property~\ref{P1}, it can be seen that the right SVD unitary matrix of
 $\bm{\Pi}^{-\frac{1}{2}}\bm{H}\bm{\Phi}^{-\frac{1}{2}}$ is the unitary matrix for
 the eigenvalue decomposition (EVD) of the lefthand side of (\ref{KKT_1_F_1}).
 Because of the equality in (\ref{KKT_1_F_1}), the EVD unitary matrix of the lefthand
 side of (\ref{KKT_1_F_1}) is exactly the EVD unitary matrix of the righthand side of
 (\ref{KKT_1_F_1}). In other words, the following property holds.

\begin{property}\label{P2}
 The right SVD unitary matrix of $\bm{\Pi}^{-\frac{1}{2}}\bm{H}\bm{\Phi}^{-\frac{1}{2}}$
 is the EVD unitary matrix of $\bm{\Phi}^{-\frac{1}{2}}\bm{\Psi}\bm{\Phi}^{-\frac{1}{2}}$.
\end{property}

 With the following definitions
\begin{align}
 a_i =& \sigma_i(\bm{A}) , \label{eq19} \\
 h_i =& \sigma_i\big(\bm{\Pi}^{-\frac{1}{2}}\bm{H}\bm{\Phi}^{-\frac{1}{2}}\big) , \label{eq20} \\
 \psi_i=& \lambda_i\big(\bm{\Phi}^{-\frac{1}{2}}\bm{\Psi}\bm{\Phi}^{-\frac{1}{2}}\big) , \label{eq21}
\end{align}
 and together with Properties~\ref{P1} and \ref{P2}, (\ref{KKT_1_F_1}) becomes
\begin{align}\label{eq22}
 \frac{h_i^2}{\big(1+a_i^2 h_i^2\big)^K} =& \mu - \psi_i ,
\end{align}
 based on which $a_i^2$, $\forall i$, are derived as
\begin{align}\label{solution} 
 a_i^2 =& \frac{ \big(h_i^2\big)^{\frac{1}{K}} }{ h_i^2 } \left( \frac{1}{\big(\mu - \psi_i\big)^{\frac{1}{K}}}
  - \frac{1}{\big(h_i^2\big)^{\frac{1}{K}}} \right) .
\end{align}

 Moreover, because of $\bm{Q}^{\frac{1}{2}}\bm{\Psi}\bm{Q}^{\frac{1}{2}}=\bm{0}$ and
 together with the definition of $\bm{A}=\bm{\Phi}^{\frac{1}{2}}\bm{Q}^{\frac{1}{2}}$, we
 can conclude that $\bm{A}^{\rm H}\bm{\Phi}^{-\frac{1}{2}}\bm{\Psi}\bm{\Phi}^{-\frac{1}{2}}
 \bm{A}=\bm{0}$, which implies
\begin{align}\label{eq24}
 a_i^2 \psi_i =& 0 ,
\end{align}
 based on Properties~\ref{P1} and \ref{P2}. Therefore, $\psi_i$ in (\ref{solution}) can
 be removed. The reason is as follows. With the inequalities $\mu > 0$ and $\psi_i \ge 0$
 together with (\ref{eq24}), it can be seen that for (\ref{solution})
\begin{align}
 \psi_i = 0 & \text{ when } a_i^2 > 0 , \label{eq25} \\
 \psi_i > 0 & \text{ when } a_i^2 = 0 , \label{eq26}
\end{align}
 based on which we conclude that
\begin{align}
 \psi_i = 0 & \text{ when } a_i^2 = \frac{(h_i^2)^{\frac{1}{K}}}{h_i^2}
  \left(\frac{1}{\mu^{\frac{1}{K}}} - \frac{1}{(h_i^2)^{\frac{1}{K}}}\right) > 0 , \label{eq27} \\
 \psi_i > 0 & \text{ when } a_i^2 = \frac{(h_i^2)^{\frac{1}{K}}}{h_i^2}
  \! \left(\! \frac{1}{(\mu\! -\! \psi_i)^{\frac{1}{K}}}\! -\! \frac{1}{(h_i^2)^{\frac{1}{K}}}\! \right) = 0  \label{eq28} .
\end{align}
 By combining (\ref{eq27}) and (\ref{eq28}) together, we have
\begin{align}\label{eq29}
 a_i^2 =& \frac{(h_i^2)^{\frac{1}{K}}}{h_i^2} \left(\frac{1}{\mu^{\frac{1}{K}}} -
  \frac{1}{(h_i^2)^{\frac{1}{K}}}\right)^+ .
\end{align}
 Noting (\ref{eq19}) to (\ref{eq21}), we obtain (\ref{eq14}).
\end{proof}

\begin{remark}\label{R1}
 When applying Theorem~\ref{T1} to KKT-condition-based transmission covariance matrix
 optimization for MIMO systems, $\bm{Q}$ corresponds to the covariance matrix to be
 optimized and $\bm{H}$ is the channel matrix, while $\bm{\Phi}$ is the power
 weighting matrix, $\bm{\Pi}$ is the noise covariance matrix, and $\bm{\Psi}$
 corresponds to the Lagrange multipliers. It is worth emphasizing that Theorem~\ref{T1}
 is very general and independent of specific MIMO system setups, including the objective
 functions, power constraints, signal models, and channel assumptions. Moreover, in
 Theorem~\ref{T1}, the matrices $\bm{\Pi}$ and $\bm{\Phi}$ are not restricted to be
 constant but can be functions of $\bm{Q}$.
\end{remark}

\subsection{Differences with Existing Works}\label{S2.3}

 As water-filling structures have been extensively studied, we would like to discuss
 the main differences between our derivations and existing ones. To our best knowledge,
 the existing derivation methods can be classified into two categories.

\subsubsection{Comparison with the first existing category}

 The first existing approach is based on matrix inequalities, e.g.,
 \cite{Xing2016,Telatar1999,Palomar03}. However, it is usually very difficult to
 guarantee that the extreme values of the matrix inequalities can be achieved due to
 the variations in the objectives or constraints of the optimization problems
 \cite{Xing2016}. For example, in Telatar's paper \cite{Telatar1999}, the matrix
 inequality $\log\big|\bm{I}+\bm{H}\bm{Q}\bm{H}^{\rm H}\big|\le \sum_i
 \log\big(1+\lambda_i(\bm{H}^{\rm H}\bm{H})\lambda_i(\bm{Q})\big)$ is used to derive
 the water-filling structure of the optimal solutions. If the sum power constraint
 ${\rm Tr}(\bm{Q})\le P$ is replaced by $[\bm{Q}]_{i,i}\le P_i$, the equality cannot
 be achieved \cite{Xing2016}. In Section~\ref{SU_MIMO}, we will show that this
 problem can be overcome by using our method presented in Theorem~\ref{T1}, because
 in Theorem~\ref{T1}, $\bm{\Pi}$ and $\bm{\Phi}$ can be functions of $\bm{Q}$. Thus,
 compared with the existing works \cite{Xing2016,Telatar1999,Palomar03}, our
 conclusion is more general.

\subsubsection{Comparison with the second existing category}

 The second category of existing water-filling structure derivations is purely based on
 KKT conditions, which consists of two phases. In the first phase, the argument is made
 that when the product of two matrices $\bm{\Lambda}_1\bm{\Lambda}_2$ is a Hermitian
 matrix and $\bm{\Lambda}_1$ is a diagonal matrix, then $\bm{\Lambda}_2$ is also a
 diagonal matrix. However, if some diagonal elements of $\bm{\Lambda}_1$ are zeros, this
 claim does not hold. To avoid this difficulty, in some existing works, it is usually
 assumed that $\bm{H}$ is full rank \cite{Mai2011}. Different from these existing works,
 Theorem~\ref{T1} does not rely on this argument and we do no impose the full rank condition
 on $\bm{H}$.

 In the second phase of the existing KKT-condition-based derivation methods
 \cite{Sampth01,Ding09}, the diagonalizable structure is used to reduce the KKT
 conditions to some equations that involve only diagonal matrices, and then to solve
 the optimal covariance matrix from these equations. As the diagonal elements must be
 nonnegative, the operation $(a)^+=\max\{a,0\}$ is introduced. In Theorem~\ref{T1}, this
 operation appears in the solution via rigorous mathematical derivation, which has
 clear physical interpretation or insight in applications, e.g., the transmission power
 cannot be negative. By contrast, some existing derivations use the argument that
 $(a)^+=\max\{a,0\}$ comes from the fact that $\bm{Q}$ is positive semi-definite and
 therefore the negative eigenvalues must be set to zero. This argument may be incorrect
 in some applications. To see this, let us consider a positive semidefinite matrix
 $\bm{\Phi}$. Generally,
\begin{align}\label{eq30}
 & \bm{\Phi}^{-\frac{1}{2}}\big[\bm{V}_{\bm{\mathcal{H}}}\big]_{:,1:N}
  \Big( \mu^{-\frac{1}{K}}\big[\bm{\Lambda}_{\bm{\mathcal{H}}}\big]_{1:N,1:N}^{\frac{2}{K}-2}
  - \big[\bm{\Lambda}_{\bm{\mathcal{H}}}\big]_{1:N,1:N}^{-2}\Big)^{+} \nonumber \\
 & \hspace*{20mm}\times \big[\bm{V}_{\bm{\mathcal{H}}}\big]_{:,1:N}^{\rm H} \bm{\Phi}^{-\frac{1}{2}} \nonumber \\
 & \neq \bigg(\bm{\Phi}^{-\frac{1}{2}}\big[\bm{V}_{\bm{\mathcal{H}}}\big]_{:,1:N}
  \Big( \mu^{-\frac{1}{K}}\big[\bm{\Lambda}_{\bm{\mathcal{H}}}\big]_{1:N,1:N}^{\frac{2}{K}-2}
  - \big[\bm{\Lambda}_{\bm{\mathcal{H}}}\big]_{1:N,1:N}^{-2}\Big) \nonumber \\
 & \hspace*{20mm}\times \big[\bm{V}_{\bm{\mathcal{H}}}\big]_{:,1:N}^{\rm H} \bm{\Phi}^{-\frac{1}{2}}\bigg)^+ .
\end{align}
 The equality holds only when $\bm{\Phi}\propto\bm{I}$.

 Moreover, in some existing works,  $\bm{Q}$ is replaced by $\bm{F}\bm{F}^{\rm H}$ in
 the derivation and $\bm{F}$ is a tall matrix instead of a square one. Then the
 following KKT condition is achieved
\begin{align}\label{KKT_F} 
 & \bm{H}^{\rm H}\bm{\Pi}^{-\frac{1}{2}}\big(\bm{I}\! +\! \bm{\Pi}^{-\frac{1}{2}}\bm{H}\bm{F}
  \bm{F}^{\rm H}\bm{H}^{\rm H}\bm{\Pi}^{-\frac{1}{2}}\big)^{-1}\bm{\Pi}^{-\frac{1}{2}}
  \bm{H}\bm{F} = \mu \bm{\Phi}\bm{F} .
\end{align}
 In this case, as revealed in \cite{Xing2016}, due to the turning-off effect, the
 water-filling structure cannot be achieved. To clarify this, a brief discussion is given
 as follows. First, (\ref{KKT_F}) is not equivalent to the following equation
\begin{align}\label{KKT_F_A} 
 & \hspace*{-2mm}\bm{H}^{\rm H}\bm{\Pi}^{-\frac{1}{2}}\big(\bm{I}\! +\! \bm{\Pi}^{-\frac{1}{2}}\bm{H}
  \bm{F}\bm{F}^{\rm H}\bm{H}^{\rm H}\bm{\Pi}^{-\frac{1}{2}}\big)^{-1}\bm{\Pi}^{-\frac{1}{2}}
 \bm{H}\! =\! \mu \bm{\Phi} , \!
\end{align}
 as the right inverse of $\bm{F}$  may not exist. Also (\ref{KKT_F}) is clearly not
 equivalent to (\ref{two_equation}). As discussed in \cite{Xing2016}, any eigen-channel
 can be turned off (allocated zero power) and (\ref{KKT_F}) can still be satisfied. This
 fact is referred to as the `turning-off effect', and because of it, Theorem~\ref{T1}
 cannot be achieved based on (\ref{KKT_F}).

\begin{figure*}[!b]
\vspace*{-3mm}
\begin{center}
\includegraphics[width=0.95\textwidth]{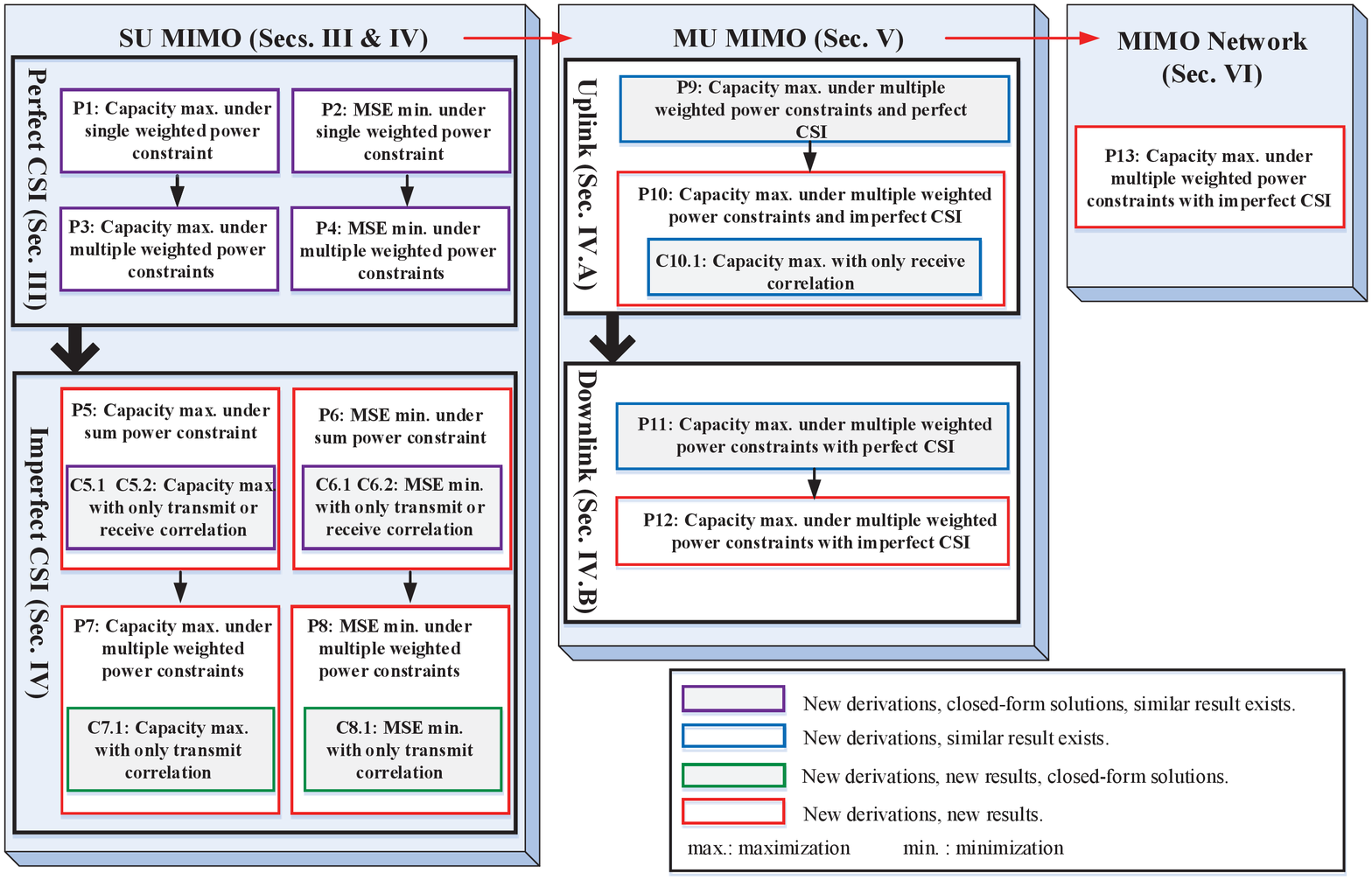}
\end{center}
\vspace*{-5mm}
\caption{Summary diagram of the MIMO design problems studied in this paper.}
\label{Fig_1}
\vspace*{-1mm}
\end{figure*}

\subsubsection{Summary of our derivation}

 Based on the KKT conditions, which are the necessary conditions of optimality for the
 corresponding optimization problem, Theorem~\ref{T1} reveals that all the solutions
 of $\bm{Q}$ satisfying the two mathematical equations in (\ref{two_equation}) have
 the structure given by (\ref{eq14}). The fundamental conclusion in Theorem~\ref{T1}
 has the following properties:
\begin{itemize}
\item It does not depend on the considered optimization problem being convex.
\item It is related to the two equations in (\ref{two_equation}), which compose
 a partial set of the KKT conditions, instead of the full set.
\item It is applicable when the involved parameters are constants or functions of the
 optimization variable.
\item It provides a common structure of the solutions that satisfy the KKT conditions.
\end{itemize}

 In the following four sections, we use Theorem~\ref{T1} to derive the optimal
 water-filling structures for the transmission covariance matrices of MIMO systems
 with different objective functions, power constraints, and CSI assumptions. A
 diagram is provided in Fig.~\ref{Fig_1} to show the organization of these different
 MIMO scenarios, their interconnections, and our new contributions. Generally, we
 investigate three kinds of MIMO systems: SU MIMO systems, MU MIMO systems, and
 network MIMO systems. More specifically, for SU MIMO systems, we investigate
 the capacity maximization and MSE minimization with multiple weighted power
 constraints and under perfect CSI. In addition, we find the optimal structures for
 the corresponding optimization problems under the imperfect CSI with Kronecker
 structured channel errors. For MU MIMO systems, both uplink and downlink are
 considered for capacity maximization with multiple weighted power constraints under
 both perfect CSI and imperfect CSI. Based on Theorem~\ref{T1}, the optimal structures
 of the transmission covariance matrices are derived. Finally, general MIMO networks
 are investigated, in which multiple nodes communicate with multiple destinations
 with arbitrary network topologies and under imperfect CSI. From Theorem~\ref{T1},
 the structures of the optimal transmission covariance matrices are derived, which
 can lead to low-complexity suboptimal solutions that can work as a benchmark
 algorithm for many complicated MIMO systems.

\begin{remark}\label{R2}
 Taking the transmission covariance matrices as optimization variables is a widely
 adopted approach for MIMO optimizations \cite{Jafar2005,Goldsmith2003a,Goldsmith2003b}.
 For MSE minimization, when a linear precoding matrix is adopted at the source, there
 may exist a constraint on the rank of the covariance matrix when the number of data
 streams is smaller than that of the transmit antennas. Unfortunately, this constraint
 is nonconvex, and rank-relaxation is often needed in solving the associated
 optimization. On the other hand, for capacity maximization, there is no rank
 constraint. This is because for MIMO capacity maximization, the number of data streams
 should also be optimized and is not given a priori. Moreover, from the information
 theory viewpoint, orthogonality between different data steams can also be realized
 via certain coding strategies in addition to the spatial multiplexing provided by
 the antennas \cite{Goldsmith2003a}. In such cases, the rank constraints imposed by
 antenna arrays do not exist.
\end{remark}

\section{Transmission Covariance Matrix Optimization for SU MIMO Systems under Perfect CSI}\label{SU_MIMO} 

 For SU MIMO systems of one source and one destination, both equipped with
 multiple antennas, the signal model is expressed by
\begin{align}\label{eq33}
 \bm{y} =& \bm{H} \bm{s} + \bm{n} ,
\end{align}
 where $\bm{y}$ is the received signal vector, $\bm{s}$ is the transmitted signal
 vector whose covariance matrix is $\bm{Q}=\mathbb{E}\big\{\bm{s}\bm{s}^{\rm H}\big\}$,
 $\bm{H}$ is the channel matrix, and $\bm{n}$ is the additive noise vector at the
 destination with the covariance matrix $\bm{R}_{\rm n}=\mathbb{E}\big\{\bm{n}
 \bm{n}^{\rm H}\big\}$.

 This section applies Theorem~\ref{T1} to the transmission covariance matrix
 optimization for SU MIMO systems under perfect CSI with the weighted sum power
 constraint and the multiple weighted power constraints, respectively. Both
 capacity maximization and MSE minimization are considered. The reason to
 investigate this `simple' transmission covariance matrix optimization is two-fold.
 First, the optimal solutions for the related problems have been derived using other
 approaches in the  existing literature. Thus these cases can be used for
 verification of our derivations using Theorem~\ref{T1}. Second, the results of
 these simple cases and their comparison with their counterparts for more
 complicated systems can help reveal the physical meanings and insights of the
 proposed solutions.

\subsection{Single Weighted Sum Power Constraint}\label{S3.1}

 For MIMO systems, when the channel statistics for different antennas are similar,
 the sum power constraint is a useful power model for transceiver optimization
 \cite{MatrixWaterFilling}. Weighted sum power constraint is a generalization of
 the sum power constraint, modeled as ${\rm Tr}(\bm{W}\bm{Q})\le P$, where $P$
 is the maximum transmit power and $\bm{W}$ is the weight matrix, which must
 be positive definite. If $\bm{W}$ has a zero eigenvalue, there is no power
 constraint at the corresponding direction, which is impractical.

 The capacity maximization problem under the weighted sum power constraint is
 formulated as follows
\begin{align}\label{capacity_weighted} 
\begin{array}{lcl}
 \textbf{P1}: & \min\limits_{\bm{Q}} & -\log\left|\bm{I} + \bm{R}_{\rm n}^{-1}\bm{H}
  \bm{Q}\bm{H}^{\rm H}\right| , \\
 & {\rm{s.t.}} & {\rm Tr}(\bm{W}\bm{Q})\le P, ~ \bm{Q}\succeq \bm{0} .
\end{array}
\end{align}
 The full set of KKT conditions for \textbf{P1} (\ref{capacity_weighted}) are given by
\begin{align}\label{eq35}
\hspace*{-2mm} \left\{ \!\!\!
\begin{array}{l}
 \bm{H}^{\rm H} \bm{R}_{\rm n}^{-\frac{1}{2}} \big( \bm{I}\! +\! \bm{R}_{\rm n}^{-\frac{1}{2}}
  \bm{H} \bm{Q} \bm{H}^{\rm H} \bm{R}_{\rm n}^{-\frac{1}{2}} \big)^{-1}\! \bm{R}_{\rm n}^{-\frac{1}{2}}
  \bm{H}\\ \hspace*{10mm} = \mu \bm{W}\! -\! \bm{\Psi} , \\
 \mu \ge 0 , ~ \mu \big({\rm Tr}(\bm{W}\bm{Q}) - P\big) = 0 , ~
  {\rm Tr}(\bm{Q}\bm{\Psi}) = 0 , \\
  \hspace*{10mm} \bm{\Psi} \succeq \bm{0} , ~{\rm Tr}(\bm{W}\bm{Q})\le P , ~ \bm{Q}\succeq \bm{0} ,
\end{array} \right.
\end{align}
 where $\mu$ is the Lagrangian multiplier associated with the constraint ${\rm Tr}(\bm{W}\bm{Q})
 \le P$ and the positive semi-definite matrix $\bm{\Psi}$ is the Lagrangian multiplier
 corresponding to the constraint $\bm{Q}\succeq \bm{0}$. Based on Theorem~\ref{T1} together
 with the replacements $\bm{\Pi}=\bm{R}_{\rm n}$ and $\bm{\Phi}=\bm{W}$, we have the following
 conclusion.

\begin{conclusion}\label{C1}
 The optimal transmission covariance matrix for $\textbf{P1}$ has the following water-filling
 structure
\begin{align}\label{eq36}
 \bm{Q}\! =& \bm{W}^{-\frac{1}{2}}\big[\bm{V}_{\bm{\mathcal{H}}}\big]_{:,1:N} \Big(\!
  \mu^{-1} \bm{I}\! -\! \big[\bm{\Lambda}_{\bm{\mathcal{H}}}\big]_{1:N,1:N}^{-2}\Big)^{+}\!
  \big[\bm{V}_{\bm{\mathcal{H}}}\big]_{:,1:N}^{\rm H} \bm{W}^{-\frac{1}{2}}\! ,
\end{align}
 where the unitary matrix $\bm{V}_{\bm{\mathcal{H}}}$ is defined by the SVD
\begin{align}\label{eq37}
 \bm{R}_{\rm n}^{-\frac{1}{2}}\bm{H}\bm{W}^{-\frac{1}{2}} =& \bm{U}_{\bm{\mathcal{H}}}
  \bm{\Lambda}_{\bm{\mathcal{H}}}\bm{V}_{\bm{\mathcal{H}}}^{\rm H}  \text{ with }
  \bm{\Lambda}_{\bm{\mathcal{H}}} \searrow .
\end{align}
\end{conclusion}

\noindent \textbf{Computation of $\mu$}:
 It is well-known that $\mu$ in Conclusion~\ref{C1} can be computed based on
\begin{align}\label{eq38}
 {\rm Tr}(\bm{W}\bm{Q}) =& P .
\end{align}
 This is obviously a standard water-level computation for the water-filling solution
 of MIMO capacity maximization. It can be computed both analytically or numerically
 \cite{Palomar03}.

 Similarly, the MSE minimization problem is formulated as
\begin{align}\label{eq39}
\begin{array}{lcl}
 \textbf{P2}: & \min\limits_{\bm{Q}} & {\rm Tr}\left(\big( \bm{I} + \bm{R}_{\rm n}^{-1}
  \bm{H}\bm{Q}\bm{H}^{\rm H}\big)^{-1}\right) , \\
 & {\rm{s.t.}} & {\rm Tr}(\bm{W}\bm{Q})\le P , ~ \bm{Q} \succeq \bm{0} .
\end{array}
\end{align}
 The full KKT conditions of \textbf{P2} (\ref{eq39}) are given by
\begin{align}\label{eq40}
\hspace*{-2mm} \left\{ \!\!\!
\begin{array}{l}
 \bm{H}^{\rm H}\bm{R}_{\rm n}^{-\frac{1}{2}}\big(\bm{I}\! +\! \bm{R}_{\rm n}^{-\frac{1}{2}}\bm{H}
  \bm{Q}\bm{H}^{\rm H}\bm{R}_{\rm n}^{-\frac{1}{2}}\big)^{-2}\bm{R}_{\rm n}^{-\frac{1}{2}}\bm{H}
  \\ \hspace*{10mm} =\! \mu \bm{W}\! -\! \bm{\Psi} , \\
 \mu \ge 0 , ~ \mu ({\rm Tr}(\bm{W}\bm{Q}) - P) = 0 , ~ {\rm Tr}(\bm{Q}\bm{\Psi}) = 0 , \\
  \hspace*{10mm} \bm{\Psi} \succeq \bm{0} , ~ {\rm Tr}(\bm{W}\bm{Q})\le P , ~ \bm{Q}\succeq \bm{0} .
\end{array} \right.
\end{align}
 Again based on Theorem~\ref{T1} with the replacements $\bm{\Pi}=\bm{R}_{\rm n}$ and
 $\bm{\Phi}=\bm{W}$,  the following conclusion holds.

\begin{conclusion}\label{C2}
 The optimal transmission covariance matrix for \textbf{P2} satisfies the following
 water-filling structure
\begin{align}\label{eq41}
 \bm{Q} =& \bm{W}^{-\frac{1}{2}}\big[\bm{V}_{\bm{\mathcal{H}}}\big]_{:,1:N}\Big(
  \mu^{-\frac{1}{2}}\big[\bm{\Lambda}_{\bm{\mathcal{H}}}\big]_{1:N,1:N}^{-1} -
  \big[\bm{\Lambda}_{\bm{\mathcal{H}}}\big]_{1:N,1:N}^{-2}\Big)^{+} \nonumber \\
 & \times \big[\bm{V}_{\bm{\mathcal{H}}}\big]_{:,1:N}^{\rm H}\bm{W}^{-\frac{1}{2}} .
\end{align}
\end{conclusion}

\noindent \textbf{Computation of $\mu$}:
 Similarly, $\mu$ in Conclusion~\ref{C2} can also be computed based on (\ref{eq38})
 analytically or numerically \cite{Palomar03}, which is a standard water-level
 computation for the water-filling solution for MIMO sum MSE minimization.

 For the capacity maximization problem \textbf{P1} and the MSE minimization problem
 \textbf{P2} for SU MIMO systems with perfect CSI and a weighted sum power constraint,
 the optimal solutions have been derived in the literature \cite{Telatar1999,Sampth01}.
 But our proposed method is different in the following two aspects. First, we solve
 the two optimizations using the same framework and case-by-case studies are avoided.
 Second, the proposed proofs are simpler and more rigorous. More specifically, the
 derivations in \cite{Telatar1999} require two steps, i.e., deriving the diagonalizable
 structure based on matrix inequality and deriving water-filling solution based on KKT
 conditions. By contrast, our work needs only one step, namely, from the KKT conditions
 to directly derive the optimal solutions. As pointed out in \cite{Xing2016}, the work
 in \cite{Sampth01} suffers from the turning-off effect, while our work successfully
 resolves this problem.

\subsection{Multiple Weighted Power Constraints}\label{S3.2}

 As each antenna in an antenna array has its own amplifier, the per-antenna power
 constraints are more practical than the sum power constraint. We consider the multiple
 weighted power constraints, which include the case of per-antenna power constraints as a
 special case. The generic multiple weighted power constraints take the following form
\begin{align}\label{per-cons} 
 {\rm Tr}(\bm{\Omega}_i\bm{Q}) \le P_i , ~ 1\le i\le I ,
\end{align}
 where $P_i$ is the $i$th power constraint with the positive semi-definite matrix
 $\bm{\Omega}_i$ as the corresponding weight matrix. When $\bm{\Omega}_i=\bm{b}_i
 \bm{b}_i^{\rm H}$, where $\bm{b}_i$ is the vector whose $i$th element is one and
 all the other elements are zeros,  we have ${\rm Tr}(\bm{\Omega}_i\bm{Q})=
 \bm{b}_i^{\rm H}\bm{Q}\bm{b}_i$. In this case, the $i$th power constraint in
 (\ref{per-cons}) becomes a per-antenna power constraint for the $i$th antenna.

 The capacity maximization for SU MIMO systems under multiple weighted power
 constraints is formulated as
\begin{align}\label{eq43}
\begin{array}{lcl}
 \textbf{P3}: & \min\limits_{\bm{Q}} & -\log\big|\bm{I}+\bm{R}_{\rm n}^{-1}\bm{H}
  \bm{Q}\bm{H}^{\rm H}\big| , \\
 & {\rm{s.t.}} & {\rm Tr}(\bm{\Omega}_i\bm{Q}) \le P_i , 1\le i\le I , ~
  \bm{Q} \succeq \bm{0} .
\end{array}
\end{align}
 The corresponding full KKT conditions are
\begin{align}\label{eq44}
\hspace*{-2mm} \left\{ \!\!\!
\begin{array}{l}
 \bm{H}^{\rm H}\bm{R}_{\rm n}^{-\frac{1}{2}}\big( \bm{I} + \bm{R}_{\rm n}^{-\frac{1}{2}}\bm{H}
  \bm{Q}\bm{H}^{\rm H}\bm{R}_{\rm n}^{-\frac{1}{2}}\big)^{-1}\bm{R}_{\rm n}^{-\frac{1}{2}}\bm{H} \\
 \hspace*{10mm} =\sum_{i=1}^I \mu_i\bm{\Omega}_i - \bm{\Psi} , \\
 \mu_i \ge 0 , ~ \mu_i\big({\rm Tr}(\bm{\Omega}_i\bm{Q}) - P_i\big)=0 , ~
  {\rm Tr}(\bm{\Omega}_i\bm{Q})\le P_i , \\
  \hspace*{5mm} 1\le i\le I , ~ {\rm Tr}(\bm{Q}\bm{\Psi})=0 , ~ \bm{\Psi} \succeq \bm{0} , ~
  \bm{Q}\succeq \bm{0} ,
\end{array} \right.
\end{align}
 where $\mu_i$ is the Lagrangian multiplier for the constraint ${\rm Tr}(\bm{\Omega}_i
 \bm{Q})\le P_i$. Based on Theorem~\ref{T1} and introducing an auxiliary variable $\mu$,
 together with the replacement $\bm{\Pi}=\bm{R}_{\rm n}$, we have the following
 conclusion for \textbf{P3}.

\begin{conclusion}\label{C3}
 The optimal transmission covariance matrix for \textbf{P3} has the following
 water-filling structure
\begin{align}\label{eq45}
 \bm{Q}\! =& \bm{\Phi}^{-\frac{1}{2}}\big[\bm{V}_{\bm{\mathcal{H}}}\big]_{:,1:N}
  \Big(\mu^{-1}\bm{I}\! -\! \big[\bm{\Lambda}_{\bm{\mathcal{H}}}\big]_{1:N,1:N}^{-2}\Big)^{+}
  \big[\bm{V}_{\bm{\mathcal{H}}}\big]_{:,1:N}^{\rm H}\bm{\Phi}^{-\frac{1}{2}} ,
\end{align}
 where $\bm{\Phi}$ is defined by
\begin{align}\label{Phi} 
 \bm{\Phi} =& \sum\nolimits_{i=1}^I \widetilde{\mu}_i\bm{\Omega}_i ,
\end{align}
 with
\begin{align}\label{eq47}
 \widetilde{\mu}_i =& \frac{\mu_i}{\mu} ,
\end{align}
 and the unitary matrix $\bm{V}_{\bm{\mathcal{H}}}$ is defined by the following SVD
\begin{align}\label{V_H} 
 \bm{R}_{\rm n}^{-\frac{1}{2}}\bm{H}\bm{\Phi}^{-\frac{1}{2}} =& \bm{U}_{\bm{\mathcal{H}}}
  \bm{\Lambda}_{\bm{\mathcal{H}}}\bm{V}_{\bm{\mathcal{H}}}^{\rm H} \text{ with }
  \bm{\Lambda}_{\bm{\mathcal{H}}} \searrow .
\end{align}
\end{conclusion}

\noindent \textbf{Computation of $\mu$ and $\mu_i$, $\forall i$}:
 It is worth noting that based on Theorem~\ref{T1}, in Conclusion~\ref{C3}, we can choose
 $\mu =1$ and the remaining task becomes how to compute $\mu_i$, $\forall i$, using
 numerical methods, such as subgradient methods. More specifically, to make the
 formulation consistent with its counterpart under the sum power constraint, $\mu$ is
 introduced, which makes sure
\begin{align}\label{eq49}
 {\rm Tr}\left(\sum\nolimits_{i=1}^I \widetilde{\mu}_i \bm{\Omega}_i\bm{Q}\right) =&
  \sum\nolimits_{i=1}^I P_i .
\end{align}
 When $\widetilde{\mu}_i$ are given, the solution of $\mu$ is unique. For this
 formulation, therefore, the remaining task is to search for $\widetilde{\mu}_i$
 instead of $\mu_i$, which can be solved by using subgradient algorithm. Thus,
 the benefit of introducing $\mu$ is threefold.
\begin{itemize}
\item First, it reveals that the multiple weighted constraints are equivalent to a single
 weighted constraint with proper weights $\widetilde{\mu}_i$. This can be verified by
 simply comparing the optimal solutions of \textbf{P1} and \textbf{P3}.
\item Second, in the process of searching for $\widetilde{\mu}_i$, $\bm{Q}$ can be
 restricted in a reasonable region, because (\ref{eq49}) always holds for the proposed
 formula. Hence, this formula can simplify the choice of initial $\widetilde{\mu}_i$.
\item The artful usage of $\mu$ can facilitate optimal solution derivations in some
 more difficult cases, e.g., the case under multiple weighted power constraints and
 imperfect CSI. In other cases, when it is unnecessary to compute $\mu$, such as for
 \textbf{P3}, we can directly set $\mu =1$.
\end{itemize}
 For water-filling solutions, the Lagrangian multipliers are computed iteratively
 or numerically. The solution given by Conclusion~\ref{C3}, however, can be taken as a
 closed-form solution.

 The MSE minimization problem under multiple weighted power constraints on the other hand
 is formulated as
\begin{align}\label{eq50}
\begin{array}{lcl}
 \textbf{P4}: & \min\limits_{\bm{Q}} & {\rm Tr}\Big(\big( \bm{I}+\bm{R}_{\rm n}^{-1}\bm{H}
  \bm{Q}\bm{H}^{\rm H}\big)^{-1}\Big) , \\
 & {\rm{s.t.}} & {\rm Tr}(\bm{\Omega}_i\bm{Q}) \le P_i , 1\le i\le I , ~
  \bm{Q} \succeq \bm{0} .
\end{array}
\end{align}
 After tedious but straightforward derivations, the corresponding full KKT conditions
 can be obtained as follows
\begin{align}\label{eq51}
\hspace*{-2mm} \left\{ \!\!\!
\begin{array}{l}
 \bm{H}^{\rm H}\bm{R}_{\rm n}^{-\frac{1}{2}}\big( \bm{I} + \bm{R}_{\rm n}^{-\frac{1}{2}}\bm{H}
  \bm{Q}\bm{H}^{\rm H}\bm{R}_{\rm n}^{-\frac{1}{2}}\big)^{-2}\bm{R}_{\rm n}^{-\frac{1}{2}}\bm{H} \\
 \hspace*{10mm} =\sum_{i=1}^I \mu_i\bm{\Omega}_i - \bm{\Psi} , \\
 \mu_i \ge 0 , ~ \mu_i\big({\rm Tr}(\bm{\Omega}_i\bm{Q}) - P_i\big)=0 , ~
  {\rm Tr}(\bm{\Omega}_i\bm{Q})\le P_i , \\
  \hspace*{5mm} 1\le i\le I , ~ {\rm Tr}(\bm{Q}\bm{\Psi})=0 , ~ \bm{\Psi} \succeq \bm{0} , ~
  \bm{Q}\succeq \bm{0} ,
\end{array} \right.
\end{align}
 By following exactly the same procedure as that for the capacity maximization, the
 optimal solution for \textbf{P4} is solved by Theorem~\ref{T1} and the following
 conclusion is obtained. Similar to the capacity maximization, the solution given in
 Conclusion~\ref{C4} is also a closed-form.

\begin{conclusion}\label{C4}
 The optimal transmission covariance matrix for \textbf{P4} has the following
 water-filling structure
\begin{align}\label{eq52}
 \bm{Q} =& \bm{\Phi}^{-\frac{1}{2}}\big[\bm{V}_{\bm{\mathcal{H}}}\big]_{:,1:N}
  \Big( \mu^{-\frac{1}{2}}\big[\bm{\Lambda}_{\bm{\mathcal{H}}}\big]_{1:N,1:N}^{-1} -
  \big[\bm{\Lambda}_{\bm{\mathcal{H}}}\big]_{1:N,1:N}^{-2}\Big)^{+} \nonumber \\
 & \times \big[\bm{V}_{\bm{\mathcal{H}}}\big]_{:,1:N}^{\rm H}\bm{\Phi}^{-\frac{1}{2}} ,
\end{align}
 where $\bm{\Phi}$ is given by (\ref{Phi}) and $\bm{V}_{\bm{\mathcal{H}}}$ is
 defined in (\ref{V_H}).
\end{conclusion}

\noindent \textbf{Computation of $\mu$ and $\mu_i$, $\forall i$}: The scalar $\mu$
 can be computed based on the equality (\ref{eq49}). When $\widetilde{\mu}_i$ are
 given, the solution of $\mu$ is unique. For the formulation in Conclusion~\ref{C4},
 the remaining task is also to search for $\widetilde{\mu}_i$ instead of $\mu_i$,
 which can be solved by using subgradient algorithms.

 For \textbf{P3} and \textbf{P4}, our distinct contribution is three-fold. First,
 although \textbf{P3} has been studied in \cite{Mai2011}, our proposed method is simpler
 and it overcomes the turning-off effect. Furthermore, we do not require that the
 channel matrix must be column or row full rank. Second, for \textbf{P4}, no solution
 has been discovered before. Based on our proposed framework, the optimal solution
 structure is derived for this problem. Third, our work reveals that the effect of
 multiple weighted constraints are in nature equivalent to a weighted power constraint
 by comparing the solutions for \textbf{P3} and \textbf{P4} with those for
 \textbf{P1} and \textbf{P2}. Again, like \textbf{P1} and \textbf{P2}, the optimal
 solutions given in Conclusions~\ref{C3} and \ref{C4} are in closed-form.

\section{Transmission Covariance Matrix Optimization for SU MIMO Systems Under Imperfect CSI}\label{S4}

 In practice, it is unrealistic to assume perfect CSI, as CSI must be estimated
 via training process. The limited training length and the ubiquitous noise together
 with the time varying nature of wireless channels make channel estimation error
 inevitable \cite{XingTSP201501}. By taking channel estimation error into account,
 the CSI can usually be modeled as \cite{XingTSP201501,Ding09}:
\begin{align}\label{error_model} 
 \bm{H} =& \widehat{\bm{H}} + \Delta\bm{H} \text{ with } \Delta\bm{H} = \bm{R}_{\rm R}^{\frac{1}{2}}
  \bm{H}_{\rm W}\bm{R}_{\rm T}^{\frac{1}{2}} .
\end{align}
 In this model, $\widehat{\bm{H}}$ is the estimated channel matrix, $\Delta\bm{H}$ is
 the CSI error, $\bm{R}_{\rm R}$ and $\bm{R}_{\rm T}$ are the corresponding
 receive and transmit spatial correlation matrices, respectively, while
 $\bm{H}_{\rm W}$ is a random matrix whose elements are independently and identically
 distributed (i.i.d.) random variables following the complex Gaussian distribution
 $\mathcal{CN}(0,1)$.

\subsection{Total Power Constraint with Imperfect CSI}\label{S4.1}

\subsubsection{Capacity maximization}

 With the CSI error (\ref{error_model}) and the white noise i.e., $\bm{R}_{\rm n}
 =\sigma_n^2\bm{I}$, the capacity maximization problem under a total power constraint
 is formulated as
\begin{align}\label{eq54}
\begin{array}{lcl}
 \textbf{P5}: & \min\limits_{\bm{Q}} & -\log \left| \bm{I} + \bm{K}_{\rm n}^{-1}\widehat{\bm{H}}
  \bm{Q}\widehat{\bm{H}}^{\rm H} \right| , \\
 & {\rm{s.t.}} & {\rm Tr}(\bm{Q})\le P , ~ \bm{Q}\succeq \bm{0},
\end{array}
\end{align}
 where the equivalent noise covariance matrix $\bm{K}_{\rm n}$ is given by
\begin{align}\label{eq55}
 \bm{K}_{\rm n} =& \sigma_n^2 \bm{I} + {\rm Tr}(\bm{R}_{\rm T} \bm{Q}) \bm{R}_{\rm R} ,
\end{align}
 and $P$ is the total power constraint.

 Based on the following matrix derivative equation
\begin{align}\label{eq56}
 & \frac{\partial\! \log\! \big|\bm{I}\! +\! \bm{K}_{\rm n}^{-1} \widehat{\bm{H}}\bm{Q}
  \widehat{\bm{H}}^{\rm H}\big|}{\partial \bm{Q}^*}\! =\!
  \frac{\partial\! \log\! \big|\bm{K}_{\rm n}\!\! +\! \widehat{\bm{H}}\bm{Q}\widehat{\bm{H}}^{\rm H}\big|}
  {\partial \bm{Q}^*}\! -\! \frac{\partial\! \log\! |\bm{K}_{\rm n}|}{\partial \bm{Q}^*}  \nonumber \\
 & \hspace*{5mm}= \widehat{\bm{H}}^{\rm H}\left(\bm{K}_{\rm n}+\widehat{\bm{H}}\bm{Q}
  \widehat{\bm{H}}^{\rm H}\right)^{-1}\widehat{\bm{H}} - {\rm Tr}\left(\bm{K}_{\rm n}^{-1}
  \bm{R}_{\rm R}\right) \bm{R}_{\rm T} \nonumber \\
 & \hspace*{9mm}+{\rm Tr}\left(\big(\bm{K}_{\rm n} + \widehat{\bm{H}}\bm{Q}
  \widehat{\bm{H}}^{\rm H}\big)^{-1}\bm{R}_{\rm R}\right)\bm{R}_{\rm T} ,
\end{align}
 the full KKT conditions of \textbf{P5} can be derived as follows
\begin{align}\label{eq57}
\hspace*{-2mm} \left\{ \!\!\!
\begin{array}{l}
 \widehat{\bm{H}}^{\rm H}\bm{K}_{\rm n}^{-\frac{1}{2}}\Big(\bm{I}\! +\! \bm{K}_{\rm n}^{-\frac{1}{2}}
  \widehat{\bm{H}}\bm{Q}\widehat{\bm{H}}^{\rm H}\bm{K}_{\rm n}^{-\frac{1}{2}}\Big)^{-1}
  \bm{K}_{\rm n}^{-\frac{1}{2}}\widehat{\bm{H}}\! =\! \mu \bm{I} \\
 \hspace*{4mm} +\! {\rm Tr}\bigg(\Big(\bm{K}_{\rm n}^{-1}\! -\! \big(\bm{K}_{\rm n}\! +\!
  \widehat{\bm{H}}\bm{Q}\widehat{\bm{H}}^{\rm H}\big)^{-1}\Big)
  \bm{R}_{\rm R}\bigg)\bm{R}_{\rm T}\! -\! \bm{\Psi} , \\
 \mu \ge 0 , ~ \mu\left({\rm Tr}(\bm{Q}) - P\right) = 0 , ~ {\rm Tr}(\bm{Q}\bm{\Psi}) = 0 , \\
  \hspace*{4mm} \bm{\Psi} \succeq \bm{0} , ~ {\rm Tr}(\bm{Q})\le P , ~ \bm{Q}\succeq \bm{0} .
\end{array} \right.
\end{align}
 Using Theorem~\ref{T1}, the following conclusion for the optimal solution of
 \textbf{P5} can be obtained.

\begin{conclusion}\label{C5}
 The optimal transmission covariance matrix for \textbf{P5} has the following
 water-filling structure
\begin{align}\label{Conclusion_5} 
 \bm{Q}\! =& \bm{\Phi}^{-\frac{1}{2}} \big[\bm{V}_{\bm{\mathcal{H}}}\big]_{:,1:N} \Big( \mu^{-1}
  \bm{I}\! -\! \big[\bm{\Lambda}_{\bm{\mathcal{H}}}\big]_{1:N,1:N}^{-2}\Big)^+
  \big[\bm{V}_{\bm{\mathcal{H}}}\big]_{:,1:N}^{\rm H}\bm{\Phi}^{-\frac{1}{2}} \! ,
\end{align}
 where $\bm{\Phi}$ is given by
\begin{align}\label{Phi_56} 
 \bm{\Phi}\! =&\! \bm{I}\! +\! \frac{1}{\mu}{\rm Tr}\bigg(\Big(\bm{K}_{\rm n}^{-1}\! -\!
  \big(\bm{K}_{\rm n}\! +\! \widehat{\bm{H}}\bm{Q}\widehat{\bm{H}}^{\rm H}\big)^{-1}\Big)
  \bm{R}_{\rm R}\bigg)\bm{R}_{\rm T} ,
\end{align}
 and the unitary matrix $\bm{V}_{\bm{\mathcal{H}}}$ is defined by the SVD
\begin{align}\label{eq60}
 & \bm{K}_{\rm n}^{-\frac{1}{2}}\widehat{\bm{H}}\bm{\Phi}^{-\frac{1}{2}} =
  \bm{U}_{\bm{\mathcal{H}}}\bm{\Lambda}_{\bm{\mathcal{H}}}\bm{V}_{\bm{\mathcal{H}}}^{\rm H}
  \text{ with } \bm{\Lambda}_{\bm{\mathcal{H}}} \searrow .
\end{align}
\end{conclusion}

 In Conclusion~\ref{C5}, the optimization variable $\bm{Q}$ appears in the both sides
 of the equality (\ref{Conclusion_5}). Hence, although the optimal structure is derived,
 it is not in a closed-form. Next, two special cases are considered: 1)~\textbf{P5.1}:
 the transmit antennas have no spatial correlation, and 2)~\textbf{P5.2}: the receive
 antennas have no spatial correlation. In these two special cases, the optimal solutions
 can be derived in closed-forms.

 \textbf{P5.1}: $\bm{R}_{\rm T}\propto {\bf{I}}$. It can be shown that $\bm{\Phi}
 \propto{\bf{I}}$, and $\bm{K}_{\rm n}$ of (\ref{eq55}) becomes
\begin{align}\label{eq61}
 \bm{K}_{\rm n} =& \sigma_n^2\bm{I} + r_b {\rm Tr}(\bm{Q})\bm{R}_{\rm R} = \sigma_n^2 \bm{I}
  + r_b\,P\,\bm{R}_{\rm R} .
\end{align}
 Consequently, Conclusion~\ref{C5} can be greatly simplified.

\vspace*{2mm}

\noindent
\textbf{Conclusion~5.1} \textit{When $\bm{R}_{\rm T}=r_b\bm{I}$, the optimal solution
 of \textbf{P5} can be simplified as}
\begin{align}\label{eq62}
 \bm{Q} =& \big[\bm{V}_{\bm{\mathcal{H}}}\big]_{:,1:N}\Big(\widetilde{\mu}^{-1}\bm{I} -
  \big[\bm{\Lambda}_{\bm{\mathcal{H}}}\big]_{1:N,1:N}^{-2}\Big)^{+}
  \big[\bm{V}_{\bm{\mathcal{H}}}\big]_{:,1:N}^{\rm H} ,
\end{align}
 \textit{where the unitary matrix $\bm{V}_{\bm{\mathcal{H}}}$ is defined by the following SVD}
\begin{align}\label{V_H_C} 
 \bm{K}_{\rm n}^{-\frac{1}{2}}\widehat{\bm{H}} =& \bm{U}_{\bm{\mathcal{H}}}
  \bm{\Lambda}_{\bm{\mathcal{H}}}\bm{V}_{\bm{\mathcal{H}}}^{\rm H}
  \text{ with } \bm{\Lambda}_{\bm{\mathcal{H}}} \searrow .
\end{align}

\noindent \textbf{Computation of $\widetilde{\mu}$}: The real parameter $\widetilde{\mu}$
 can be computed based on the equality
\begin{align}\label{eq64}
 {\rm Tr}(\bm{Q}) =& P .
\end{align}
 It can be seen that the computation of $\widetilde{\mu}$ is a standard water-level
 computation of the water-filling solution for MIMO capacity maximization. In other
 words, $\widetilde{\mu}$ in Conclusion~{5.1} can be computed analytically or
 numerically \cite{Palomar03}.

 Here, we would like to highlight that based on Conclusion~\ref{C5}, $\widetilde{\mu}$ in
 Conclusion~{5.1} actually equals to
\begin{align}\label{mu_f} 
 \widetilde{\mu}\! =& \mu \left(\! 1\! +\! \frac{r_b}{\mu}{\rm Tr}\left(\! \Big( \bm{K}_{\rm n}^{-1}
  \! -\! \big(\bm{K}_{\rm n}\! +\! \widehat{\bm{H}}\bm{Q}\widehat{\bm{H}}^{\rm H}\big)^{-1} \Big)
 \bm{R}_{\rm R}\! \right)\! \right) \!.
\end{align}
 It might appear that $\widetilde{\mu}$ is a function of $\bm{Q}$ and thus there might
 exist some constraint in computing $\widetilde{\mu}$. However this is actually not true.
 Because of the existence of the variable $\mu$, for the computed $\bm{Q}$ and
 $\widetilde{\mu}$, it is always possible to find a $\mu$ to make sure that the above
 equation holds. Note that substituting $\bm{Q}$ into (\ref{mu_f}), it can be proved
 that the solved  $\mu$ is always positive. As a result, $\widetilde{\mu}$ in
 Conclusion~{5.1} can be taken as an unconstrained variable. In a nutshell, the solution
 given by Conclusion~{5.1} is a closed-form solution.

 \textbf{P5.2}: When $\bm{R}_{\rm R}=r_a\bm{I}$, the following equation holds
\begin{align}\label{eq66}
 \bm{K}_{\rm n} =& \underbrace{\left(\sigma_n^2+r_a{\rm Tr}(\bm{Q}\bm{R}_{\rm T})
 \right)}_{\triangleq k_n} \bm{I}.
\end{align}
 Based on the definition of $k_n$, we have the following two equivalent equations
\begin{align}\label{power} 
 {\rm Tr}(\bm{Q})=P \leftrightarrow \frac{1}{k_n}{\rm Tr}\left(\left(\sigma_n^2\bm{I}
  +r_a \, P\, \bm{R}_{\rm T}\right)\bm{Q}\right)=P .
\end{align}
 Then by defining
\begin{align}\label{Q_Q}  
 \widetilde{\bm{Q}} =& \frac{1}{k_n}\bm{Q} ,
\end{align}
 the first KKT condition can be rewritten as
\begin{align}\label{KKT-2} 
 & \widehat{\bm{H}}^{\rm H}\left(\bm{I} + \widehat{\bm{H}}\widetilde{\bm{Q}}\widehat{\bm{H}}^{\rm H}
  \right)^{-1}\widehat{\bm{H}} = \mu \, k_n\bm{I} \nonumber \\
 & \hspace*{5mm} + r_a {\rm Tr}\left(\widehat{\bm{H}}^{\rm H}\left(\bm{I}\! + \!
  \widehat{\bm{H}}\widetilde{\bm{Q}}\widehat{\bm{H}}^{\rm H}\right)^{-1}
  \widehat{\bm{H}}\widetilde{\bm{Q}}\right)\bm{R}_{\rm T}\! -\! k_n \bm{\Psi} .
\end{align}
 When $\widetilde{\bm{Q}}$ has been computed, to derive $\bm{Q}$ based on (\ref{Q_Q}),
 $k_n$ should be computed. This is straightforward as ${\rm Tr}(\bm{Q})=P$. Based on
 (\ref{Q_Q}), it is obvious that
\begin{align}\label{k_n_equ} 
 k_n =& \frac{{\rm{Tr}}(\bm{Q})}{{\rm{Tr}}(\widetilde{\bm{Q}})} =
  \frac{P}{{\rm{Tr}}(\widetilde{\bm{Q}})} .
\end{align}
 Right multiplying $\widetilde{\bm{Q}}$ and taking trace operation on the both sides
 of (\ref{KKT-2}), we have
\begin{align}\label{eq71}
 & {\rm Tr}\left(\widehat{\bm{H}}^{\rm H}\left(\bm{I} + \widehat{\bm{H}}\widetilde{\bm{Q}}
  \widehat{\bm{H}}^{\rm H}\right)^{-1} \widehat{\bm{H}}\widetilde{\bm{Q}}\right)
  = \mu \,k_n {\rm Tr}\big(\widetilde{\bm{Q}}\big) \nonumber \\
 & \hspace{4mm} + r_a {\rm Tr}\! \left( \widehat{\bm{H}}^{\rm H}\left(\bm{I}\! +\! \widehat{\bm{H}}
  \widetilde{\bm{Q}}\widehat{\bm{H}}^{\rm H}\right)^{-1}\! \widehat{\bm{H}}\widetilde{\bm{Q}}
  \right) {\rm Tr}\! \left( \bm{R}_{\rm T}\widetilde{\bm{Q}}\right)\! ,
\end{align}
 based on which the following equalities are obtained
\begin{align}\label{KKT_1_a} 
 & {\rm Tr}\left( \widehat{\bm{H}}^{\rm H} \left( \bm{I} + \widehat{\bm{H}} \widetilde{\bm{Q}}
  \widehat{\bm{H}}^{\rm H} \right)^{-1} \widehat{\bm{H}} \widetilde{\bm{Q}} \right)
  =\frac{ \mu k_n {\rm Tr}\big(\widetilde{\bm{Q}}\big) }{ 1 - r_a {\rm Tr} \big( \bm{R}_{\rm T}
  \widetilde{\bm{Q}} \big)} \nonumber \\
 & \hspace*{10mm} =\frac{\mu k_n {\rm Tr}(\bm{Q})}{k_n - r_a{\rm Tr}(\bm{R}_{\rm T}\bm{Q})}
  = \frac{\mu k_n P}{\sigma_n^2} .
\end{align}
 From (\ref{KKT_1_a}), the first KKT condition (\ref{KKT-2}) becomes
\begin{align}\label{eq73}
 \widehat{\bm{H}}^{\rm H}\! \Big(\! \bm{I}\! +\! \widehat{\bm{H}} \widetilde{\bm{Q}} \widehat{\bm{H}} \Big)^{-1}
  \! \widehat{\bm{H}} =& \frac{\mu k_n}{\sigma_n^2} \left(\sigma_n^2 \bm{I}\! +\!
  r_a P \bm{R}_{\rm T}\right)\! -\! k_n \bm{\Psi} .\!
\end{align}
 Therefore, based on Theorem~\ref{T1} and defining
\begin{align}
 \widetilde{\mu} =& \frac{\mu k_n}{\sigma_n^2} , \label{eq74} \\
 \bm{\Phi} =& \sigma_n^2\bm{I}+r_a P\, \bm{R}_{\rm T} , \label{Phi_C} 
\end{align}
 we have
\begin{align}\label{eq76}
 & \widetilde{\bm{Q}}\! =\! \bm{\Phi}^{-\frac{1}{2}}\big[\bm{V}_{\bm{\mathcal{H}}}\big]_{:,1:N}
  \Big(\widetilde{\mu}^{-1}\bm{I}\! -\! \big[\bm{\Lambda}_{\bm{\mathcal{H}}}\big]_{1:N,1:N}^{-2}
  \Big)^{+}\big[\bm{V}_{\bm{\mathcal{H}}}\big]_{:,1:N}^{\rm H}\bm{\Phi}^{-\frac{1}{2}} ,
\end{align}
 in which the unitary matrix $\bm{V}_{\bm{\mathcal{H}}}$ is defined by the SVD
\begin{align}\label{V_H_2} 
 \widehat{\bm{H}} \bm{\Phi}^{-\frac{1}{2}} =& \bm{U}_{\bm{\mathcal{H}}}
  \bm{\Lambda}_{\bm{\mathcal{H}}}\bm{V}_{\bm{\mathcal{H}}}^{\rm H} \text{ with }
  \bm{\Lambda}_{\bm{\mathcal{H}}} \searrow .
\end{align}
 Based on the water-filling structure of (\ref{eq76}) together with (\ref{k_n_equ}),
 we arrive at the following conclusion.

\vspace*{2mm}

\noindent
\textbf{Conclusion~5.2} \textit{When $\bm{R}_{\rm R}=r_a\bm{I}$, the optimal solution
 of \textbf{P5} can be simplified as}
\begin{align}\label{eq78}
 \bm{Q} =& \frac{P}{{\rm Tr}\big(\widetilde{\bm{Q}}\big)} \widetilde{\bm{Q}} .
\end{align}

\noindent \textbf{Computation of ${\tilde \mu}$}:
 The real scalar $\widetilde{\mu}$ in Conclusion~{5.2} can be computed using the
 following equality which is derived based on (\ref{power})
\begin{align}\label{eq79}
 {\rm Tr}\big(\bm{\Phi}\widetilde{\bm{Q}}\big) =& P .
\end{align}
 It can be seen that the computation of $\widetilde{\mu}$ is a standard water-filling
 computation for MIMO capacity maximization. Therefore, $\widetilde{\mu}$ can be computed
 analytically or numerically \cite{Palomar03}.

 Based on Conclusion~\ref{C5}, $\widetilde{\mu}$ in Conclusion~{5.2} equals to
\begin{align}\label{eq80}
 \widetilde{\mu} =& \frac{k_n}{\sigma_n^2} \mu .
\end{align}
 For any computed $\widetilde{\mu}$ and $\bm{Q}$  or $k_n$ (as $k_n$ is a function
 of $\bm{Q}$), it is always possible to compute a $\mu$ to ensure that the above
 equality hold. This means that $\widetilde{\mu}$ can be directly solved based on
 (\ref{eq79}) without any additional constraint. Hence, the solution given in
 Conclusion~{5.2} is a closed-form solution.

\subsubsection{MSE minimization}

 The MSE minimization problem can similarly be formulated in the following form
\begin{align}\label{eq81}
\begin{array}{lcl}
 \textbf{P6}: & \min\limits_{\bm{Q}} & {\rm Tr} \left( \left( \bm{I} + \bm{K}_{\rm n}^{-1}
  \widehat{\bm{H}} \bm{Q} \widehat{\bm{H}}^{\rm H} \right)^{-1} \right) , \\
 & {\rm{s.t.}} & {\rm Tr}(\bm{Q})\le P , ~ \bm{Q}\succeq \bm{0} ,
\end{array}
\end{align}
 where $\bm{K}_{\rm n}$ is given in (\ref{eq55}). Based on the following matrix
 derivative equality
\begin{align}\label{eq82}
 & \frac{\partial {\rm Tr}\Big(\big( \bm{I} + \bm{K}_{\rm n}^{-1} \widehat{\bm{H}}
  \widetilde{\bm{Q}} \widehat{\bm{H}}\big)^{-1}\Big)}{\partial \bm{Q}^{*}} = \nonumber \\
 & \hspace*{4mm} - \widehat{\bm{H}}^{\rm H} \bm{K}_{\rm n}^{-\frac{1}{2}} \Big(
  \bm{I} + \bm{K}_{\rm n}^{-\frac{1}{2}} \widehat{\bm{H}} \bm{Q} \widehat{\bm{H}}^{\rm H}
   \bm{K}_{\rm n}^{-\frac{1}{2}}\Big)^{-2} \bm{K}_{\rm n}^{-\frac{1}{2}} \widehat{\bm{H}} \nonumber \\
 & \hspace*{4mm} + {\rm Tr} \bigg( \bm{K}_{\rm n}^{-1} \widehat{\bm{H}} \bm{Q}^{\frac{1}{2}}
  \Big( \bm{I} + \bm{Q}^{\frac{1}{2}} \widehat{\bm{H}}^{\rm H} \bm{K}_{\rm n}^{-1}
  \widehat{\bm{H}} \bm{Q}^{\frac{1}{2}} \Big)^{-2} \nonumber \\
 & \hspace{10mm} \times \bm{Q}^{\frac{1}{2}} \widehat{\bm{H}}^{\rm H}
  \bm{K}_{\rm n}^{-1} \bm{R}_{\rm R}\bigg) \bm{R}_{\rm T} ,
\end{align}
 the full KKT conditions of \textbf{P6} can be derived to be
\hspace*{-2mm}\begin{align}\label{eq83}
\left\{\begin{array}{l}\!\!\!
 \widehat{\bm{H}}^{\rm H} \bm{K}_{\rm n}^{-\frac{1}{2}} \! \Big(\! \bm{I}\! +\! \bm{K}_{\rm n}^{-\frac{1}{2}}
  \widehat{\bm{H}} \bm{Q} \widehat{\bm{H}}^{\rm H} \bm{K}_{\rm n}^{-\frac{1}{2}}\Big)^{-2}\!
  \bm{K}_{\rm n}^{-\frac{1}{2}} \widehat{\bm{H}}\! =\! \mu \bm{I} \\
 \hspace*{4mm} + {\rm Tr}\bigg( \bm{K}_{\rm n}^{-1} \widehat{\bm{H}} \bm{Q}^{\frac{1}{2}}
  \Big( \bm{I} + \bm{Q}^{\frac{1}{2}} \widehat{\bm{H}}^{\rm H} \bm{K}_{\rm n}^{-1}
  \widehat{\bm{H}} \bm{Q}^{\frac{1}{2}} \Big)^{-2} \\
 \hspace{8mm} \times \bm{Q}^{\frac{1}{2}} \widehat{\bm{H}}^{\rm H} \bm{K}_{\rm n}^{-1}
  \bm{R}_{\rm R} \bigg) \bm{R}_{\rm T} - \bm{\Psi} ,  \\
 \mu \ge 0 , ~ \mu ({\rm Tr}(\bm{Q})-P)=0 , ~ {\rm Tr}(\bm{Q}\bm{\Psi})=0 , \\
  \hspace*{4mm} \bm{\Psi} \succeq \bm{0} , ~ {\rm Tr}(\bm{Q})\le P , ~ \bm{Q}\succeq \bm{0} .
\end{array}\right. \!\!
\end{align}
 Based on Theorem~\ref{T1}, we have the following conclusion.

\begin{conclusion}\label{C6}
 The optimal transmission covariance matrix for \textbf{P6} has the following
 water-filling structure
\begin{align}\label{eq84}
 \bm{Q} =& \bm{\Phi}^{-\frac{1}{2}}\big[\bm{V}_{\bm{\mathcal{H}}}\big]_{:,1:N} \Big(
  \mu^{-\frac{1}{2}}\big[\bm{\Lambda}_{\bm{\mathcal{H}}}\big]_{1:N,1:N}^{-1} -
  \big[\bm{\Lambda}_{\bm{\mathcal{H}}}\big]_{1:N,1:N}^{-2}\Big)^{+} \nonumber \\
 & \times \big[\bm{V}_{\bm{\mathcal{H}}}\big]_{:,1:N}^{\rm H}\bm{\Phi}^{-\frac{1}{2}} ,
\end{align}
 where $\bm{\Phi}$ is given by
\begin{align}\label{eq85}
 \bm{\Phi} =& \bm{I} + \frac{1}{\mu} {\rm Tr}\bigg( \bm{K}_{\rm n}^{-1} \widehat{\bm{H}}
  \bm{Q}^{\frac{1}{2}} \Big( \bm{I} + \bm{Q}^{\frac{1}{2}} \widehat{\bm{H}}^{\rm H}
  \bm{K}_{\rm n}^{-1} \widehat{\bm{H}} \bm{Q}^{\frac{1}{2}} \Big)^{-2} \nonumber \\
 & \times \bm{Q}^{\frac{1}{2}} \widehat{\bm{H}}^{\rm H} \bm{K}_{\rm n}^{-1}
  \bm{R}_{\rm R} \bigg) \bm{R}_{\rm T} ,
\end{align}
 and the unitary matrix $\bm{V}_{\bm{\mathcal{H}}}$ is defined by the SVD
\begin{align}\label{eq86}
 \bm{K}_{\rm n}^{-\frac{1}{2}} \widehat{\bm{H}} \bm{\Phi}^{-\frac{1}{2}} =&
  \bm{U}_{\bm{\mathcal{H}}} \bm{\Lambda}_{\bm{\mathcal{ H}}} \bm{V}_{\bm{\mathcal{H}}}^{\rm H}
  \text{ with } \bm{\Lambda}_{\bm{\mathcal{H}}} \searrow .
\end{align}
\end{conclusion}

 The solution given by Conclusion~\ref{C6} is not in closed-form. Similarly, two
 special cases are considered: 1)~\textbf{P6.1}: the transmit antennas have no
 spatial correlation, and 2)~\textbf{P6.2}: the receive antennas have no spatial
 correlation.

 \textbf{P6.1}: For $\bm{R}_{\rm T}\propto \bm{I}$, we have the following simplified
 solution.

\vspace*{2mm}

\noindent
\textbf{Conclusion~6.1} \textit{When $\bm{R}_{\rm T}=r_b\bm{I}$, the optimal solution
 of \textbf{P6} has the following water-filling structure}
\begin{align}\label{eq87}
 \bm{Q}\! =& \! \big[\bm{V}_{\bm{\mathcal{H}}}\big]_{:,1:N}\! \Big( \widetilde{\mu}^{-1}
  \big[\bm{\Lambda}_{\bm{\mathcal{H}}}\big]_{1:N,1:N}^{-1}\! -\!
  \big[\bm{\Lambda}_{\bm{\mathcal{H}}}\big]_{1:N,1:N}^{-2} \Big)^{+}\!
  \big[\bm{V}_{\bm{\mathcal{H}}}\big]_{:,1:N}^{\rm H} , \!
\end{align}
 \textit{where the unitary matrix $\bm{V}_{\bm{\mathcal{H}}}$ is defined in (\ref{V_H_C}).}

\noindent \textbf{Computation of $\widetilde{\mu}$}:
 The real scalar $\widetilde{\mu}$ can be computed based on (\ref{eq64}). Obviously,
 this is a standard water-level computation for MIMO sum MSE minimization. In other words,
 $\widetilde{\mu}$ can be computed analytically or numerically \cite{Palomar03}.

 \textbf{P6.2}: For $\bm{R}_{\rm R}=r_a \bm{I}$, with the definition (\ref{Q_Q}), the
 first KKT condition can be reformulated as
\begin{align}\label{eq88}
 \widehat{\bm{H}}^{\rm H}\! \Big(\! \bm{I}\! +\! \widehat{\bm{H}} \widetilde{\bm{Q}}
  \widehat{\bm{H}}^{\rm H} \Big)^{-2} \widehat{\bm{H}}\! =& \frac{\mu k_n}{\sigma_n^2}
  \big(\sigma_n^2 \bm{I}\! +\! r_a P \bm{R}_{\rm T}\big)\! -\! k_n \bm{\Psi} .\!
\end{align}
 Based on Theorem~\ref{T1} and recalling the definitions $\widetilde{\mu}$ (\ref{eq74})
 and $\bm{\Phi}$ (\ref{Phi_C}), we have
\begin{align}\label{eq89}
 \widetilde{\bm{Q}} =& \bm{\Phi}^{-\frac{1}{2}} \big[\bm{V}_{\bm{\mathcal{H}}}\big]_{:,1:N}
  \Big( \widetilde{\mu}^{-\frac{1}{2}} \big[\bm{\Lambda}_{\bm{\mathcal{H}}}\big]_{1:N,1:N}^{-1}
  - \big[\bm{\Lambda}_{\bm{\mathcal{H}}}\big]_{1:N,1:N}^{-2} \Big)^{+} \nonumber \\
 & \times \big[\bm{V}_{\bm{\mathcal{H}}}\big]_{:,1:N}^{\rm H} \bm{\Phi}^{-\frac{1}{2}} ,
\end{align}
 where the unitary matrix $\bm{V}_{\bm{\mathcal{H}}}$ is specified by (\ref{V_H_2}).
 Together with (\ref{k_n_equ}), we arrive at the following conclusion.

\vspace*{2mm}

\noindent
\textbf{Conclusion~6.2} \textit{When $\bm{R}_{\rm R}=r_a\bm{I}$, the optimal $\bm{Q}$ of
 \textbf{P6} satisfies the following structure}
\begin{align}\label{eq90}
 \bm{Q} =& \frac{P}{{\rm Tr}\big(\widetilde{\bm{Q}}\big)}{\widetilde{\bm{Q}}} .
\end{align}

\noindent \textbf{Computation of $\widetilde{\mu}$}:
 The real scalar $\widetilde{\mu}$ can be computed based on (\ref{eq79}). By substituting
 $\widetilde{\bm{Q}}$ of (\ref{eq89}) into the equality (\ref{eq79}), it can readily be
 seen that the computation of $\widetilde{\mu}$ is a standard water-level computation for
 the water-filling computation of MIMO sum MSE minimization. Therefore, $\widetilde{\mu}$
 in Conclusion~{6.2} can be computed analytically or numerically \cite{Palomar03}.

 The solutions given in Conclusions~{5.1}, {5.2}, {6.1} and {6.2} are all in closed-forms,
 while the general solutions given by Conclusions~\ref{C5} and \ref{C6} can be computed
 using iterative algorithms, such as fixed point algorithms. Some similar optimization
 problems have also been discussed in the existing literature \cite{XingTSP2013,Ding09,Ding10}.
 Compared to Ding and Blostein's works \cite{Ding09,Ding10}, our derivations are much
 simpler and provide a unified framework optimization for both capacity maximization and
 MSE minimization. As discussed in \cite{Xing2016}, the methods in \cite{Ding09,Ding10}
 suffer from the turning-off effect, while the proposed framework does not have this
 problem. It is also worth highlighting that our previous work \cite{XingTSP2013} cannot
 cover the general conclusions given by Conclusions~\ref{C5} and \ref{C6}.

\subsubsection{Suboptimal solutions for \textbf{P5} and \textbf{P6}}

 As mentioned previously, the solutions given by Conclusions~\ref{C5} and \ref{C6} are not
 closed-form solutions. There exist two approaches to compute the solutions of \textbf{P5}
 and \textbf{P6}.

 \emph{3.1)}~In the first approach, we simply replace $\bm{R}_{\rm T}$ by
 $\lambda_{\max}(\bm{R}_{\rm T})\bm{I}$ or replace $\bm{R}_{\rm R}$ by
 $\lambda_{\max}(\bm{R}_{\rm R})\bm{I}$ in \textbf{P5} and \textbf{P6}. Then
 Conclusions~{5.1}, {5.2}, {6.1} and {6.2} become applicable with the closed-form solutions.
 These closed-form solutions are clearly suboptimal for \textbf{P5} and \textbf{P6}.

 \emph{3.2)}~On the other hand, a fixed point approach can be used to compute the solutions
 of \textbf{P5} and \textbf{P6}. From Conclusions~\ref{C5} and \ref{C6}, it can be seen that
 the main difficulty in computing the solutions of \textbf{P5} and \textbf{P6} come from
 the following three scalar terms in $\bm{\Phi}$ and $\bm{K}_{\rm n}$, defined respectively
 as $\alpha$, $\beta$ and $\gamma$:
\begin{align}
 \alpha =& {\rm Tr}\big(\bm{R}_{\rm T}\bm{Q}\big) , \label{eq91}
\end{align}
\begin{align}
 \beta =& {\rm Tr}\bigg( \Big( \bm{K}_{\rm n}^{-1} - \Big(\bm{K}_{\rm n} + \widehat{\bm{H}}
  \bm{Q} \widehat{\bm{H}}^{\rm H}\Big)^{-1} \Big) \bm{R}_{\rm R} \bigg) , \label{eq92} \\
 \gamma =& {\rm Tr}\bigg( \bm{K}_{\rm n}^{-1} \widehat{\bm{H}} \bm{Q}^{\frac{1}{2}}
  \Big( \bm{I} + \bm{Q}^{\frac{1}{2}} \widehat{\bm{H}}^{\rm H} \bm{K}_{\rm n}^{-1}
  \widehat{\bm{H}} \bm{Q}^{\frac{1}{2}} \Big)^{-2} \nonumber \\
 & \times \bm{Q}^{\frac{1}{2}} \widehat{\bm{H}}^{\rm H} \bm{K}_{\rm n}^{-1} \bm{R}_{\rm R} \bigg) . \label{eq93}
\end{align}
 For a fixed point approach, an iterative algorithm can be used to compute the solutions
 in Conclusions~\ref{C5} and \ref{C6}. Specifically, in the each iteration, the value of
 $\bm{Q}$ obtained in the previous iteration is used in computing $\alpha$, $\beta$ and
 $\gamma$. The simulation results show that this approach enjoys good performance. This
 is because when the channel estimation error is small, the values of $\alpha$, $\beta$
 and $\gamma$ also become small, and they do not affect the matrices $\bm{\Phi}$ and
 $\bm{K}_{\rm n}$ much.

 The above fixed point approach has some known weaknesses, e.g., the rigorous
 convergence proof is lacking. An alternative is to use its non-iterative approximation.
 It is worth noting that $\bm{K}_{\rm n}$ is the equivalent noise covariance matrix, and
 $\alpha$ and $\beta$ reflect the penalty of channel estimation error on the objective
 function. Therefore, using their upper bounds are preferred. Since $\alpha =
 {\rm Tr}\big(\bm{R}_{\rm T}\bm{Q}\big)\le \lambda_{\max}\big(\bm{R}_{\rm T}\big) P$,
 $\alpha$ can be replaced by the following upper bound
\begin{align}\label{eq94}
 \alpha_{\max} =& \lambda_{\max}\big(\bm{R}_{\rm T}\big) P .
\end{align}
 For given $\alpha$, $\bm{K}_{\rm n}$ becomes a constant. For any given $\bm{K}_{\rm n}$,
 an upper bound of $\beta$ is
\begin{align}\label{eq95}
 \beta_{\max} =& {\rm Tr}\big(\bm{K}_{\rm n}^{-1}\bm{R}_{\rm R}\big) .
\end{align}
 Furthermore, for $\gamma$, its upper bound is
\begin{align}\label{eq96}
 \gamma_{\max} =&\frac{1}{4}{\rm Tr}\big(\bm{K}_{\rm n}^{-1}\bm{R}_{\rm R}\big) .
\end{align}

\subsection{Multiple Weighted Power Constraints with imperfect CSI}\label{S4.2}

 To our best knowledge, the transmission covariance matrix optimizations for MIMO
 systems under multiple weighted power constraints and imperfect CSI have not been
 derived in the existing open literature.

\subsubsection{Capacity maximization}

 The capacity maximization problem can be formulated as follows
\begin{align}\label{eq97}
\begin{array}{lcl}
 \textbf{P7}: & \min\limits_{\bm{Q}} & -\log\Big| \bm{I} + \bm{K}_{\rm n}^{-1}
  \widehat{\bm{H}} \bm{Q} \widehat{\bm{H}}^{\rm H}\Big| , \\
 & {\rm s.t.} & {\rm Tr}(\bm{\Omega}_i\bm{Q}) \le P_i , 1\le i\le I, ~ \bm{Q} \succeq \bm{0} ,
\end{array}
\end{align}
 where $\bm{K}_{\rm n}$ is defined in (\ref{eq55}). The corresponding full KKT
 conditions are given by
\begin{align}\label{KKT_67} 
\left\{\begin{array}{l}\!\!\!
  \widehat{\bm{H}}^{\rm H} \bm{K}_{\rm n}^{-\frac{1}{2}} \! \Big(\! \bm{I}\! +\! \bm{K}_{\rm n}^{-\frac{1}{2}}
  \widehat{\bm{H}} \bm{Q} \widehat{\bm{H}}^{\rm H} \bm{K}_{\rm n}^{-\frac{1}{2}} \Big)^{-1}\!
  \bm{K}_{\rm n}^{-\frac{1}{2}} \widehat{\bm{H}}\! =\! \sum\limits_{i=1}^I \mu_i \bm{\Omega}_i \\
 \hspace{5mm} + {\rm Tr}\bigg( \Big( \bm{K}_{\rm n}^{-1}\! -\! \Big( \bm{K}_{\rm n}\! +\! \widehat{\bm{H}}
  \bm{Q} \widehat{\bm{H}}^{\rm H} \Big)^{-1} \Big) \bm{R}_{\rm R} \bigg) \bm{R}_{\rm T}
  - \bm{\Psi} , \\
 \mu_i \ge 0 , ~ \mu_i({\rm Tr}(\bm{\Omega}_i\bm{Q})-P_i)=0 , ~ {\rm Tr}(\bm{\Omega}_i\bm{Q})\le P_i , \\
  \hspace{5mm} 1\le i\le I , ~ {\rm Tr}(\bm{Q}\bm{\Psi})=0 , ~ \bm{\Psi}\succeq \bm{0} , ~
  \bm{Q}\succeq \bf{0} .
\end{array}\right. \!\!
\end{align}
 Based on Theorem~\ref{T1}, after introducing an auxiliary variable $\mu$, which ensures
 that the following equality holds
\begin{align}\label{mu_equ} 
 {\rm Tr}\left(\sum_{i=1}^I\frac{\mu_i}{\mu}\bm{\Omega}_i\bm{Q}\right) =& \sum_{i=1}^I
  \frac{\mu_i}{\mu} P_i = \sum_{i=1}^I P_i ,
\end{align}
 the following conclusion is obtained.

\begin{conclusion}\label{C7}
 The optimal transmission covariance matrix for \textbf{P7} has the following
 water-filling structure
\begin{align}\label{eq100}
 \bm{Q}\! =& \bm{\Phi}^{-\frac{1}{2}} \big[\bm{V}_{\bm{\mathcal{H}}}\big]_{:,1:N} \Big(
  \mu^{-1} \bm{I}\! -\! \big[\bm{\Lambda}_{\bm{\mathcal{H}}}\big]_{1:N,1:N}^{-2} \Big)^{+}\!
  \big[\bm{V}_{\bm{\mathcal{H}}}\big]_{:,1:N}^{\rm H} \bm{\Phi}^{-\frac{1}{2}}\! ,
\end{align}
 where
\begin{align}\label{eq101}
 \bm{\Phi}\! =&\! \sum_{i=1}^I\! \frac{\mu_i}{\mu}\bm{\Omega}_i\! +\! \frac{1}{\mu}\bigg(\! \Big(
  \bm{K}_{\rm n}^{-1}\! -\! \Big(\! \bm{K}_{\rm n}\! +\! \widehat{\bm{H}} \bm{Q} \widehat{\bm{H}}^{\rm H}
  \Big)^{-1} \Big) \bm{R}_{\rm R} \bigg)\! \bm{R}_{\rm T},
\end{align}
 and the unitary matrix $\bm{V}_{\bm{\mathcal{H}}}$ is defined by the SVD
\begin{align}\label{eq102}
 \bm{K}_{\rm n}^{-\frac{1}{2}} \widehat{\bm{H}} \bm{\Phi}^{-\frac{1}{2}} =& \bm{U}_{\bm{\mathcal{H}}}
  \bm{\Lambda}_{\bm{\mathcal{H}}} \bm{V}_{\bm{\mathcal{H}}}^{\rm H} \text{ with }
  \bm{\Lambda}_{\bm{\mathcal{H}}} \searrow .
\end{align}
\end{conclusion}

 For the special case of $\bm{R}_{\rm R}=r_a\bm{I}$, the first KKT condition in
 (\ref{KKT_67}) can be rewritten as
\begin{align}\label{eq103}
 & \widehat{\bm{H}}^{\rm H} \bm{K}_{\rm n}^{-\frac{1}{2}} \Big( \bm{I} + \bm{K}_{\rm n}^{-\frac{1}{2}}
  \widehat{\bm{H}} \bm{Q} \widehat{\bm{H}}^{\rm H} \bm{K}_{\rm n}^{-\frac{1}{2}} \Big)^{-1}
  \bm{K}_{\rm n}^{-\frac{1}{2}} \widehat{\bm{H}} = \sum_{i=1}^I \mu_i \bm{\Omega}_i \nonumber \\
 & \hspace*{5mm} + r_a {\rm Tr}\bigg(\! \Big( \bm{K}_{\rm n}^{-\frac{1}{2}}\! -\! \Big( \bm{K}_{\rm n}
  \! +\! \widehat{\bm{H}} \bm{Q} \widehat{\bm{H}}^{\rm H} \Big)^{-1} \Big) \bigg)\! \bm{R}_{\rm T}
  \! -\! \bm{\Psi} .
\end{align}
 It can be shown that the following two equalities are equivalent
\begin{align}
 & {\rm Tr}\left(\sum_{i=1}^I \mu_i \bm{\Omega}_i \bm{Q}\right) = \sum_{i=1}^I\mu_i P_i , \label{eq104} \\
 & \frac{1}{k_n}{\rm Tr}\left(\left( \sigma_n^2 \sum_{i=1}^I \mu_i \bm{\Omega}_i + r_a
  \sum_{i=1}^I \mu_i P_i \bm{R}_{\rm T} \right) \bm{Q}\right) = \sum_{i=1}^I \mu_i P_i , \label{eq105}
\end{align}
 where
\begin{align}\label{equ_98} 
 k_n =& \sigma_n^2 + r_a {\rm Tr}\big(\bm{R}_{\rm T}\bm{Q}\big) .
\end{align}
 Noting the definition $\widetilde{\bm{Q}}$ of (\ref{Q_Q}), the first KKT condition
 in (\ref{KKT_67}) becomes
\begin{align}\label{eq107}
 & \widehat{\bm{H}}^{\rm H} \Big( \bm{I} + \widehat{\bm{H}} \widetilde{\bm{Q}}
  \widehat{\bm{H}}^{\rm H} \Big)^{-1} \widehat{\bm{H}} = \sum_{i=1}^I k_n \mu_i \bm{\Omega}_i \nonumber \\
 & \hspace{2mm} +\! r_a {\rm Tr} \bigg(\! \Big(\! \bm{I}\! -\! \Big(\bm{I}\! +\! \widehat{\bm{H}}
  \widetilde{\bm{Q}} \widehat{\bm{H}}^{\rm H}\Big)^{-1}\Big)\! \bigg)\! \bm{R}_{\rm T}\! -\!
  k_n \bm{\Psi}\! =\! \sum_{i=1}^I\! k_n \mu_i \bm{\Omega}_i \nonumber \\
 & \hspace*{2mm} +\! r_a {\rm Tr} \bigg(\! \widehat{\bm{H}}^{\rm H}
  \! \Big(\! \bm{I}\! +\! \widehat{\bm{H}} \widetilde{\bm{Q}} \widehat{\bm{H}}^{\rm H} \Big)^{-1}
  \widehat{\bm{H}} \widetilde{\bm{Q}}\bigg) \bm{R}_{\rm T}\! -\! k_n \bm{\Psi} .
\end{align}
 Similar to (\ref{KKT_1_a}), based on (\ref{eq107}), we can prove that
\begin{align}\label{eq108}
 & {\rm Tr} \bigg( \widehat{\bm{H}}^{\rm H} \Big( \bm{I} + \widehat{\bm{H}} \widetilde{\bm{Q}}
  \widehat{\bm{H}}^{\rm H} \Big)^{-1} \widehat{\bm{H}} \widetilde{\bm{Q}}\bigg) =
  \frac{k_n {\rm Tr}\left(\sum_{i=1}^I \mu_i \bm{\Omega}_i \widetilde{\bm{Q}}\right)}
  {1-r_a {\rm Tr}\big(\bm{R}_{\rm T}  \widetilde{\bm{Q}}\big)} \nonumber \\
& \hspace*{10mm} =\frac{ k_n {\rm Tr}\left(\sum_{i=1}^I \mu_i\bm{\Omega}_i\bm{Q}\right)}
  {k_n-r_a {\rm Tr}\big(\bm{R}_{\rm T}\bm{Q}\big)} = \frac{k_n\sum_{i=1}^I \mu_i P_i}{\sigma_n^2} .
\end{align}
 Therefore, the first KKT condition in (\ref{KKT_67}) can be rewritten as
\begin{align}\label{equ_101_1} 
 &  \widehat{\bm{H}}^{\rm H} \Big( \bm{I} + \widehat{\bm{H}} \widetilde{\bm{Q}}
  \widehat{\bm{H}}^{\rm H} \Big)^{-1} \widehat{\bm{H}} = \sum\nolimits_{i=1}^I k_n \mu_i
  \bm{\Omega}_i\nonumber \\
 & \hspace*{10mm} + \frac{r_a k_n\sum\nolimits_{i=1}^I \mu_i P_i}{\sigma_n^2} \bm{R}_{\rm T}
  - k_n \bm{\Psi} ,
\end{align}
 based on which and using Theorem~\ref{T1}, we have
\begin{align}\label{equ_101} 
 & \widetilde{\bm{Q}}\! =\! \bm{\Phi}^{-\frac{1}{2}} \big[\bm{V}_{\bm{\mathcal{H}}}\big]_{:,1:N}
  \Big( \widetilde{\mu}^{-1} \bm{I}\! -\! \big[\bm{\Lambda}_{\bm{\mathcal{H}}}\big]_{1:N,1:N}^{-2}
  \Big)^{+} \big[\bm{V}_{\bm{\mathcal{H}}}\big]_{:,1:N}^{\rm H}\bm{\Phi}^{-\frac{1}{2}} ,
\end{align}
 where
\begin{align}
 \widetilde{\mu} =& \frac{k_n \mu}{\sigma_n^2} , \label{eq111} \\
 \bm{\Phi} =& \sigma_n^2 \sum\nolimits_{i=1}^I \widetilde{\mu}_i \bm{\Omega}_i + r_a
  \sum\nolimits_{i=1}^I \widetilde{\mu}_i P_i \bm{R}_{\rm T}, \label{eq112} \\
 \widetilde{\mu}_i =& \frac{\mu_i}{\mu} , \label{eq113}
\end{align}
 and the unitary matrix $\bm{V}_{\bm{\mathcal{H}}}$ is defined by the SVD
\begin{align}\label{eq114}
 & \widehat{\bm{H}} \bm{\Phi}^{-\frac{1}{2}} = \bm{U}_{\bm{\mathcal{H}}}\bm{\Lambda}_{\bm{\mathcal{H}}}
  \bm{V}_{\bm{\mathcal{H}}}^{\rm H} \text{ with } \bm{\Lambda}_{\bm{\mathcal{H}}} \searrow .
\end{align}
 Noting $\widetilde{\bm{Q}}=\frac{1}{k_n}\bm{Q}$, we have
\begin{align}\label{equ_103} 
 k_n =& \frac{{\rm Tr}\left(\sum\nolimits_{i=1}^I \mu_i \bm{\Omega}_i \bm{Q}\right)}
  {{\rm Tr}\left(\sum\nolimits_{i=1}^I \mu_i \bm{\Omega}_i \widetilde{\bm{Q}}\right)}
 = \frac{\sum\nolimits_{i=1}^I \mu_i P_i}{{\rm Tr}\left(\sum\nolimits_{i=1}^I \mu_i
  \bm{\Omega}_i \widetilde{\bm{Q}}\right)} \nonumber \\
 =& \frac{\sum\nolimits_{i=1}^I \widetilde{\mu}_i P_i}{{\rm Tr}\left(\sum\nolimits_{i=1}^I
  \widetilde{\mu}_i \bm{\Omega}_i \widetilde{\bm{Q}}\right)} =
  \frac{\sum\nolimits_{i=1}^I P_i}{{\rm Tr}\left(\sum\nolimits_{i=1}^I
  \widetilde{\mu}_i \bm{\Omega}_i \widetilde{\bm{Q}}\right)} ,
\end{align}
 where the last equality is due to (\ref{mu_equ}). Based on (\ref{equ_98}), (\ref{equ_101})
 and (\ref{equ_103}), the following conclusion is obtained.

\vspace*{2mm}

\noindent
\textbf{Conclusion~7.1} \textit{When $\bm{R}_{\rm R}=r_a\bm{I}$, the optimal $\bm{Q}$ for
 \textbf{P7} has the following structure}
\begin{align}\label{eq116}
 \bm{Q} =& \frac{\sum\nolimits_{i=1}^I P_i}{{\rm Tr}\left(\sum\nolimits_{i=1}^I \widetilde{\mu}_i
  \bm{\Omega}_i \widetilde{\bm{Q}}\right)} \widetilde{\bm{Q}} .
\end{align}

\noindent \textbf{Computation of $\widetilde{\mu}$ and $\widetilde{\mu}_i$, $\forall i$}:
 The scalar $\widetilde{\mu}$ ensures that the following equality holds
\begin{align}\label{eq117}
 {\rm Tr}\big(\bm{\Phi}\widetilde{\bm{Q}}\big)=\sum\nolimits_{i=1}^I P_i .
\end{align}
 Substituting Conclusion~{7.1} into the above equation, it is obvious that computation of
 $\widetilde{\mu}$ is a standard water-level computation of the water-filling solution for
 MIMO capacity maximization. The scalars $\widetilde{\mu}_i$ can be computed by using
 subgradient methods.

\subsubsection{MSE minimization}

 The MSE minimization problem for SU MIMO systems under multiple weighted power constraints
 and imperfect CSI can be written in the following form \cite{XingTSP2013}
\begin{align}\label{eq118}
\begin{array}{lcl}
 \textbf{P8}: & \min\limits_{\bm{Q}} & {\rm Tr}\bigg(\Big( \bm{I} + \bm{K}_{\rm n}^{-1}
  \widehat{\bm{H}} \bm{Q} \widehat{\bm{H}}^{\rm H}\Big)^{-1}\bigg) , \\
 & {\rm{s.t.}} & {\rm Tr}(\bm{\Omega}_i\bm{Q}) \le P_i , 1\le i\le I, ~ \bm{Q} \succeq \bm{0} ,
\end{array}
\end{align}
 with $\bm{K}_{\rm n}$ defined in (\ref{eq55}). The full KKT conditions of \textbf{P8} are
\begin{align}\label{eq119}
\left\{\begin{array}{l}\!\!\!
 \widehat{\bm{H}}^{\rm H} \bm{K}_{\rm n}^{-\frac{1}{2}}\! \Big(\! \bm{I}\! +\! \bm{K}_{\rm n}^{-\frac{1}{2}}
  \widehat{\bm{H}} \bm{Q} \widehat{\bm{H}}^{\rm H} \bm{K}_{\rm n}^{-\frac{1}{2}} \Big)^{-2}\!
  \bm{K}_{\rm n}^{-\frac{1}{2}} \widehat{\bm{H}}\! =\! \sum\limits_{i=1}^I \mu_i \bm{\Omega}_i  \\
 \hspace{2mm} + {\rm Tr}\bigg( \bm{K}_{\rm n}^{-1} \widehat{\bm{H}} \bm{Q}^{\frac{1}{2}}
  \Big( \bm{I} + \bm{Q}^{\frac{1}{2}} \widehat{\bm{H}}^{\rm H} \bm{K}_{\rm n}^{-1}
  \widehat{\bm{H}} \bm{Q}^{\frac{1}{2}}\Big)^{-2} \\
 \hspace{5mm} \times \bm{Q}^{\frac{1}{2}} \widehat{\bm{H}}^{\rm H} \bm{K}_{\rm n}^{-1}
  \bm{R}_{\rm R} \bigg) \bm{R}_{\rm T} - \bm{\Psi} , \\
 \mu_i \ge 0, ~ \mu_i({\rm Tr}(\bm{\Omega}_i \bm{Q})-P_i)=0, ~ {\rm Tr}(\bm{\Omega}_i\bm{Q})\le P_i, \\
 \hspace*{2mm} 1\le i\le I , ~ {\rm Tr}(\bm{Q}\bm{\Psi})=0, ~ \bm{\Psi} \succeq \bm{0} , ~
  \bm{Q}\succeq \bm{0} .
\end{array}\right. \!\!
\end{align}
 Based on Theorem~\ref{T1} and introducing an auxiliary variable $\mu$,
 the following conclusion is obtained.

\begin{conclusion}\label{C8}
 The optimal transmission covariance matrix of \textbf{P8} has the following structure
\begin{align}\label{eq120}
 \bm{Q} =& \bm{\Phi}^{-\frac{1}{2}}\big[\bm{V}_{\bm{\mathcal{H}}}\big]_{:,1:N} \Big(
  \mu^{-\frac{1}{2}} \big[\bm{\Lambda}_{\bm{\mathcal{H}}}\big]_{1:N,1:N}^{-1} -
  \big[\bm{\Lambda}_{\bm{\mathcal{H}}}\big]_{1:N,1:N}^{-2} \Big)^{+} \nonumber \\
 &\times \big[\bm{V}_{\bm{\mathcal{H}}}\big]_{:,1:N}^{\rm H} \bm{\Phi}^{-\frac{1}{2}} ,
\end{align}
 where $\bm{\Phi}$ is given by
\begin{align}\label{eq121}
 \bm{\Phi}\! =&\! \sum\limits_{i=1}^I\! \frac{\mu_i}{\mu} \bm{\Omega}_i\! +\! \frac{1}{\mu}
  {\rm Tr} \bigg(\! \bm{K}_{\rm n}^{-1} \widehat{\bm{H}} \bm{Q}^{\frac{1}{2}}
  \Big( \bm{I}\! +\! \bm{Q}^{\frac{1}{2}} \widehat{\bm{H}}^{\rm H} \bm{K}_{\rm n}^{-1}
   \widehat{\bm{H}} \bm{Q}^{\frac{1}{2}} \Big)^{-2} \nonumber \\
 & \times \bm{Q}^{\frac{1}{2}} \widehat{\bm{H}}^{\rm H} \bm{K}_{\rm n}^{-1}
  \bm{R}_{\rm R}\bigg) \bm{R}_{\rm T} ,
\end{align}
 and the unitary matrix $\bm{V}_{\bm{\mathcal{H}}}$ is defined by the SVD
\begin{align}\label{eq122}
 \bm{K}_{\rm n}^{-\frac{1}{2}} \widehat{\bm{H}} \bm{\Phi} ^{-\frac{1}{2}} =&
  \bm{U}_{\bm{\mathcal{H}}} \bm{\Lambda}_{\bm{\mathcal{H}}} \bm{V}_{\bm{\mathcal{H}}}^{\rm H}
  \text{ with } \bm{\Lambda}_{\bm{\mathcal{H}}} \searrow .
\end{align}
\end{conclusion}

 For the special case of $\bm{R}_{\rm R}=r_a\bm{I}$, the first KKT condition can be written as
\begin{align}\label{eq123}
 & \widehat{\bm{H}}^{\rm H} \bm{K}_{\rm n}^{-\frac{1}{2}} \Big( \bm{I} +
  \bm{K}_{\rm n}^{-\frac{1}{2}} \widehat{\bm{H}} \bm{Q} \widehat{\bm{H}}^{\rm H}
  \bm{K}_{\rm n}^{-\frac{1}{2}} \Big)^{-2} \bm{K}_{\rm n}^{-\frac{1}{2}} \widehat{\bm{H}}
  = \sum\limits_{i=1}^I \! \mu_i \bm{\Omega}_i \nonumber \\
 & \hspace{2mm} + r_a {\rm Tr} \bigg( \bm{K}_{\rm n}^{-1} \widehat{\bm{H}} \bm{Q}^{\frac{1}{2}}
  \Big( \bm{I} + \bm{Q}^{\frac{1}{2}} \widehat{\bm{H}}^{\rm H} \bm{K}_{\rm n}^{-1}
  \widehat{\bm{H}} \bm{Q}^{\frac{1}{2}} \Big)^{-2} \nonumber \\
 & \hspace{5mm} \times \bm{Q}^{\frac{1}{2}} \widehat{\bm{H}}^{\rm H} \bm{K}_{\rm n}^{-1}
  \bigg) \bm{R}_{\rm T} - \bm{\Psi} .
\end{align}
 Based on  (\ref{equ_101_1}), (\ref{eq123}) can be rewritten as
\begin{align}\label{eq124}
 & \widehat{\bm{H}}^{\rm H} \Big( \bm{I} + \widehat{\bm{H}} \widetilde{\bm{Q}}
  \widehat{\bm{H}}^{\rm H}\Big)^{-2}\widehat{\bm{H}}  \nonumber \\
 & \hspace*{2mm} = \widetilde{\mu} \left(\sum_{i=1}^I \sigma_n^2 \widetilde{\mu}_i \bm{\Omega}_i
  + \bigg( \sum_{i=1}^I \widetilde{\mu}_i P_i\bigg) r_a \bm{R}_{\rm T}\right) - k_n \bm{\Psi} ,
\end{align}
 where $\widetilde{\mu}$ and $\widetilde{\mu}_i$ are defined in (\ref{eq111}) and
 (\ref{eq113}), respectively. Based on Theorem~\ref{T1}, we have the following conclusion.

\vspace*{2mm}

\noindent
\textbf{Conclusion~8.1} \textit{When $\bm{R}_{\rm R}=r_a\bm{I}$, the optimal
 $\bm{Q}$ of \textbf{P8} satisfies the following structure}
\begin{align}\label{eq125}
 \bm{Q} =& \frac{\sum_{i=1}^I P_i}{{\rm Tr}\left(\sum_{i=1}^I \widetilde{\mu}_i \bm{\Omega}_i
  \widetilde{\bm{Q}}\right)} \widetilde{\bm{Q}} ,
\end{align}
 \textit{with}
\begin{align}\label{eq126}
 \widetilde{\bm{Q}} =& \bm{\Phi}^{-\frac{1}{2}}\big[\bm{V}_{\bm{\mathcal{H}}}\big]_{:,1:N}
  \Big(\widetilde{\mu}^{-\frac{1}{2}} \big[\bm{\Lambda}_{\bm{\mathcal{H}}}\big]_{1:N,1:N}^{-1}
  - \big[\bm{\Lambda}_{\bm{\mathcal{H}}}\big]_{1:N,1:N}^{-2}\Big)^{+} \nonumber \\
 & \times \big[\bm{V}_{\bm{\mathcal{H}}}\big]_{:,1:N}^{\rm H}\bm{\Phi}^{-\frac{1}{2}} ,
\end{align}
 \textit{where $\bm{\Phi}$ and $\bm{V}_{\bm{\mathcal{H}}}$ are defined in
 (\ref{eq112}) and (\ref{eq114}), respectively.}

\vspace*{2mm}

\noindent \textbf{Computation of $\widetilde{\mu}$ and $\widetilde{\mu}_i$, $\forall i$}:
 The scalar $\widetilde{\mu}$ makes sure that the equality (\ref{eq117}) holds.
 Substituting Conclusion~{8.1} into (\ref{eq117}), it is obvious that computation of
 $\widetilde{\mu}$ is a standard water-level computation of the water-filling solution
 for MIMO sum MSE minimization. Thus $\widetilde{\mu}$ can be computed effectively. The
 scalars $\widetilde{\mu}_i$ can be computed by using subgradient methods.

 The capacity maximization \textbf{P7} and MSE minimization \textbf{P8} for SU MIMO
 systems under imperfect CSI and multiple weighted power constraints, to our best
 knowledge, have not been derived in the existing open literature. The structures of
 the optimal solutions for \textbf{P7} and \textbf{P8}, given in Conclusions~\ref{C7}
 and \ref{C8}, are in complicated forms and iterative schemes, such as fixed point
 algorithm are needed to compute the optimal solutions. In the special case with only
 transmit antenna correlation, the optimal solutions of \textbf{P7} and \textbf{P8}
 can be derived in closed-form as given in Conclusions~{7.1} and {8.1}, respectively.

\subsubsection{Suboptimal solutions for \textbf{P7} and \textbf{P8}}

 In general, the water-filling structures for the optimal covariance matrices for SU
 MIMO systems under multiple weighted power constraints and imperfect CSI, derived in
 Conclusions~\ref{C7} and \ref{C8}, are not closed-form solutions. Similarly, there
 are two approaches for computing the solutions of \textbf{P7} and \textbf{P8}.

 \emph{3.1)}~Some approximation is used to achieve closed-form solution. Specifically,
 by replacing $\bm{R}_{\rm R}$ with  $\lambda_{\max}(\bm{R}_{\rm R})\bm{I}$, the
 closed-form solutions of Conclusions~{7.1} and {8.1} become applicable, which offer
 good approximate closed-form solutions to \textbf{P7} and \textbf{P8}, respectively.

 \emph{3.2)}~Fixed point methods can be used to compute the solutions of \textbf{P7}
 and \textbf{P8}. To compute $\bm{\Phi}$ and $\bm{K}_{\rm n}$, the three scalar terms,
 $\alpha$, $\beta$ and $\gamma$, are needed, which are defined in (\ref{eq91}) to
 (\ref{eq93}). From (\ref{eq91}) to (\ref{eq93}), it can be seen that $\alpha$, $\beta$
 and $\gamma$ all depend on $\bm{Q}$. When using a fixed point algorithm to compute the
 solutions in Conclusions~\ref{C7} and \ref{C8}, in the each iteration, the value of
 $\bm{Q}$ obtained in the previous iteration is used to calculate the current $\alpha$,
 $\beta$ and $\gamma$.

 Although the simulation results show that this iterative algorithm enjoys good
 performance, it suffers from some well-known drawbacks, e.g., convergence may be slow.
 Again, an alternative is to use its non-iterative approximation by replacing $\alpha$,
 $\beta$ and $\gamma$ with their respective upper bounds $\alpha_{\max}$, $\beta_{\max}$
 and $\gamma_{\max}$, given in (\ref{eq94}) to (\ref{eq96}).

\section{Transmission Covariance Matrix Optimization for MU MIMO under Multiple Weighted Power Constraints}\label{S5}

 In the MU MIMO system investigated, $K$ multi-antenna aided users communicate with the
 multi-antenna aided base station (BS). The channel between the $k$th user and the BS is
 denoted by $\bm{H}_k$ and the $k$th user's transmission covariance matrix is denoted
 by $\bm{Q}_k$, for $1\le k\le K$. Multiple weighted power constraints are considered
 with $\bm{\Omega}_{k,i}$ being the weighting matrix for the $i$th power constraint
 of the $k$th user and $P_{k,i}$ being the corresponding power limit, for $1\le i\le I_k$.

\subsection{MU MIMO Uplink}\label{S5.1}

\subsubsection{Perfect CSI}

 For the uplink where the users send data to the BS, the capacity maximization
 problem under multiple weighted power constraints and perfect CSI is formulated as
\begin{align}\label{eq127}
\begin{array}{lcl}
 \textbf{P9}: \!\! &\!\! \min\limits_{\big\{\bm{Q}_k\big\}_{k=1}^K}\!\! &\!\! -\log\left|\bm{I} + \bm{R}_{\rm n}^{-1}
 \sum\limits_{k=1}^K\big(\bm{H}_k\bm{Q}_k\bm{H}_k^{\rm H}\big)\right| , \\
 \!\! &\!\! {\rm{s.t.}}\!\! &\!\! {\rm Tr}\big(\bm{\Omega}_{k,i}\bm{Q}_k\big)\le P_{k,i} , 1\le i\le I_k , \\
 \!\! &\!\! \!\! &\!\! \bm{Q}_k \succeq \bm{0}, 1\le k\le K .
\end{array}
\end{align}
 The full KKT conditions of \textbf{P9} are
\begin{align}\label{eq128}
\left\{\begin{array}{l}\!\!\!
 \bm{H}_k^{\rm H}\left( \bm{R}_{\rm n} + \sum\limits_{j\neq k} \bm{H}_j\bm{Q}_j\bm{H}_j^{\rm H}
  + \bm{H}_k\bm{Q}_k\bm{H}_k^{\rm H}\right)^{-1} \bm{H}_k \\
 \hspace{5mm} =\sum\limits_{i=1}^{I_k} \mu_{k,i} \bm{\Omega}_{k,i} - \bm{\Psi}_k, ~
  1\le k\le K , \\
 \mu_{k,i} \ge 0 , ~ \mu_{k,i}({\rm Tr}(\bm{\Omega}_{k,i}\bm{Q}_k)-P_{k,i})=0 , \\
 \hspace{5mm} {\rm Tr}(\bm{\Omega}_{k,i}\bm{Q}_k)\le P_{k,i} , ~ 1\le i\le I_k , \\
 {\rm Tr}(\bm{Q}_k\bm{\Psi}_k)=0 , ~ \bm{\Psi}_k \succeq \bm{0} , ~
  \bm{Q}_k\succeq \bm{0} , ~ 1\le k\le K ,
\end{array}\right. \!\!
\end{align}
 where $\bm{\Psi}_k$ is the Lagrangian multiplier corresponding to the constraint
 $\bm{Q}_k \succeq \bm{0}$. By using Theorem~\ref{T1} and defining
\begin{align}\label{eq129}
 \widetilde{\mu}_{k,i} =& \frac{\mu_{k,i}}{\mu_k}
\end{align}
 the following conclusion is obtained.

\begin{conclusion}\label{C9}
 The optimal transmission covariance matrices for \textbf{P9} have the following
 water-filling structure
\begin{align}\label{eq130}
 \bm{Q}_k =& \bm{\Phi}_k^{-\frac{1}{2}}\big[\bm{V}_{\bm{\mathcal{H}}_k}\big]_{:,1:N}
  \left(\mu_k^{-1}\bm{I} - \big[\bm{\Lambda}_{\bm{\mathcal{H}}_k}\big]_{1:N,1:N}^{-2}\right)^{+} \nonumber \\
 & \times \big[\bm{V}_{\bm{\mathcal{H}}_k}\big]_{:,1:N}^{\rm H}\bm{\Phi}_k^{-\frac{1}{2}} ,
  ~ 1\le k\le K ,
\end{align}
 where $\bm{\Phi}_k$ is defined as
\begin{align}\label{eq131}
 \bm{\Phi}_k =&\sum\limits_{i=1}^{I_k} \widetilde{\mu}_{k,i}\bm{\Omega}_{k,i} ,
\end{align}
 and the unitary matrix ${\bf{V}}_{{\boldsymbol{\mathcal{H}}}_k}$ is defined by the following SVD
\begin{align}\label{eq132}
 \bm{\Pi}_k^{-\frac{1}{2}}\bm{H}_k\bm{\Phi}_k^{-\frac{1}{2}} =& \bm{U}_{\bm{\mathcal{H}}_k}
  \bm{\Lambda}_{\bm{\mathcal{H}}_k} \bm{V}_{\bm{\mathcal{H}}_k}^{\rm H} \text{ with }
  \bm{\Lambda}_{\bm{\mathcal{H}}_k} \searrow ,
\end{align}
 in which
\begin{align}\label{eq133}
 \bm{\Pi}_k =& \bm{R}_{\rm n} + \sum\limits_{j\neq k} \bm{H}_j\bm{Q}_j\bm{H}_j^{\rm H} .
\end{align}
\end{conclusion}

\noindent \textbf{Computation of $\mu_k$ and $\widetilde{\mu}_{k,i}$, $\forall k,i$}:
 The scalars $\mu_k$ can be computed based on the following equality
\begin{align}\label{eq134}
 {\rm Tr}\left(\sum\limits_{i=1}^{I_k} \widetilde{\mu}_{k,i}\bm{\Omega}_{k,i}\bm{Q}_k\right)
 =& \sum\limits_{i=1}^{I_k}P_{k,i} , 1\le k\le K .
\end{align}
 Hence the computation of $\mu_k$ is a standard water-level computation for iterative
 water-filling solution \cite{YUWEI2004}. In addition, $\widetilde{\mu}_{k,i}$ can be
 effectively computed using subgradient algorithm.

\subsubsection{Imperfect CSI}

 For MU MIMO systems with imperfect CSI, the following channel model is adopted
\begin{align}\label{eq135}
 \bm{H}_k = \widehat{\bm{H}}_k + \Delta\bm{H}_k \text{ with } \Delta\bm{H}_k =
  \bm{R}_{{\rm R},k}^{\frac{1}{2}}\bm{H}_{{\rm W},k}\bm{R}_{{\rm T},k}^{\frac{1}{2}} ,
\end{align}
 for $1\le k\le K$, where similar to the previous notations in (\ref{error_model}),
 $\Delta\bm{H}_k$ represents the CSI error for the channel matrix between the $k$th
 user and the BS, $\bm{R}_{{\rm R},k}$ is the BS receive correlation matrix related to
 user $k$, $\bm{R}_{{\rm T},k}$ is the transmit correlation matrix of user $k$, and
 $\bm{H}_{{\rm W},k}$ is a random matrix whose entries are i.i.d. Gaussian variables
 with the distribution $\mathcal{CN}(0,1)$.

 The capacity maximization problem under multiple weighted power constraints, imperfect
 CSI and white noise, i.e., $\bm{R}_{\rm n}=\sigma_n^2\bm{I}$, can be formulated as follows
\begin{align}\label{eq136}
\begin{array}{lcl}
 \textbf{P10}:\!\! &\!\! \min\limits_{\big\{\bm{Q}_k\big\}_{k=1}^K}\!\! &\!\! -\log\left|\bm{I} +
  \bm{K}_k^{-1}\! \sum\limits_{k=1}^K\! \Big(\widehat{\bm{H}}_k \bm{Q}_k \widehat{\bm{H}}_k^{\rm H}\Big) \right|\! , \\
 \!\! &\!\! {\rm{s.t.}} \!\! &\!\! {\rm Tr}\left(\bm{\Omega}_{k,i}\bm{Q}_k\right)\le P_{k,i}, 1\le i\le I_k , \\
 \!\! &\!\! \!\! &\!\! \bm{Q}_k \succeq \bm{0}, 1\le k\le K ,
\end{array}\!\!
\end{align}
 where
\begin{align}\label{eq137}
 \bm{K}_k =& \sigma_n^2 \bm{I} + \sum\limits_{k=1}^K {\rm Tr}\big(\bm{Q}_k \bm{R}_{{\rm T},k}\big)
  \bm{R}_{{\rm R},k} .
\end{align}
 By defining the following positive semi-definite matrices
\begin{align}\label{eq138}
 \bm{\Sigma}_k =& \bm{K}_k + \sum\limits_{j\neq k} \widehat{\bm{H}}_j \bm{Q}_j
  \widehat{\bm{H}}_j^{\rm H} , 1\le k\le K ,
\end{align}
 the full KKT conditions of \textbf{P10} are listed as follows
\begin{align}\label{eq139}
\left\{\begin{array}{l}\!\!\!
 \widehat{\bm{H}}_k^{\rm H} \Big( \bm{\Sigma}_k + \widehat{\bm{H}}_k \bm{Q}_k
  \widehat{\bm{H}}_k^{\rm H} \Big)^{-1}  \widehat{\bm{H}} = \sum_{i=1}^{I_k} \mu_{k,i}
  \bm{\Omega}_{k,i} \\
 \hspace{2mm} + {\rm Tr}\bigg( \Big( \bm{\Sigma}_{k}^{-1} \Big( \bm{\Sigma}_k\! +\!
  \widehat{\bm{H}}_k \bm{Q}_k \widehat{\bm{H}}_k^{\rm H} \Big)^{-1} \Big) \bm{R}_{{\rm R},k}
  \bigg) \bm{R}_{{\rm T},k} \\
 \hspace*{2mm} - \bm{\Psi}_k, ~ 1\le k\le K , \\
 \mu_{k,i} \ge 0 , ~ \mu_{k,i}({\rm Tr}(\bm{\Omega}_{k,i}\bm{Q})-P_{k,i})=0 , \\
 \hspace{2mm} {\rm Tr}(\bm{\Omega}_{k,i}\bm{Q}_k)\le P_{k,i} , ~ 1\le i\le I_k , \\
 {\rm Tr}(\bm{Q}_k\bm{\Psi}_k)=0, ~ \bm{Q}_k\succeq \bm{0}, ~ \bm{\Psi}_k \succeq \bm{0} ,
  ~ 1\le k\le K ,
\end{array}\right. \!\!
\end{align}
 where $\bm{\Psi}_k$ is the Lagrangian multiplier related to $\bm{Q}_k \succeq \bm{0}$.
 Based on Theorem~\ref{T1}, we have the following conclusion.

\begin{conclusion}\label{C10}
 The optimal transmission covariance matrices for \textbf{P10} have the following structure
\begin{align}\label{eq140}
 \bm{Q}_k =& \bm{\Phi}_k^{-\frac{1}{2}} \big[\bm{V}_{\bm{\mathcal{H}}_k}\big]_{:,1:N}
 \left( \mu_k^{-1} \bm{I} - \big[\bm{\Lambda}_{\bm{\mathcal{H}}_k}\big]_{1:N,1:N}^{-2}\right)^{+} \nonumber \\
 & \times \big[\bm{V}_{\bm{\mathcal{H}}_k}\big]_{:,1:N}^{\rm H}\bm{\Phi}_k^{-\frac{1}{2}} ,
 ~ 1\le k\le K ,
\end{align}
 where $\bm{\Phi}_k$ is given by
\begin{align}\label{eq141}
 & \bm{\Phi}_k = \sum\limits_{i=1}^{I_k} \frac{\mu_{k,i}}{\mu_k} \bm{\Omega}_{k,i} \nonumber  \\
 & \hspace*{2mm} +\! \frac{1}{\mu_k}{\rm Tr}\Bigg(\! \bigg(\! \bm{\Sigma}_k^{-1}\! -\! \Big( \bm{\Sigma}_k\! +\!
  \widehat{\bm{H}}_k \bm{Q}_k \widehat{\bm{H}}_k^{\rm H} \Big)^{-1} \bigg)\! \bm{R}_{{\rm R},k}
  \Bigg)\! \bm{R}_{{\rm T},k} ,
\end{align}
 and the unitary matrix $\bm{V}_{\bm{\mathcal{H}}}$ is defined by the SVD
\begin{align}\label{eq142}
 \bm{\Pi}_k^{-\frac{1}{2}} \widehat{\bm{H}}_k \bm{\Phi}_k^{-\frac{1}{2}} =& \bm{U}_{\bm{\mathcal{H}}_k}
  \bm{\Lambda}_{\bm{\mathcal{H}}_k} \bm{V}_{\bm{\mathcal{H}}_k}^{\rm H} \text{ with }
  \bm{\Lambda}_{\bm{\mathcal{H}}} \searrow  ,
\end{align}
 in which
\begin{align}\label{eq143}
 \bm{\Pi}_k\! =& \sigma_n^2\bm{I}\! +\! \sum\limits_{k=1}^K\! {\rm Tr}\big(\bm{Q}_k\bm{R}_{{\rm T},k}\big)
  \bm{R}_{{\rm R},k}\! +\! \sum\limits_{j\neq k}\! \widehat{\bm{H}}_j \bm{Q}_j \widehat{\bm{H}}_j^{\rm H} \! . \!
\end{align}
\end{conclusion}

 When only a sum power constraint for each user is considered and there is no
 spatial correlation  at the transmit antennas, Conclusion~\ref{C10} can be
 simplified as follows.

\vspace*{2mm}

\noindent
\textbf{Conclusion~10.1} \textit{When $I_k=1$, $\bm{\Omega}_{k,1}=\bm{I}$ and
 $\bm{R}_{{\rm T},k}= r_{b_k}\bm{I}$, the optimal solution for \textbf{P10} has the
 following structure}
\begin{align}\label{eq144}
 \bm{Q}_k =& \big[\bm{V}_{\bm{\mathcal{H}}_k}\big]_{:,1:N} \left( \mu_k^{-1} \bm{I} -
  \big[\bm{\Lambda}_{\bm{\mathcal{H}}_k}\big]_{1:N,1:N}^{-2}\right)^{+}
  \big[\bm{V}_{\bm{\mathcal{H}}_k}\big]_{:,1:N}^{\rm{H}} , \nonumber \\
 & 1\le k\le K ,
\end{align}
 \textit{where the unitary matrix $\bm{V}_{\bm{\mathcal{H}}_k}$ is defined by the SVD}
\begin{align}\label{eq145}
 \bm{\Pi}_k^{-\frac{1}{2}} \widehat{\bm{H}}_k =& \bm{U}_{\bm{\mathcal{H}}_k}
  \bm{\Lambda}_{\bm{\mathcal{H}}_k} \bm{V}_{\bm{\mathcal{H}}_k}^{\rm H} \text{ with }
  \bm{\Lambda}_{\bm{\mathcal{H}}_k} \searrow ,
\end{align}
 \textit{and}
\begin{align}\label{eq146}
 \bm{\Pi}_k =& \sigma_n^2\bm{I} + \sum\limits_{k=1}^K P_k r_{b_k} \bm{R}_{{\rm R},k}
  + \sum\limits_{j\neq k} \widehat{\bm{H}}_j \bm{Q}_j \widehat{\bm{H}}_j^{\rm H} .
\end{align}

\noindent \textbf{Computation of $\mu_k$, $\forall k$}: In Conclusion~{10.1}, the
 scalar $\mu_k$ can be computed based on the following equality
\begin{align}\label{eq147}
 {\rm Tr}\left(\bm{Q}_k\right) =& P_{k,1} .
\end{align}
 In this case, the computation of $\mu_k$ is a standard iterative water-filling solution
 \cite{YUWEI2007}.

\subsubsection{Suboptimal solutions for \textbf{P10}}

 For the general case, the solution given by Conclusion~\ref{C10} is not in closed-form.
 From Conclusion~\ref{C10}, it can be seen that the following scalar terms,
\begin{align}
 \alpha =& \sum\limits_{k=1}^{K} {\rm Tr}\big(\bm{Q}_k\bm{R}_{{\rm T},k}\big) ,\label{eq148} \\
 \beta_k\! =& {\rm Tr} \Bigg(\! \bigg(\! \bm{\Sigma}_k^{-1}\! -\! \Big( \bm{\Sigma}_k\! +\! \widehat{\bm{H}}_k
  \bm{Q}_k \widehat{\bm{H}}_k^{\rm H} \Big)^{-1} \bigg) \bm{R}_{{\rm R},k}\! \Bigg)\! ,
  1\le k\le K , \label{eq149}
\end{align}
 are required to compute $\bm{\Phi}_k$ and $\bm{K}_k$. Clearly, $\alpha$ and $\beta_k$
 depend on $\bm{Q}_k$, $1\le k\le K$. When using an iterative algorithm to suboptimally
 compute the solutions of Conclusion~\ref{C10}, therefore, in each iteration, we can use
 the values of $\bm{Q}_k$, $1\le k\le K$, obtained in the previous iteration, to calculate
 $\alpha$ and $\beta_k$. The simulation results show that this iterative algorithm enjoys
 a good performance.

 To avoid the possible slow convergence difficulty, an alternative is to adopt its non-iterative
 approximation. Specifically, we use the maximum of $\alpha$, denoted as $\alpha_{\max}$,
 which is the solution of the following optimization
\begin{align}\label{eq150}
\hspace*{-2mm}\begin{array}{cl}
 \max & \sum_{k=1}^K {\rm Tr}\big(\bm{Q}_k\bm{R}_{{\rm T},k}\big) ,\\
 {\rm{s.t.}} & {\rm Tr}\big(\bm{\Omega}_{k,i}\bm{Q}_k\big)\le P_{k,i} ,
  1\le i\le I_k, 1\le k\le K .
\end{array} \!\!
\end{align}
 Note that given $\alpha$, $\bm{\Sigma}_k$ is a constant, and thus we can replace
 $\beta_k$ with its corresponding maximum, which is given by
\begin{align}\label{eq151}
 \beta_{k,\max} =& {\rm Tr}\big(\bm{\Sigma}_k^{-1}\bm{R}_{{\rm R},k}\big) .
\end{align}

\subsection{MU MIMO Downlink}\label{S5.2}

\subsubsection{Perfect CSI}

 In downlink, the BS sends the user-related information to all the $K$ users.
 For the perfect CSI case and under multiple weighted power constraints, the
 transmission covariance matrices optimization problem for the sum-capacity
 maximization is formulated as
\begin{align}\label{eq152}
\hspace*{-2mm}\begin{array}{lcl}
 \textbf{P11}:\!\! &\!\! \min\limits_{\big\{\bm{Q}_k\big\}_{k=1}^K}\!\! &\!\! \sum\limits_{k=1}^K -\log
  \left| \bm{I} + \bm{\Sigma}_k^{-1} \bm{H}_k \bm{Q}_k \bm{H}_k^{\rm H}\right| , \\
 \!\! &\!\! {\rm{s.t.}} \!\! &\!\! {\rm Tr}\Big(\bm{\Omega}_i \sum_{k=1}^K \bm{Q}_k\Big)\le P_i,
  1\le i\le I , \\
 \!\! &\!\!  \!\! &\!\! \bm{Q}_k \succeq \bm{0} , 1\le k\le K ,
\end{array}\!\!
\end{align}
 where $\bm{\Omega}_i$ is the $i$th constraint's weighting matrix and $P_i$ is
 the corresponding power limit, while
\begin{align}\label{eq153}
 \bm{\Sigma}_k = & \bm{R}_{{\rm n}_k} + \bm{H}_k\sum\nolimits_{j\neq k} \bm{Q}_j \bm{H}_k^{\rm H} ,
\end{align}
 and $\bm{R}_{{\rm n}_k}$ is the noise covariance matrix at user $k$. The full
 KKT conditions of \textbf{P11} are
\begin{align}\label{eq154}
\left\{\begin{array}{l}\!\!\!
 \bm{H}_k^{\rm H}\Big(\bm{\Sigma}_k + \bm{H}_k \bm{Q}_k \bm{H}_k^{\rm H} \Big)^{-1}
  \bm{H}_k = \sum\limits_{i=1}^I \mu_i \bm{\Omega}_i + \sum\limits_{j\neq k} \bm{H}_j^{\rm H} \\
 \hspace{2mm} \times \Big( \bm{\Sigma}_j^{-1}\!\! -\! \Big( \bm{\Sigma}_j\! +\!
  \bm{H}_j \bm{Q}_j \bm{H}_j^{\rm H} \Big)^{-1} \Big) \bm{H}_j\! -\! \bm{\Psi}_k, 1\le k\le K , \\
 \mu_i \ge 0, ~ \mu_i\Big({\rm Tr}\Big(\bm{\Omega}_i\sum\nolimits_{k=1}^K \bm{Q}_k\Big)-P_i\Big)=0 , \\
 \hspace{2mm} {\rm Tr}\Big(\bm{\Omega}_i \sum\nolimits_{k=1}^K \bm{Q}_k\Big)\le P_ i, ~ 1\le i\le I , \\
 {\rm Tr}(\bm{Q}_k\bm{\Psi}_k)=0, \bm{\Psi}_k \succeq \bm{0}, \bm{Q}_k\succeq \bm{0} ,
  ~ 1\le k\le K ,
\end{array}\right. \!\!
\end{align}
 where $\bm{\Psi}_k$ is the Lagrangian multiplier related to $\bm{Q}_k \succeq \bm{0}$.
 Based on Theorem~\ref{T1} and introducing auxiliary variables $\widetilde{\mu}_k$ for
 $1\le k\le K$, the following conclusion is obtained.

\begin{conclusion}\label{C11}
 The optimal transmission covariance matrices for \textbf{P11} satisfy the following structure
\begin{align}\label{eq155}
 \bm{Q}_k =& \bm{\Phi}_k^{-\frac{1}{2}}\big[\bm{V}_{\bm{\mathcal{H}}_k}\big]_{:,1:N}
  \left( \widetilde{\mu}_k^{-1} \bm{I} - \big[\bm{\Lambda}_{\bm{\mathcal{H}}_k}]_{1:N,1:N}^{-2}
  \right)^{+} \nonumber \\
 & \times \big[\bm{V}_{\bm{\mathcal{H}}_k}\big]_{:,1:N}^{\rm H} \bm{\Phi}_k^{-\frac{1}{2}},
  ~ 1\le k\le K ,
\end{align}
 where
\begin{align}\label{eq156}
 \bm{\Phi}_k =& \sum_{i=1}^I\frac{\mu_i}{\widetilde{\mu}_k} \bm{\Omega}_i
  + \frac{1}{\widetilde{\mu}_k} \sum_{j\neq k} \bm{H}_j^{\rm H} \nonumber \\
 & \times \left(\bm{\Sigma}_j^{-1} - \Big( \bm{\Sigma}_j + \bm{H}_j \bm{Q}_j \bm{H}_j^{\rm H}
  \Big)^{-1}\right) \bm{H}_j ,
\end{align}
 and the unitary matrix $\bm{V}_{\bm{\mathcal{H}}_k}$ is defined by the following SVD
\begin{align}\label{eq157}
 \bm{\Sigma}_{k}^{-\frac{1}{2}} \bm{H}_k \bm{\Phi}_k^{-\frac{1}{2}} =&
  \bm{U}_{\bm{\mathcal{H}}_k} \bm{\Lambda}_{\bm{\mathcal{H}}_k} \bm{V}_{\bm{\mathcal{H}}_k}^{\rm H}
  \text{ with } \bm{\Lambda}_{\bm{\mathcal{H}}_k} \searrow .
\end{align}
\end{conclusion}

\noindent \textbf{Computation of $\mu_i$, $\forall i$}: In the MU MIMO downlink
 communications under multiple weighted power constraint and perfect CSI, the
 computation of $\bm{Q}_k$ is performed in an iterative manner. The auxiliary
 variable $\widetilde{\mu}_k$ in (\ref{eq155}) and (\ref{eq156}) cancels output
 each other and, therefore, there is no need to compute it. The remaining variables
 $\mu_i$, $\forall i$, are computed using subgradient algorithms.

\subsubsection{Imperfect CSI}

 When the CSI error is present, a similar model to (\ref{eq135}) is adopted. But in this
 case, the BS has the single transmit correlation matrix $\bm{R}_{\rm T}$ and the $k$th
 user's receive correlation matrix is denoted by $\bm{R}_{{\rm R},k}$. Further assuming
 the white noises at the users' receivers, i.e., $\bm{R}_{{\rm n}_k}=
 \sigma_{{\rm n}_k}^2\bm{I}$, the sum-capacity maximization problem can be written in the
 following form
\begin{align}\label{eq158}
\hspace*{-2mm}\begin{array}{lcl}
 \textbf{P12}:\!\! &\!\! \min\limits_{\big\{\bm{Q}_k\big\}_{k-1}^K}\!\! &\!\! \sum\limits_{k=1}^K -\log
  \left| \bm{I} + \bm{\Sigma}_k^{-1} \widehat{\bm{H}}_k \bm{Q}_k \widehat{\bm{H}}_k^{\rm H}\right| ,  \\
 \!\! &\!\! {\rm{s.t.}} \!\! &\!\! {\rm Tr}\Big(\bm{\Omega}_i \sum_{k=1}^K \bm{Q}_k\Big)\le P_i,
  1\le i\le I , \\
 \!\! &\!\!  \!\! &\!\! \bm{Q}_k \succeq \bm{0}, 1\le k\le K ,
\end{array}\!\!
\end{align}
 where
\begin{align}\label{eq159}
 \bm{\Sigma}_k\! =& \sigma_{{\rm n}_k}^2\! \bm{I}\! +\!  {\rm Tr}\! \left(\sum_{k=1}^K \bm{Q}_k
  \bm{R}_{\rm T}\! \right)\! \bm{R}_{{\rm R},k}\! +\! \widehat{\bm{H}}_k \sum_{j\neq k} \bm{Q}_j
 \widehat{\bm{H}}_k^{\rm H}\! .\!
\end{align}
 The full KKT conditions of \textbf{P12} are derived as
\begin{align}\label{eq160}
\left\{\begin{array}{l}\!\!\!
 \widehat{\bm{H}}_k^{\rm H} \left(\bm{\Sigma}_k + \widehat{\bm{H}}_k \bm{Q}_k
  \widehat{\bm{H}}_k^{\rm H}\right)^{-1} \widehat{\bm{H}}_k = \sum\limits_{i=1}^I \mu_i \bm{\Omega}_i \\
 \hspace{2mm} + \sum\limits_{k=1}^K {\rm Tr} \bigg( \Big( \bm{\Sigma}_k^{-1} - \Big( \bm{\Sigma}_k
  + \widehat{\bm{H}}_k \bm{Q} _k \widehat{\bm{H}}_k^{\rm H} \Big)^{-1} \Big) \bm{R}_{{\rm R},k}
  \bigg) \bm{R}_{\rm T} \\
 \hspace{2mm} + \sum\limits_{j\neq k} \widehat{\bm{H}}_j^{\rm H} \bigg( \bm{\Sigma}_j^{-1}
  - \Big( \bm{\Sigma}_j + \widehat{\bm{H}}_j \bm{Q}_j \widehat{\bm{H}}_j^{\rm H} \Big)^{-1}
  \bigg) \widehat{\bm{H}}_j - \bm{\Psi}_k ,  \\
 \mu_i \ge 0, ~ \mu_i \bigg( {\rm Tr}\Big( \bm{\Omega}_i \sum\limits_{k=1}^K \bm{Q}_k\Big)-P_i\bigg) = 0 , \\
 \hspace{2mm} {\rm Tr}\Big( \bm{\Omega}_i \sum\limits_{k=1}^K \bm{Q}_k\Big) \le P_i , ~ 1\le i\le I , \\
 {\rm Tr}\big(\bm{Q}_k\bm{\Psi}_k\big)=0, ~ \bm{\Psi}_k \succeq \bm{0}, ~ \bm{Q}_k\succeq \bm{0} ,
  1\le k\le K .
\end{array}\right. \!\!
\end{align}
 Based on Theorem~\ref{T1}, the following conclusion is obtained.

\begin{conclusion}\label{C12}
 The optimal transmission covariance matrices for \textbf{P12} satisfies the following
 structure
\begin{align}\label{eq161}
 \bm{Q} _k=& \bm{\Phi}_k^{-\frac{1}{2}} \big[\bm{V}_{\bm{\mathcal{H}}_k}\big]_{:,1:N}
  \left( \widetilde{\mu}_k^{-1} \bm{I} - \big[\bm{\Lambda}_{\bm{\mathcal{H}}_k}\big]_{1:N,1:N}^{-2}
  \right)^{+} \nonumber \\
 & \times \big[\bm{V}_{\bm{\mathcal{H}}_k}\big]_{:,1:N}^{\rm H} \bm{\Phi}_k^{-\frac{1}{2}} ,
  ~ 1\le k\le K ,
\end{align}
 where $\bm{\Phi}_k$ is defined by
\begin{align}\label{eq162}
 \bm{\Phi}_k =& \sum\limits_{i=1}^I \frac{\mu_i}{\widetilde{\mu}_k} \bm{\Omega}_i  \nonumber \\
 & \hspace*{-2mm}+\! \frac{1}{\widetilde{\mu}_k} \sum\limits_{k=1}^K {\rm Tr}
  \bigg(\! \Big( \bm{\Sigma}_k^{-1}\! -\!
  \Big( \bm{\Sigma}_k\! +\! \widehat{\bm{H}}_k \bm{Q}_k \widehat{\bm{H}}_k^{\rm H} \Big)^{-1}\!
  \bm{R}_{{\rm R},k} \! \bigg)\! \bm{R}_{\rm T} \nonumber \\
 & \hspace*{-2mm}+\! \frac{1}{\widetilde{\mu}_k}\! \sum\limits_{j\neq k}\! \widehat{\bm{H}}_j^{\rm H}
  \bigg(\! \bm{\Sigma}_j^{-1}\! -\! \Big( \bm{\Sigma}_j\! +\! \widehat{\bm{H}}_j \bm{Q}_j
  \widehat{\bm{H}}_j^{\rm H} \Big)^{-1}\bigg) \widehat{\bm{H}}_j ,
\end{align}
 and the unitary matrix $\bm{V}_{\bm{\mathcal{H}}_k}$ is defined by the SVD
\begin{align}\label{eq163}
 \bm{\Sigma}_k^{-\frac{1}{2}} \widehat{\bm{H}}_k \bm{\Phi}_k^{-\frac{1}{2}} =&
  \bm{U}_{\bm{\mathcal{H}}_k} \bm{\Lambda}_{\bm{\mathcal{H}}_k} \bm{V}_{\bm{\mathcal{H}}_k}^{\rm H}
  \text{ with } \bm{\Lambda}_{\bm{\mathcal{H}}_k} \searrow .
\end{align}
\end{conclusion}

\subsubsection{Suboptimal Solution for \textbf{P12}}

 In the general case, the solution given by Conclusion~\ref{C12} is not in
 closed-form, and thus fixed point methods must be used to numerically compute
 the solution from the complicated solution form given in Conclusion~\ref{C12}.
 From Conclusion~\ref{C12}, it can be seen that the main difficulty comes from
 the following two scalar terms in $\bm{\Sigma}_k$ and $\bm{\Phi}_k$, denoted as
 $\alpha$ and $\beta$, respectively,
\begin{align}
 \alpha =& {\rm Tr}\left( \sum\limits_{k=1}^K \bm{Q}_k \bm{R}_{\rm T} \right) , \label{eq164} \\
 \beta =& \sum\limits_{k=1}^K {\rm Tr}\left(\! \bigg( \bm{\Sigma}_k^{-1}\! -\!
  \Big( \bm{\Sigma}_k\! +\! \widehat{\bm{H}}_k \bm{Q}_k \widehat{\bm{H}}_k^{\rm H} \Big)^{-1}
  \bigg) \bm{R}_{{\rm R},k} \right) \! , \label{eq165}
\end{align}
 which depends on $\bm{Q}_k$, $\forall k$. An iterative algorithm can be used to
 compute the solution of Conclusion~\ref{C12}. Specifically, in each iteration
 the values of $\bm{Q}_k$, $\forall k$, obtained in the previous iteration are used
 to calculate $\alpha$ and $\beta$. The simulation results show that this algorithm
 enjoys a good performance.

 Alternatively, a non-iterative approximation can be adopted. Note that the maximum
 of $\alpha$, $\alpha_{\max}$, is the solution of the following optimization
\begin{align}\label{eq166}
\begin{array}{cl}
 \max & {\rm Tr}\left(\sum\limits_{k=1}^K \bm{Q}_k \bm{R}_{\rm T}\right) , \\
 {\rm{s.t.}} & {\rm Tr}\left( \bm{\Omega}_i\sum\limits_{k=1}^K \bm{Q}_k\right)\le P_i ,
 1\le i\le I .
\end{array}
\end{align}
 For given $\alpha$, $\bm{\Sigma}_k$ becomes a constant. For any given $\bm{\Sigma}_k$,
 the maximum of $\beta$ is
\begin{align}\label{eq167}
 \beta_{\max} =& \sum\limits_{k=1}^K {\rm Tr}\big(\bm{\Sigma}_k^{-1}\bm{R}_{{\rm R},k}\big) .
\end{align}
 By using $\alpha_{\max}$ and $\beta_{\max}$ to replace $\alpha$ and  $\beta$, we
 readily obtain a suboptimal solution of \textbf{P12}.

\begin{remark}\label{R3}
 The optimization problems \textbf{P9} and \textbf{P11} were studied in
 \cite{GeneralizedWF} and \cite{YUWEI2007}, respectively. But our proposed method
 is more straightforward and furthermore it avoids the turning-off effect. The
 optimization problems \textbf{P10} and \textbf{P12} have not been found in the
 existing open literature.

 Based on our results, it can be discovered that multiple weighted power constraints
 have the same impact as CSI errors in MU MIMO uplink. It is also revealed that in
 the special case of MU MIMO uplink with the sum power constraint and only the
 receive antenna correlation, the optimal solution is in closed-form as given in
 Conclusion~{10.1}. The general solutions given in Conclusions~\ref{C9} to \ref{C12}
 require iterative numerical algorithms to solve them.

 For MU MIMO optimizations, the sum-MSE minimization is not considered in this work
 because: 1)~the  structure of the optimal solution for MU-MIMO uplink cannot be
 derived directly based on its KKT conditions, and 2)~the optimization solution has
 different structure from that for the capacity maximization. Naturally, the
 sum-MSE minimization problem can be handled via numerical convex optimization
 methods, such as semi-definite programming (SDP), second order cone programming
 (SOCP), etc. \cite{XingIET}. But such solutions provide little insights and are
 not closely related to the main theme of this paper. Interested readers are
 referred to \cite{XingIET} and the references therein.
\end{remark}

\begin{figure*}[bp!]\setcounter{equation}{170}
\vspace*{-5mm}
\hrule
\begin{align}\label{eq171}
\left\{\begin{array}{l}\!\!\!
 \widehat{\bm{H}}_l^{\rm H}\! \bigg(\! \bm{\Sigma}_r\! +\!\!\!\! \sum\limits_{i\in \psi_{{\rm D},r}}\!\!\!\!
  \widehat{\bm{H}}_i \bm{Q}_{v(i)} \widehat{\bm{H}}_i^{\rm H}\! \bigg)^{-1}\!\!
  \widehat{\bm{H}}_l\! =\!\! \sum\limits_{m=1}^{I_{S(l)}} \mu_{S(l),m} \bm{\Omega}_{S(l),m}\!\! +\!\!\!\!
  \sum\limits_{i\in \pi_{v(l)},m=D(i)}\!\!\! {\rm Tr} \Bigg(\!\! \Bigg(\! \bm{\Sigma}_m^{-1}\! -\!\!
  \bigg(\! \bm{\Sigma}_m\!\! +\!\!\! \sum\limits_{j\in \psi_{{\rm D},m}}\!\!\! \widehat{\bm{H}}_j
  \bm{Q}_{v(j)} \widehat{\bm{H}}_j^{\rm H}\! \bigg)^{-1} \Bigg)\!
  \bm{R}_{{\rm R},i}\! \Bigg)\! \bm{R}_{{\rm T},i} \\
 \hspace{2mm} + \sum\limits_{i\in \pi_{v(l)}, i\neq l, m=D(i)} \widehat{\bm{H}}_i^{\rm H}
  \Bigg( \bm{\Sigma}_m^{-1} - \bigg( \bm{\Sigma}_m + \sum\limits_{j\in \psi_{{\rm D},m}}
  \widehat{\bm{H}}_j \bm{Q}_{v(j)} \widehat{\bm{H}}_j^{\rm H} \bigg)^{-1} \Bigg)
  \widehat{\bm{H}}_i - \bm{\Psi}_l ,  \\
 \mu_{S(l),m} \ge 0, \ \mu_{S(l),m}\Bigg( {\rm Tr} \bigg( \bm{\Omega}_{S(l),m}\sum\limits_{i\in \xi_{S(l)}}
  \bm{Q}_i\bigg)- P_{S(l),m}\Bigg)=0 ,  ~  {\rm Tr} \bigg( \bm{\Omega}_{S(l),m}\sum\limits_{i\in \xi_{S(l)}}
  \bm{Q}_i\bigg) \le P_{S(l),m} , ~ 1\le m\le I_{S(l)} , \\
 \hspace{2mm} {\rm Tr}\big(\bm{Q}_{v(l)} \bm{\Psi}_l\big)=0, ~ \bm{\Psi}_l \succeq \bm{0}, ~
   \bm{Q}_{v(l)} \succeq \bm{0},
\end{array}\right. \!\!
\end{align}
\vspace*{-1mm}
\end{figure*}

\section{Transmission Covariance Matrix Optimization for MIMO Networks}\label{S6}

 In our MIMO network, multiple sources communicate with multiple destinations
 and every node is equipped with multiple antennas. Each source node can send
 distinct information to several destinations simultaneously, and every
 destination can receive distinct signals from several sources. To formulate
 the transmission covariance matrix optimization for this very general MIMO
 network, we use the links between the sources and the destinations as basic
 elements in the system modeling. Because wireless communications are discussed,
 for any source node and destination node, there always exists a link. For each
 channel link there is at most one desired signal. When the destination of the
 link does not want to communicate with the source of the link, the signal
 transmitted by the source of the link becomes the interference to the
 destination through the link. Moreover, CSI error has to be taken into account.

 The network links' parameters are first defined. Denote the index set of
 destination nodes by $\mathcal{D}$ and the index set of source nodes by
 $\mathcal{S}$, respectively. Moreover, the index set of links is denoted by
 $\mathcal{L}$. Let $\psi_{{\rm D},r}$ be the link set whose destination node
 is node $r$. In other words, a desired signal to node $r$ exists on each link
 of $\psi_{{\rm D},r}$. The symbol $v(l)$ represents the index for the desired
 signal vector on link $l$, whose covariance matrix is denoted as $\bm{Q}_{v(l)}$.
 If there is no desired signal, $v(l)={\rm null}$ and $\bm{Q}_{\rm null}=
 \bm{0}$. The symbol $\psi_l$ denotes the index set for all the signal
 transmitted in the $l$th link and $\xi_s$ is the index set of all the signal
 transmitted from the $s$th source. Note that if the $s$th node is the source
 node of the $l$th link, $\psi_l=\xi_s$. A signal on a given link cannot be both
 desired signal and interference simultaneously. Similar to the channel models
 of the previous sections, for the $l$th link, the channel matrix is given by\setcounter{equation}{167}
\begin{align}\label{eq168}
 \bm{H}_l =& \widehat{\bm{H}}_l + \bm{R}_{{\rm R},l}^{\frac{1}{2}}
  \bm{H}_{{\rm W},l} \bm{R}_{{\rm T},l}^{\frac{1}{2}},
\end{align}
 where $\widehat{\bm{H}}_l$ is the estimated channel matrix and
 $\bm{R}_{{\rm R},l}^{\frac{1}{2}}\bm{H}_{{\rm W},l}\bm{R}_{{\rm T},l}^{\frac{1}{2}}$
 is the corresponding channel estimation error. Also again the elements of
 $\bm{H}_{{\rm W},l}$ are i.i.d. Gaussian distributed, obeying ${\cal CN}(0,1)$,
 the positive definite matrices $ \bm{R}_{{\rm R},l}$ and $\bm{R}_{{\rm T},l}$ are the
 receive and transmit correlation matrices, respectively.

\subsection{Capacity Maximization for MIMO Network}\label{S6.1}

 The capacity maximization for the MIMO network with imperfect CSI and under multiple
 weighted power constraints can be formulated as
\begin{align}\label{eq169}
\hspace*{-2mm}\begin{array}{lcl}
 \textbf{P13}:\!\!\! &\!\!\! \max\limits_{ \{\bm{Q}_i, i\in\xi_s,\forall s\in\mathcal{S}\}}\!\!\! &\!\!\!
  \sum\limits_{\forall r\in\mathcal{D}} \log \left| \bm{I}\! +\!
  \bm{\Sigma}_r^{-1}\!\!\! \sum\limits_{l\in \psi_{{\rm D},r}}\!\!\! \Big( \widehat{\bm{H}}_l
  \bm{Q}_{v(l)} \widehat{\bm{H}}_l^{\rm H} \Big) \right|\! ,\\
 \!\!\! &\hspace*{-5mm}{\rm{s.t.}}\!\!\! &\hspace*{-10mm}{\rm Tr}\! \left(\! \bm{\Omega}_{s,m}\!\!\! \sum\limits_{i\in \xi_{s}}\!\!\!
  \bm{Q}_i\! \right)\! \le P_{s,m}, 1\le m\le I_s , \forall s\in\mathcal{S}, \! \\
 \!\!\! &\!\!\! \!\!\! &\hspace*{-10mm} \bm{Q}_i \succeq \bm{0}, \forall i\in\xi_s , \forall s\in\mathcal{S},
\end{array}\!\!
\end{align}
 where $\bm{\Omega}_{l,m}$ is the $m$th weighting matrix of the source node of link
 $l$ and $P_{l,m}$ is the corresponding power threshold, while
\begin{align}\label{eq170}
 \bm{\Sigma}_r =& \bm{R}_{\bm{n}_r} + \sum\limits_{i\in \psi_{{\rm D},r}} \left(
  \widehat{\bm{H}}_i \sum\limits_{j\in \psi_i, j\neq v(i)} \bm{Q}_j \widehat{\bm{H}}_i^{\rm H} \right) \nonumber \\
 & + \sum\limits_{i\in \psi_{{\rm D},r}} {\rm Tr}\left(
  \sum\limits_{j\in \psi_i} \bm{Q}_j \bm{R}_{{\rm T},i} \right) \bm{R}_{{\rm R},i} ,
\end{align}
 and $\bm{R}_{\bm{n}_r}$ is the noise covariance matrix of the link with destination
 node $r$. In this optimization problem, the covariance matrix of the desired signal
 vector on each link is the optimization variable.

 The corresponding KKT conditions for the covariance matrices of the desired signals
 on links $l$, where $l \in \psi_{{\rm D},r}$, $\forall r\in\mathcal{D}$, are given
 in (\ref{eq171}). In (\ref{eq171}), $D(l)$ and $S(l)$ denote the destination node
 and source node of link $l$, respectively. Based on these definitions, we have
 $r=D(l)$ in the KKT conditions (\ref{eq171}), and the link set $\pi_{v(l)}$ consists
 of the links on which the $v(l)$-th signal vector (whose covariance matrix is
 $\bm{Q}_{v(l)}$) is transmitted. For the optimization of $\big\{\bm{Q}_{v(l)}\big\}$,
 the following conclusion can be obtained based on Theorem~\ref{T1}. \setcounter{equation}{171}

\begin{conclusion}\label{C13}
 The optimal transmission covariance matrices $\big\{\bm{Q}_{v(l)}\big\}$ of \textbf{P13}
 have the following structure
\begin{align}\label{eq172}
 \bm{Q}_{v(l)} =& \bm{\Phi}_l^{-\frac{1}{2}}\big[\bm{V}_{\bm{\mathcal{H}}_l}\big]_{:,1:N}
  \Big( \mu_l^{-1} \bm{I} - \big[\bm{\Lambda}_{\bm{\mathcal{H}}_l}\big]_{1:N,1:N}^{-2}\Big)^{+} \nonumber \\
 & \times \big[\bm{V}_{\bm{\mathcal{H}}_l}\big]_{:,1:N}^{\rm H} \bm{\Phi}_l^{-\frac{1}{2}},
\end{align}
 where $\bm{\Phi}_l$ is defined as
\begin{align}\label{eq173}
 & \bm{\Phi}_l\! =\! \frac{1}{\mu_l}\! \sum\limits_{m=1}^{I_{S(l)}}\! \mu_{S(l),m} \bm{\Omega}_{S(l),m}
  \! +\! \frac{1}{\mu_l} \!\! \sum\limits_{i\in \pi_{v(l)},m=D(i)} \!\!\!\!\! {\rm Tr}
  \Bigg(\! \bm{\Sigma}_m^{-1}\! -\! \bigg(\! \bm{\Sigma}_m \nonumber \\
 & \hspace{2mm} +\! \sum\limits_{j\in \psi_{{\rm{D}},m}}\!\!\! \widehat{\bm{H}}_j
  \bm{Q}_{v(j)} \widehat{\bm{H}}_j^{\rm H}\bigg)^{-1}\! \bm{R}_{{\rm R},i} \Bigg)
  \bm{R}_{{\rm T},i}  +\!\!\!\! \sum\limits_{ i\in \pi_{v(l)},i\neq l,m=D(i)}\!\!\!\! \widehat{\bm{H}}_i^{\rm H} \nonumber \\
 & \hspace{2mm} \times \Bigg(\!
  \bm{\Sigma}_m^{-1}\! -\! \bigg(\! \bm{\Sigma}_m\! +\!\!\! \sum\limits_{j\in \psi_{{\rm{D}},m}}\!\!\!
  \widehat{\bm{H}}_j \bm{Q}_{v(j)} \widehat{\bm{H}}_j^{\rm H} \bigg)^{-1} \Bigg)\!
  \widehat{\bm{H}}_i\! -\! \bm{\Psi}_l , \!
\end{align}
 and the unitary matrix $\bm{V}_{\bm{\mathcal{H}}_l}$ is defined by the following SVD
\begin{align}\label{eq174}
 \bm{\Pi}_l^{-\frac{1}{2}} \widehat{\bm{H}}_l \bm{\Phi}_l^{-\frac{1}{2}} =& \bm{U}_{\bm{\mathcal{H}}_l}
  \bm{\Lambda}_{\bm{\mathcal{H}}_l} \bm{V}_{\bm{\mathcal{H}}_l}^{\rm H} \text{ with }
  \bm{\Lambda}_{\bm{\mathcal{H}}_l} \searrow  ,
\end{align}
 in which
\begin{align}\label{eq175}
 \bm{\Pi}_l =& \bm{\Sigma}_l + \sum\limits_{i\in \psi_{{\rm{D}},r},i\neq l}
  \widehat{\bm{H}}_i \bm{Q}_{v(i)} \widehat{\bm{H}}_i^{\rm H} .
\end{align}
\end{conclusion}

 Due to the nonconvex nature of the covariance matrix optimization for MIMO networks, the
 solution given by Conclusion~\ref{C13} is not a closed-form solution. In fact, even for
 the special case with perfect CSI, this solution is not a closed-form solution either.
 Therefore, the solution of  Conclusion~\ref{C13} has to be computed using some numerical
 algorithms such as fixed point algorithms. Comparing the solutions of \textbf{P13} and
 \textbf{P1}, we discover that despite the complicated relationships between the
 covariance matrices at different nodes in \textbf{P13}, the roles of the covariance
 matrices and the correlation matrices of channel errors are very simple. Specifically,
 some involved matrix variables work just as parts of the equivalent noise covariance
 matrix, but the others contribute to the weights on transmit covariance matrices.

\subsection{Suboptimal Solutions for \textbf{P13}}\label{S6.2}

 From Conclusion~\ref{C13}, it can be seen that the main difficulty in computing the
 solution of \textbf{P13} comes from the following scalar terms in $\bm{\Sigma}_r$
 and $\bm{\Phi}_l$, denoted as $\alpha_r$ and $\beta_{l,i}$, respectively,
\begin{align}
 & \alpha_r = \sum_{i\in \psi_{{\rm D},r}} {\rm Tr}\bigg( \sum_{j\in \psi_i} \bm{Q}_j
  \bm{R}_{{\rm T},i}\bigg) \bm{R}_{{\rm R},i} , \label{eq176} \\
 & \beta_{l,i} = {\rm Tr} \Bigg( \bm{\Sigma}_m^{-1} - \bigg( \bm{\Sigma}_m\! +\! \sum_{j\in \psi_{{\rm D},m}}
  \Big(\widehat{\bm{H}}_j \bm{Q}_{v(j)} \widehat{\bm{H}}_j^{\rm H}\Big) \bigg)^{-1}
  \! \bm{R}_{{\rm R},i}\Bigg) \! , \nonumber \\
 & \hspace*{10mm} i \in \pi_{v(l)}, m=D(i) . \label{eq177}
\end{align}
 By using the fixed point algorithm, an iterative algorithm can be adopted
 to compute the solutions of Conclusion~\ref{C13}. Specifically, in each iteration,
 $\bm{Q}_j$ in $\alpha_r$ and $\beta_{l,i}$ can be replaced by their values in the
 previous iteration. The simulation results show that this approach enjoys good
 performance, but this method suffers from some drawbacks, including the need to select
 good initial $\bm{Q}_j$ and potentially slow convergence.

 Alternatively, an non-iterative approximation can be employed. Note that the
 maximum of $\alpha_r$ is the solution of the following optimization
\begin{align}\label{eq178}
\begin{array}{cl}
 \max & \sum\limits_{i\in \psi_{{\rm D},r}} {\rm Tr}\bigg(\sum\limits_{j\in \psi_i} \bm{Q}_j
  \bm{R}_{{\rm T},i}\bigg) \bm{R}_{{\rm R},i} , \\
 {\rm{s.t.}} & {\rm Tr}\bigg(\bm{\Omega}_{s,m} \sum\limits_{i\in \xi_s} \bm{Q}_i\bigg) \le P_{s,m},
  1\le m\le I_s , \\
 & s=S(i), \forall i \in \psi_{{\rm D},r} .
\end{array}
\end{align}
 Given the solution of (\ref{eq178}), $\alpha_{\max ,r}$, $\bm{\Sigma}_l$ becomes a
 constant. For any given $\bm{\Sigma}_l$, an upper bound of $\beta_{l,i}$  is
\begin{align}\label{eq179}
 \beta_{\max ,l,i} =& {\rm Tr}\big(\bm{\Sigma}_m^{-1}\bm{R}_{{\rm R},i}\big), i\in \pi_{v(l)},m=D(i) .
\end{align}
 A suboptimal solution of Conclusion~\ref{C13} can readily be obtained by substituting
 $\alpha_r$ and $\beta_{l,i}$ with $\alpha_{\max ,r}$ and $\beta_{\max ,l,i}$

 The optimization problem \textbf{P13} is formulated according to network links where
 the channel estimation errors are taken into account. This formulation is much clearer
 than the work presented in \cite{PoliteWF_2}. Moreover, Theorem~\ref{T1} is readily
 applicable to \textbf{P13}. To our best knowledge, the transmission covariance matrix
 optimization for MIMO networks with imperfect CSI is regarding as an open problem in
 the existing literature. Similar to the weighted minimum-mean-squared-error (WMMSE)
 algorithm \cite{Christensen2008}, which is an iterative optimization algorithm, our
 proposed solution also has a very wide range of applications, and can be used as a
 benchmark algorithm for the transceiver optimization for many complicated MIMO systems.

\begin{remark}\label{R4}
 Centralized optimization requires full CSI at the processing center. For networks
 consisting of a large number of nodes, timely CSI sharing among all nodes is very
 challenging and may even be impractical. Thus distributed optimization algorithms
 are preferred, which require local or limited CSI only. Nowadays, great efforts
 have been devoted to distributed designs. The key challenge in distributed designs
 is how to decompose the overall optimization problem into a series of low-complexity
 subproblems to be handled at different nodes with limited information sharing among
 these subproblems. A common practice is to scrutinize the centralized optimization
 problem or its solution structure to find a viable decomposition. In our work,
 the structures of the optimal solutions for many MIMO design problems are derived,
 which may facilitate distributed algorithm designs. But as this paper is on the
 centralized MIMO designs based on KKT conditions, distributed algorithms are
 beyond the scope of this paper.
\end{remark}

\section{Summaries}\label{S7}

 In this section, we summarize this paper, by clarifying our distinct contributions for
 each of the problems considered in previous sections and explaining the connections
 among these problems. First of all, the major contribution of this paper is the
 proposal of a unified framework to derive the optimal structure of the transmission
 covariance matrix/matrices based on the KKT conditions. Since the KKT conditions are
 necessary conditions for the optimal solutions, the structure derived from the KKT
 conditions must hold at the optimal solutions, i.e., the optimal solutions must have
 the structure derived from the KKT conditions. This is applicable to convex
 optimization problems and nonconvex ones alike. In the following, our specific
 contributions of the aforementioned optimization problems are recapped.

 For the capacity maximization problem \textbf{P1} and the MSE minimization problem
 \textbf{P2} for SU MIMO systems with perfect CSI and a weighted sum power constraint,
 the optimal solutions have been derived in the literature \cite{Telatar1999,Sampth01}.
 But our proposed method is different in the following three aspects. First, we solve
 these two different optimization problems using the same framework and case-by-case
 studies are avoided. Second, the proposed proofs are simpler and more rigorous. In
 particular, the derivations in \cite{Telatar1999} take two steps -- deriving the
 diagonalizable structure based on matrix inequality and then deriving water-filling
 solution based on KKT conditions. By comparison, our work only requires one
 straightforward step to derive the optimal solutions from the KKT conditions. As
 pointed out in \cite{Xing2016}, the work of \cite{Sampth01} suffers from the
 turning-off effect, while our work successfully resolve this problem. Thirdly, the
 optimal solutions given in Conclusions~\ref{C1} and \ref{C2} are both in closed-form.

 For the capacity maximization problem \textbf{P3} and the MSE minimization problem
 \textbf{P4} for SU MIMO systems under perfect CSI and multiple weighted power
 constraints, our distinct contribution is also three-fold. First, although
 \textbf{P3} has been studied in \cite{Mai2011}, our proposed method is simpler, and
 it overcomes the turning-off effect as well as removes the restriction that the
 channel matrix must be column or row full rank. Second, for \textbf{P4}, no solution
 has been discovered before. Based on our proposed unified framework, the optimal
 solution structure is derived for this problem for the first time. Third, our work
 reveals that the effect of multiple weighted constraints are naturally equivalent
 to a weighted power constraint by comparing the solutions for \textbf{P3} and
 \textbf{P4} with those for \textbf{P1} and \textbf{P2}. Again, the optimal solutions
 given in Conclusions~\ref{C3} and \ref{C4} are both in closed-form.

 For the capacity maximization \textbf{P5} and the MSE minimization \textbf{P6}
 for SU MIMO systems with imperfect CSI under a sum power constraint, there exist
 some studies in the literature \cite{XingTSP2013,XingTSP201501,Ding09,Ding10} but
 they are for the special cases with only transmit correlation or receive correlation.
 This work considers the general cases with both transmit and receive correlation
 for the both problems, where the structures of the optimal solutions are derived
 in Conclusions~\ref{C5} and \ref{C6}, respectively. The solutions of
 Conclusions~\ref{C5} and \ref{C6} are entirely new and they are not found in the
 existing literature. It should be noted that that both \textbf{P5} and
 \textbf{P6} are nonconvex and, therefore, the KKT conditions are only necessary
 conditions for the optimality. For the special cases with only transmit
 correlation or receive correlation, our derived optimal solutions are in
 closed-form as given by Conclusions~5.1, 5.2, 6.1, and 6.2. Compared with our
 previous work \cite{XingTSP201501}, the derivation in this paper is largely
 different and considerably simpler. Compared with the works \cite{Ding09,Ding10},
 our results do not suffer from the turning-off effect.

 The more complicated capacity maximization \textbf{P7} and MSE minimization
 \textbf{P8} for MU MIMO systems under imperfect CSI and multiple weighted power
 constraints, to our best knowledge, are largely open problems in the existing
 literature. The structures of the optimal solutions for \textbf{P7} and
 \textbf{P8} given in Conclusions~\ref{C7} and \ref{C8} are entirely new. Due to
 the complicated nature of these two problems, the results are not in
 closed-form, and iterative schemes such as fixed point algorithm are required
 to compute the solutions. However, we discover that in the special case with
 only transmit correlation, the optimal solutions can be derived in closed-form
 as given in Conclusions~7.1 and 8.1. These cases have important practical
 significance as they are transformable to the MIMO training design problems
 \cite{Ding09}.

 For MU MIMO systems, \textbf{P9} and \textbf{P11} define the sum-capacity
 maximization problems under perfect CSI and with multiple weighted power
 constraints for the uplink and the downlink, respectively. Their counterparts
 with imperfect CSI are presented in \textbf{P10} and \textbf{P12}. Based on
 the proposed unified framework, the optimal solution structures are derived
 for these problems and given in Conclusions~\ref{C9} to \ref{C12}. As mentioned
 in Remark~\ref{R3}, the problems \textbf{P9} and \textbf{P11} were studied in
 \cite{GeneralizedWF} and \cite{YUWEI2007}, respectively. But our proposed method
 is more straightforward and furthermore it avoids the turning-off effect. The
 problems \textbf{P10} and \textbf{P12} have not been found in the existing open
 literature. The general solutions given in Conclusions~\ref{C9} to \ref{C12}
 require iterative numerical algorithms to solve them. By examining the results,
 it can be discovered that multiple weighted power constraints have the same
 impact as CSI errors in MU MIMO uplink. It is also revealed that in the special
 case of MU MIMO uplink with sum power constraint and only receive correlation,
 the optimal solution can be derived in closed-form as given in Conclusion~10.1.

 Finally, the general MIMO network is investigated in which there are multiple
 communication links and on each link there are multiple sources communicate
 with multiple destinations. The capacity maximization problem under imperfect
 CSI is presented in \textbf{P13}, which is regarded as an open problem in the
 existing literature. The optimal structure for the solution of \textbf{P13} is
 given in Conclusion~\ref{C13}, and it can be solved iteratively. Although the
 MIMO network topology considered is very complicated and general, with the
 aid of our unified framework given in Theorem~\ref{T1}, we are able to provide
 the optimal structure of the solution for \textbf{P13} in a simple and clear
 derivation. Our result shows the impacts of different system components and can
 facilitate the transceiver optimization for general MIMO networks.

\section{Numerical Results}\label{S8}

 The importance of the water-filling structure in transmission covariance matrix
 design is demonstrated by numerical results. Two representative designs are chosen.
 All the results are obtained by averaging over 500 independent channel realizations.

\begin{figure}[!bp]
\vspace*{-6mm}
\begin{center}
\includegraphics[width=\columnwidth]{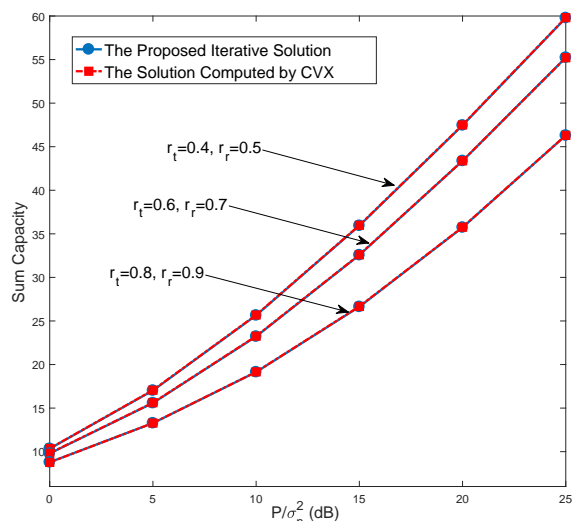}
\end{center}
\vspace*{-6mm}
\caption{Performance comparisons between the proposed iterative solution and the
 solution computed by the CVX for MU MIMO uplink.}
\label{Fig_2}
\end{figure}

\subsection{Simulation Results on MU MIMO Uplink}\label{S8.1}

 We consider the MU-MIMO uplink, where two users each with 4 antennas communicate
 with the BS equipped with 8 antennas. The signal-to-noise ratio (SNR) of the $k$th
 user is defined as $P_k/\sigma_n^2$, where $P_k$ is the sum transmit power across
 all the transmit antennas of user $k$ and $\sigma_n^2$ is the noise power at each
 receive antenna of the BS. Without loss of generality, in our simulation, the same
 SNR value is used for the both users, i.e., $P_1/\sigma_n^2=P_2/\sigma_n^2=P/\sigma_n^2$.
 Per-antenna power constraint is considered for each user, where the power limit
 for the four antennas are set to 1.6,  1.2,  0.8 and 0.4, respectively. In general,
 the power ratios between different antennas can be arbitrarily chosen. Here we set
 the power limits to be significantly different so that the difference from the
 case with a sum power constraint is sufficiently large. The widely used Kronecker
 correlation model is adopted, i.e., the $(i,j)$-th elements of $\bm{R}_{\rm R}$ and
 $\bm{R}_{{\rm T},k}$ are given respectively by $\big[\bm{R}_{\rm R}\big]_{i,j}=
 r_r^{|i-j|}$ and $\big[\bm{R}_{{\rm T},k}\big]_{i,j}=r_{t_k}^{|i-j|}$. Note that
 in uplink, the receive antenna correlation matrix $\bm{R}_{\rm R}$ is the same for
 different users. In the simulation, we also set $r_{t_1}$ and $r_{t_2}$ to the same
 value of $r_t$. For this MU MIMO uplink with the perfect CSI, the capacity maximization
 under per-antenna power constraint belongs to the case of \textbf{P9}, and the optimal
 transmission covariance matrices can be obtained using an iterative procedure based
 on Conclusion~\ref{C9}. Traditionally, this capacity maximization maximization is
 solved using numerical algorithms, such as the CVX package \cite{tool}. In
 Fig.~\ref{Fig_2}, the sum-capacity of the proposed solution is compared with
 that of the solution obtained using the CVX. As expected, both solutions achieve
 the same optimal performance. Our proposed iterative solution has the additional
 advantages that it is more straightforward and much simpler.

\subsection{Simulation Results on MIMO Networks}\label{S8.2}

 A MIMO network with four nodes, including two sources and two destinations, is
 considered. All the nodes are equipped with 4 antennas. Both source nodes communicate
 with both destinations. At each source node, the per-antenna power constraints are
 set to $2\times 1.6$, $2\times1.2$, $2\times 0.8$ and $2\times 0.4$ (Since each
 source simultaneously communicates with two destinations, the power limit is doubled).
 The estimated channel matrix is generated according to $\widehat{\bm{H}}_k=
 \bm{R}_{{\rm R},k}^{\frac{1}{2}}\widehat{\bm{H}}_{{\rm W},k}\bm{R}_{{\rm T},k}^{\frac{1}{2}}$
 \cite{XingTSP2013}. As discussed in \cite{XingTSP2013}, this model represents some
 practical channel estimators, and it is equivalent to the channel model (\ref{eq168}).
 Specifically, the elements of $\widehat{\bm{H}}_{{\rm W},k}$ and $\bm{H}_{{\rm W},k}$
 are i.i.d. Gaussian random variables. The variance of every element of $\bm{H}_{{\rm W},k}$
 is set to $\sigma_e^2$ and that of $\widehat{\bm{H}}_{{\rm W},k}$ is $1-\sigma_e^2$. Thus,
 the elements of $\bm{H}_k$ all have unit variance. The capacity maximization for MIMO
 networks with imperfect CSI is presented in \textbf{P13} and the optimal solution is
 provided in Conclusion~\ref{C13}. Naturally, our solution is equally applicable to the
 perfect CSI case by setting the channel error to zero, i.e., $\sigma_e^2=0$.

\begin{figure}[!b]
\vspace*{-6mm}
\begin{center}
\includegraphics[width=\columnwidth]{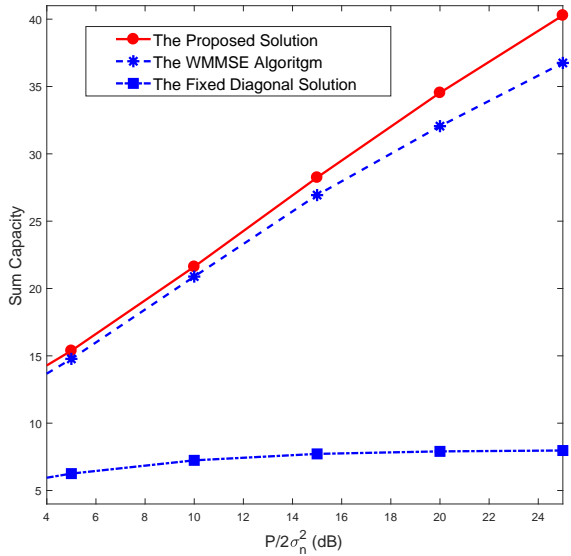}
\end{center}
\vspace*{-7mm}
\caption{Performance comparison between the proposed solution, the fixed
 diagonal solution, and the WMMSE solution for MIMO network with perfect CSI
 given $r_t=0.4$ and $r_r=0.5$.}
\label{Fig_3}
\end{figure}

 In Fig.~\ref{Fig_3}, the sum-capacity achieved by our proposed solution is shown for
 the case of perfect CSI, where two existing solutions are also shown as the benchmarks.
 For the first benchmark solution, the two signal covariance matrices at the two source
 nodes are set to the diagonal matrices that satisfy the respective per-antenna power
 constraints. The other benchmark algorithm is the WMMSE solution, which is detailed
 in Appendix~\ref{Apa}. Not surprisingly, the fixed diagonal solution has the worst
 performance, as it does not consider the mutual interference between different signals.
 The WMMSE solution performs  much better than the fixed diagonal solution as it optimizes
 the mutual interference mitigation. It can be seen that our proposed solution attains the
 best performance.

 For the case of imperfect CSI, the sum-capacity performance of our proposed solution
 as given in Conclusion~\ref{C13} is shown in Fig.~\ref{Fig_4}, where it can be observed
 that the proposed solution has better performance than its non-robust counterpart that
 takes the estimated CSI as the true CSI.

\begin{figure}[!t]
\vspace*{-1mm}
\begin{center}
\includegraphics[width=\columnwidth]{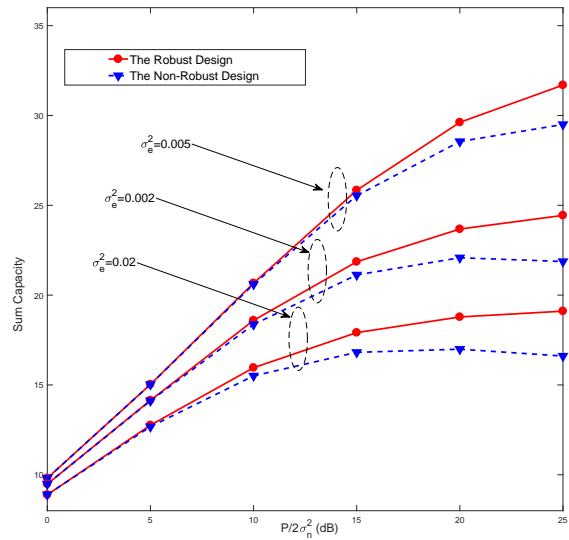}
\end{center}
\vspace*{-5mm}
\caption{The performance comparisons between the robust solution and non-robust solution
 for MIMO network with imperfect CSI given $r_t=0.4$ and $r_r=0.5$.}
\label{Fig_4}
\vspace*{-6mm}
\end{figure}

\section{Conclusions}\label{S9}

 For MIMO systems, the key task of many transceiver optimization problems is to optimize
 the covariance matrices of the transmitted signals. In this paper, a general unified
 framework has been proposed for deriving the water-filling structures of the optimal
 covariance matrices based on KKT conditions. From the general solution for the unified
 framework, interesting and important underlying relationships among the solutions of
 different MIMO optimization designs have been revealed, which help us to understand the
 related existing works much better. Furthermore, the unified framework and its general
 solution have been applied to a wide range of applications, including complicated MIMO
 networks with multiple communication links and systems with imperfect CSI, to discover
 the new solutions that do not exist in the open literature or outperform the existing
 solutions. For example, it has been shown that our proposed solution attains higher
 capacity for MIMO transceiver optimization in complicated network settings than the
 WMMSE algorithm.

\appendix

\subsection{The WMMSE Algorithm}\label{Apa}

 To our best knowledge, transforming the sum-rate or sum-capacity maximization into
 a WMMSE problem was firstly proposed in \cite{Christensen2008}, and extensions by
 other researchers followed it. In \cite{Shi2011}, the authors extended the WMMSE
 algorithm to MIMO interference broadcasting systems. In \cite{XingIET}, the WMMSE
 algorithm was extended to more general networks, e.g., multi-hop cooperative
 networks with multiple sources, relays, and destinations, where all nodes are
 equipped with multiple antennas and each source node communicates with all
 destination nodes. The channel information was assumed to be imperfect with
 Gaussian errors.

 A brief description to the WMMSE algorithm used in our simulation is given below.
 We would like to highlight that the following algorithm is more general than the
 one in \cite{Shi2011} and it is applicable to distributed MIMO networks with
 channel estimation errors. Naturally, it can directly be applied to the perfect
 CSI case by the setting the channel error to zero.

 The transceiver model for link $r$ is specified by
\begin{align}\label{ap1} 
 \bm{y}_r =& \bm{H}_r \bm{F}_r \bm{s}_r + \bm{v}_r ,
\end{align}
 where $\bm{y}_r$ is the received signal vector at the destination of link $r$,
 and the vector $\bm{s}_r$ contains all the desired signals to the destination
 by stacking all the signal vectors of the sources. Without loss of generality,
 the covariance matrix of $\bm{s}_r$ is assumed to be the identity matrix.
 Furthermore, $\bm{H}_r$ comprises all the channel matrices over which the signals
 are transmitted to the destination, i.e.,
\begin{align}\label{ap2} 
 \bm{H}_r =& \left[\bar{\bm{H}}_{i_1} ~ \bar{\bm{H}}_{i_2}\cdots\bar{\bm{H}}_{i_{N_r}}\right] , ~
  i_n \in \psi_{{\rm D},r} ,
\end{align}
 where $N_r$ is the cardinality of the set $\psi_{{\rm D},r}$. Additionally, $\bm{F}_r$
 is a block diagonal matrix, defined by
\begin{align}\label{ap3} 
 \bm{F}_r =& {\rm{diagB}}\left\{\bm{F}_{r,1},\bm{F}_{r,2},\cdots ,\bm{F}_{r,N_r}\right\} ,
\end{align}
 whose $n$th sub-matrix $\bm{F}_{r,n}$ is the corresponding precoding matrix for the
 channel $\bar{\bm{H}}_{i_n}$. Finally, $\bm{v}_r$ is the composite vector of the
 additive white noises and the interference signals at the destination of link $r$,
 whose covariance matrix is given by
\begin{align}\label{ap4} 
 \bm{\Sigma}_r =& \bm{R}_{\bm{n}_r} + \sum_{i\in \psi_{{\rm I},r}}
  \left(\bar{\bm{H}}_i\sum_{j\in \phi_i} \bm{Q}_j \bar{\bm{H}}_i^{\rm H} \right) \nonumber \\
 & + \sum_{i\in \psi_{{\rm D},r}\cup \psi_{{\rm I},r}} {\rm Tr}\left(
  \sum_{j\in i\cup \phi_i} \bm{Q}_j \bm{R}_{{\rm T},i}\right) \bm{R}_{{\rm R},i}.
\end{align}
 For the perfect CSI case, the third term in the righthand side of (\ref{ap4}) becomes zero.

 Based on the above transceiver model of link $r$, the weighted MSE can be expressed as
\begin{align}\label{WMMSE_App} 
 \sum_r {\rm Tr}\left( \bm{W}_r \mathbb{E}\Big\{ \big( \bm{G}_r \bm{y}_r - \bm{s}_r\big)
  \big(\bm{G}_r \bm{y}_r - \bm{s}_r\big)^{\rm H}\Big\}\right) ,
\end{align}
 where $\bm{G}_r$ is the receive beamforming matrix and the positive semi-definite matrix
 $\bm{W}_r$ is the weighting matrix for link $r$. Based on the weighted MSE criterion
 (\ref{WMMSE_App}), the WMMSE optimization problem can be formulated as \cite{Xing2016}
\begin{align}\label{ap6} 
\begin{array}{cl}
 \min\limits_{\bm{F}_r,\bm{G}_r,\bm{W}_r}\!\!\! &\!\!\! \sum\limits_r\! {\rm Tr}\! \left(\!
  \bm{W}_r \mathbb{E}\Big\{ \big( \bm{G}_r \bm{y}_r\! -\! \bm{s}_r\big)
  \big(\bm{G}_r \bm{y}_r\! -\! \bm{s}_r\big)^{\rm H}\Big\}\! \right)\! \\
 \!\!\! &\!\!\! - \sum\limits_r \log\left|\bm{W}_r\right|, \\
 {\rm{s.t.}}\!\!\! &\!\!\! \text{multiple weighted power constraints} .
\end{array}\!\!
\end{align}
 To find an optimal solution of (\ref{ap6}), the alternative optimization can be
 used where one set of the matrices are optimized with the other two sets fixed.

 It is obvious that given $\bm{G}_r$ matrices and $\bm{W}_r$ matrices, the
 optimization (\ref{ap6}) with weighted transmit power constraints is a standard
 quadratic matrix programming problem \cite{XingIET}. Thus the optimal $\bm{F}_r$
 matrices given $\bm{G}_r$ and $\bm{W}_r$ can be efficiently solved by standard
 SOCP, SDP or other convex optimization algorithms. When $\bm{F}_r$ and
 $\bm{W}_r$ are fixed, the optimal $\bm{G}_r$ matrices are exactly the linear
 minimum MSE equalizers, which are readily derived from the complex matrix
 derivative of the weighted MSE criterion (\ref{WMMSE_App}). Finally, for given
 $\bm{G}_r$ and $\bm{F}_r$, the optimal $\bm{W}_r$ matrices are readily given by
\begin{align}\label{ap7} 
 \bm{W}_r =& \left( \mathbb{E}\Big\{ \big( \bm{G}_r \bm{y}_r - \bm{s}_r\big)
  \big(\bm{G}_r \bm{y}_r - \bm{s}_r\big)^{\rm H}\Big\}\right)^{-1} .
\end{align}
 Thus the corresponding WMMSE optimization problem becomes the sum capacity maximization
 \cite{Xing2016}. In a nutshell, the WMMSE algorithm is an iterative alternative
 optimization algorithm with a guaranteed convergence.

\end{document}